\pdfoutput=1
\documentclass[
	bigMargin = false,
	font = palatino,
	final
	]{MyClass}

\usepackage{myMathHeader}
\usepackage{autonum}

\let\ConnSpace\relax
\let\GauGroup\relax
\let\GauAlgebra\relax
\MyNewMathOperator{\ConnSpace}		{command={\mathbrush{C}}, sort={C}, display={$\ConnSpace(P)$}, description={Space of connections of the principal bundle $P$}}
\MyNewMathOperator{\GauGroup}			{command={\mathbrush{G}}, sort={G}, display={$\GauGroup(P)$}, description={Group of gauge transformation of the (principal) bundle $P$}}
\MyNewMathOperator{\GauAlgebra}		{command={\mathrm{L}\mathbrush{G}}, sort={LG}, display={$\GauAlgebra(P)$}, description={Lie algebra of infinitesimal gauge transformation of the principal bundle $P$}}

\RenewDocumentCommand{\wedgeDual}{s m m}{#2 \mathbin{\dot{\wedge}} #3}

\Title{Singular symplectic cotangent bundle reduction of gauge field theory}

\Abstract{
	We prove a theorem on singular symplectic cotangent bundle reduction in the Fréchet setting and apply it to Yang--Mills--Higgs theory with special emphasis on the Higgs sector of the Glashow--Weinberg--Salam model.
	For the latter model we give a detailed description of the reduced phase space and show that the singular structure is encoded in a finite-dimensional Lie group action.
}



\Institution{itp}{
	Institut für Theoretische Physik, 
	Universität Leipzig, 
	04009 Leipzig, Germany	}
\Institution{mpiitp}{
	Max Planck Institute for Mathematics in the Sciences,
	04103 Leipzig, Germany
	and
	Institut für Theoretische Physik, 
	Universität Leipzig, 
	04009 Leipzig, Germany	}

\Author[
	affiliation = mpiitp,
	email = tobias.diez@mis.mpg.de,
	]{diez}{Tobias Diez}
\Author[
	affiliation = itp,
	email = {rudolph@rz.uni-leipzig.de}
	]{rudolph}{Gerd Rudolph}

\begin{document}

\MakeTitle

\tableofcontents

\listoftodos

\section{Introduction}
While perturbative methods in quantum field theory yield a satisfactory description of high energy processes in particle physics, many low energy processes are dominated by non-perturbative effects and, so far, have eluded from a rigorous theoretical explanation.
Many attempts to develop a mathematically rigorous non-perturbative quantum gauge theory are inspired by the general program of constructive quantum field theory.
In its classical version, this program is based on the Euclidean path integral concept and on lattice approximation as an intermediate step, see \parencite{Hooft2005,JaffeWitten} for status reports on the case of Yang--Mills theory.
A different option is provided by the Hamiltonian approach, see \parencite{KogutSusskind1975,BalianDrouffeItzykson1974} and \parencite{KijowskiRudolph2002,KijowskiRudolph2005} for the theory on a finite lattice and \parencite{GrundlingRudolph2017} for the construction of the thermodynamic limit.
Concerning the Hamiltonian viewpoint, one may think of another approach: instead of using the lattice approximation as an intermediate step, study classical continuum gauge theory as an infinite-dimensional Hamiltonian system with a symmetry and, next, try to extend methods from geometric quantization to the infinite-dimensional setting.
The fundamental problem in dealing with Yang--Mills theory is the elimination of the unphysical gauge degrees of freedom.
Within the Hamiltonian picture, this is accomplished via symplectic reduction.
\emph{In this paper, we develop the theory of symplectic reduction in an infinite-dimensional setting and we apply it to gauge symmetry reduction of the Yang--Mills--Higgs system.}

In the context of classical particle mechanics, the system is described by a finite-dimensional symplectic manifold \( (M, \omega) \) and the symmetry is encoded in the action of a Lie group \( G \) on $M$.
With the help of the momentum map \( J: M \to \LieA{g}^* \), which captures the conserved quantities, one passes to the reduced phase space \( M \sslash_\mu G \defeq J^{-1}(\mu) \slash G \) at \( \mu \in \LieA{g}^* \).
Often the action of \( G \) is not free.
Then, the reduced space is not a smooth manifold but a stratified space with each stratum being a symplectic manifold.
This process of passing to the reduced phase space is called singular Marsden--Weinstein symplectic reduction \parencite{MarsdenWeinstein1974,SjamaarLerman1991}.
Starting in the early nineties \parencite{EmmrichRoemer1990}, various case studies have shown that the singularities may have an influence on the properties of the quantum theory, see \parencite[Chapter~8]{RudolphSchmidt2014} for a detailed discussion.
In particular, they may carry information about the spectrum of the Hamiltonian, see \parencite{HuebschmannRudolphEtAl2009}.

In most of the applications in physics, the phase space is a cotangent bundle \( \CotBundle Q \) over the configuration space $Q$ of the system. 
Symplectic reduction for that case has attracted much attention in recent years.
In particular, it is of interest to connect the geometry of the symplectically reduced space \( \CotBundle Q \sslash_\mu G \) to the properties of the quotient \( Q \slash G \).
In finite dimensions and for proper free $G$-actions this connection is well understood: the reduction \( \CotBundle Q \sslash_0 G \) at zero is symplectomorphic to the cotangent bundle \( \CotBundle (Q \slash G) \) (with its canonical symplectic form) of the reduced configuration space \( Q \slash G \).
For reduction at non-zero momentum \( \mu \in \LieA{g}^* \), the reduced space is symplectomorphic to the fibered product of \( \CotBundle (Q \slash G) \) and the coadjoint orbit \( G \cdot \mu \) endowed with a magnetic symplectic form.
See, for example, \parencite[Section~2]{MarsdenMisiolekEtAl2007} for details.
The singular case is more complicated due to the occurrence of new phenomena that are absent in the regular case.
In finite dimensions, \textcite{PerlmutterRodriguez-OlmosSousa-Dias2007} have shown that the fibered structure of the cotangent bundle yields a refinement of the usual orbit-momentum type strata into so-called seams.
The principal seam is symplectomorphic to a cotangent bundle while the singular seams are coisotropic submanifolds of the corresponding symplectic stratum.
As in the context of standard symplectic reduction, the proof that the seams are manifolds follows from an appropriate symplectic slice theorem, see \parencite{Schmah2007,Rodriguez-OlmosTeixido-Roman2014}.
The latter theorem provides a normal form both for the symplectic structure and for the momentum map, and additionally the local symplectomorphism bringing the system to normal form is adapted to the fiber structure of \( \CotBundle Q \).

\emph{In the first part of the paper, we extend the above results concerning singular symplectic reduction of cotangent bundles to the case of infinite-dimensional Fréchet manifolds endowed with the action of an infinite-dimensional Fréchet Lie group.}
Infinite-dimensional geometry is usually developed in the functional analytic setting of Banach spaces.
However, many examples of interest do not allow for a satisfactory description in terms of Banach spaces.
For example, the group of diffeomorphisms of a given Sobolev regularity fails to be a Banach Lie group, because composition is not a smooth map.
Note, however, that the group of smooth diffeomorphisms is a Fréchet Lie group.
In order to include these important examples, we assume that all manifolds we are dealing with are Fréchet.
We stress that, in the Fréchet context, the notion of a cotangent bundle needs special care, see \cref{Setting,sec:cotangentBundle:definition}.
We also note that it is impossible to directly extend the approach of \parencite{Schmah2007,PerlmutterRodriguez-OlmosSousa-Dias2007} to the Fréchet setting according to the following serious obstacles: one needs the existence of (orthogonal) complements and an appropriate version of regular symplectic reduction to construct the normal form reference system; moreover, the inverse function theorem is central to the construction of the local symplectomorphism bringing the system to normal form.
All these tools are not readily available for Fréchet manifolds.

In sharp contrast to the strategy in \parencite{Schmah2007,Rodriguez-OlmosTeixido-Roman2014} outlined above, at the root of our approach lies the observation that, for the symplectic reduction procedure, it is not essential to bring the symplectic structure into a normal form.
Instead, the focus lies on the momentum map only.
Within this strategy, we construct a normal form for momentum maps of lifted actions (see \cref{prop:cotangentBundle:simpleNormalForm}) and we prove a singular cotangent bundle reduction theorem including the analysis of the secondary stratification into seams, see \cref{prop:cotangentBundle:singularCotangentBundleRed}.
The main technical tool is a version of the slice theorem for Fréchet $G$-manifolds as proved in \parencite{DiezSlice}.
Since cotangent bundle reduction at non-zero momentum values is still not completely understood even in finite dimensions, we restrict our attention to reduction at zero, which is fortunately the situation that we encounter for Yang--Mills--Higgs theory.
We also discuss dynamics in this context.
The possible significance of the seams for the dynamics is demonstrated by analyzing the finite-dimensional example of the harmonic oscillator, see \cref{ex:cotangentBundle:harmonicOss}.
Clearly, as a by-product, our theory provides a much simpler approach in the finite-dimensional setting. 

\emph{In the second part of the paper, we apply our general theory to the singular cotangent bundle reduction of Yang--Mills--Higgs theory.}
Concerning the classical configuration space, its stratified structure has been studied already before, see \parencite{KondrackiRogulski1986} for the general theory and \parencite{RudolphSchmidtEtAl2002b,RudolphSchmidtEtAl2002a,RudolphSchmidtEtAl2002,HertschRudolphSchmidt2010,HertschRudolphSchmidt2011} for the classification of the orbit types for all classical groups.
We use the Hamiltonian picture for this model as developed in \parencite{DiezRudolphClebsch}, \cf also \parencite{Sniatycki1999}.
First, we check that the model meets the assumptions made in the general theory showing that the singular cotangent bundle reduction theorem holds here.
This implies that the reduced phase space of the theory is a stratified symplectic space\footnote{See, however, \cref{sec:yangMillsHiggs} for the problems related to the frontier condition.}.
Next, we analyze the normal form and the stratification in some detail. 
In particular, we find that including the singular strata leads to a refinement of what is called the resolution of the Gauß constraint in the physics literature.
Our analysis of the normal form improves upon earlier work \parencite{ArmsMarsdenEtAl1981,Arms1981} on the singular geometry of the momentum map level set.
We also describe the orbit types for the model after symmetry breaking, which leads to a more transparent picture of the stratification.

Finally, we further analyze the secondary stratification in the concrete example of the Higgs sector of the Glashow--Weinberg--Salam model.
In this context, we find that the configuration space has only two orbit types.
The singular stratum is characterized by the remarkable physical property that the \( W \)-bosons are absent, \ie the Z-boson is the only non-trivial intermediate vector boson on the singular stratum.
We then pass to the discussion of the stratification of the phase space.
The secondary stratification turns out to be similar to that of the harmonic oscillator in the sense that there are only three strata: two cotangent bundles, which are glued together by one seam.
The non-generic cotangent bundle is the phase space of the sub-theory consisting of electrodynamics described by a photon, the theory of a massive vector boson described by the \( Z \)-boson and the theory of a self-interacting real scalar field described by the Higgs boson.
The seam is characterized by the condition that the \( W \)-boson field vanishes but its conjugate momentum is non-trivial.
In contrast, on the generic cotangent bundle all intermediate vector bosons of the model are present.
Finally, we study the structure of the strata of the reduced phase space in terms of gauge invariant quantities for the theory on \( S^3 \).
By implementing unitary and the Coulomb gauge fixing in a geometric fashion using momentum maps we show that the singular structure of the reduced phase space is encoded in a finite-dimensional \( \UGroup(1) \)-Lie group action.

\paragraph*{Acknowledgments}
We are very much indebted to M.~Schmidt and J.~Huebschmann for reading the manuscript and for many helpful discussions.
We gratefully acknowledge support of the Max Planck Institute for Mathematics in the Sciences in Leipzig and of the University of Leipzig.

%%%%%%%%%%%%%%%%%%%%%%%%%%%%%%%%%%%%%%%%%%%%%%%%%%%%%%%%%%%%%%%%%%%%%%%%%%%%%%%%%%%%%%%%

\section{Singular cotangent bundle reduction}
\label{sec:cotangentBundleReduction}

%%%%%%%%%%%%%%%%%%%%%%%%%%%%%%%%%%%%%%%%%%%%%%%%%%%%%%%%%%%%%%%%%%%%%%%%%%%%%%%%%%%%%%%%

\subsection{Problems in infinite dimensions. The setting}
\label{Setting}
Let \( Q \) be a Fréchet manifold (we refer the reader to \cref{Frechet} for conventions and further references concerning the calculus of infinite-dimensional manifolds).
The tangent bundle \( \TBundle Q \) of \( Q \) is a smooth manifold in such a way that the projection \( \TBundle Q \to Q \) is a smooth locally trivial bundle.
However, the topological dual bundle \( \TBundle' Q \defeq \bigDisjUnion_{q \in Q} (\TBundle_q Q)' \) is not a \emph{smooth} fiber bundle for non-Banach manifolds \( Q \), \cf \parencite[Remark~I.3.9]{Neeb2006} and \cref{sec:cotangentBundle:definition}.
As a substitute, we say that a smooth Fréchet bundle \( \CotBundleProj: \CotBundle Q \to Q \) is a cotangent bundle if there exists a fiberwise non-degenerate pairing with \( \TBundle Q \), see \cref{sec:cotangentBundle:definition} for details.
We note that according to this definition the cotangent bundle is no longer canonically associated to \( Q \) but requires the choice of a bundle \( \CotBundle Q \) and of a pairing \( \CotBundle Q \times_Q \TBundle Q \to \R \).
Similarly to the finite-dimensional case, for a given cotangent bundle \( \CotBundleProj: \CotBundle Q \to Q \), the formula
\begin{equation}
	\theta_p (X) = \dualPair{p}{\tangent_p \CotBundleProj (X)}, \qquad X \in \TBundle_p (\CotBundle Q) \, ,
\end{equation}
defines a smooth \( 1 \)-form \( \theta \) on \( \CotBundle Q \).
Furthermore, \( \omega = \dif \theta \) is a symplectic form.

Next, assume that a Fréchet Lie group \( G \) acts smoothly on \( Q \), that is, assume that the action map \( G \times Q \to Q \) be smooth.
In the sequel, by a Fréchet $G$-manifold \( Q \) we mean a Fréchet manifold $Q$ endowed with a smooth action of a Fréchet Lie group $G$.
We will often write the action using the dot notation as \( (g, q) \mapsto g \cdot q \).
Similarly, the induced action of the Lie algebra \( \LieA{g} \) of \( G \) is denoted by \( \xi \ldot q \in \TBundle_q Q \) for \( \xi \in \LieA{g} \) and \( q \in Q \).
Clearly, \( q \mapsto \xi \ldot q \) is the fundamental vector field generated by \( \xi \). 
Throughout the paper, we assume that the \( G \)-action is proper, that is, inverse images of compact subsets under the map
\begin{equation}
	G \times Q \to Q \times Q, \qquad (g, q) \mapsto (g \cdot q, q)
\end{equation}
are compact. 

Recall that in the finite-dimensional setting, for a proper action, a slice always exists \parencite{Palais1961}.
In the infinite-dimensional case, this may no longer be true and additional hypotheses have to be made; see \parencite{DiezSlice,Subramaniam1986} for general slice theorems in infinite dimensions and \parencite{AbbatiCirelliEtAl1989,Ebin1970,CerveraMascaroEtAl1991} for constructions of slices in concrete examples.
Having in mind the application to Yang--Mills theory, it is especially important to note that, in particular, the action of the group of gauge transformations on the space of connections admits a slice.
We refer the reader to \cref{sec:calculus:groupActionsSlices} for more details and background information regarding slices.
Here, we are not going to spell out the additional conditions of the general slice theorem.
We will simply assume that a slice exists at every point.
Properness of the action and the existence of slices have important consequences for the orbit type stratification of \( Q \).
We refer the reader to \parencite{DiezSlice} for the proof of the following statements in infinite dimensions.
First, recall that the conjugacy class of a stabilizer subgroup \( G_q \) is called the orbit type\footnote{In fact, the orbit type is constant along the orbit by the equivariance property \( G_{g \cdot q} = g G_q g^{-1} \) for every \( g \in G \).} of \( q \).
We put a preorder on the set of orbit types by declaring \( (H) \leq (K) \) for two orbit types represented by the stabilizer subgroups \( H \) and \( K \) if there exists \( a \in G \) such that \( a H a^{-1} \subseteq K \).
Since the action is proper, this preorder is actually a partial ordering.
Moreover, since the action admits a slice at every point, the subset \( Q_{(H)} \) of orbit type \( (H) \) is a locally closed submanifold of \( Q \).
Similarly, the quotient \( \check{Q}_{(H)} = Q_{(H)} \slash G \) carries a smooth manifold structure such that the natural projection is a smooth submersion.
If the orbit type decomposition of \( Q \) satisfies the frontier condition, then the decomposition of \( \check{Q} = Q \slash G \) into orbit types \( \check{Q}_{(H)} \) is a stratification (\cf \cref{def:stratification:stratification}).

By linearization, we also get a smooth action of \( G \) on the tangent bundle \( \TBundle Q \), which we write using the lower dot notation as \( g \ldot Y \in \TBundle_{g \cdot q} Q \) for \( g \in G \) and \( Y \in \TBundle_q Q \).
The action on \( \TBundle Q \) induces a \( G \)-action on \( \CotBundle Q \) by requiring that the pairing be left invariant, that is,
\begin{equation}
 	\dualPair{g \cdot p}{Y}
 		= \dualPair{p}{g^{-1} \ldot Y}
 		\quad p \in \CotBundle_{g^{-1} \cdot q} Q, Y \in \TBundle_q Q. 
\end{equation}
In order that this equation defines a smooth action on \( \CotBundle Q \), the action \( \TBundle_{g^{-1} \cdot q} Q \to \TBundle_q Q \) needs to be weakly continuous with respect to the pairing \( \dualPairDot \) for every \( g \in G \).
In finite dimensions, there always exists a \( G \)-equivariant diffeomorphism between \( \TBundle Q \) and \( \CotBundle Q \) and, therefore, the orbit types of \( \TBundle Q \) and \( \CotBundle Q \) coincide.
In the infinite-dimensional setting, a vector space may not be isomorphic to its dual and thus the orbit types may differ.
Indeed, we now give an example where the action on the dual space has more orbit types.
\begin{example}
	\label{ex:contangentBundle:differentOrbitTypesOnCotangent}
	Consider the space \( \SequenceSpace^1 \) of real-valued doubly infinite sequences \( (x_n)_{n \in \Z} \) satisfying
	\begin{equation}
		\norm{(x_n)}_1 \defeq \sum_{n \in \Z} \abs{x_n} < \infty.
	\end{equation}
	With respect to the norm \( \normDot_1 \), \( \SequenceSpace^1 \) is a Banach space.
	The topological dual is isomorphic to the space \( \SequenceSpace^\infty \) of all bounded sequences \( (\alpha_n)_{n \in \Z} \), which is also a Banach space with respect to the uniform norm
	\begin{equation}
	 	\norm{(\alpha_n)}_\infty \defeq \sup_{n \in \Z} \abs{\alpha_n}.
	 \end{equation}
	 The pairing between \( \SequenceSpace^1 \) and \( \SequenceSpace^\infty \) is given by
	 \begin{equation}
	 	\dualPair{(\alpha_n)}{(x_n)} \defeq \sum_{n \in \Z} \alpha_n x_n, \qquad (\alpha_n) \in \SequenceSpace^\infty, (x_n) \in \SequenceSpace^1.
	 \end{equation}
	 The group \( G = \Z \) of integers acts on \( \SequenceSpace^1 \) via the shift operators \( T^j: \SequenceSpace^1 \to \SequenceSpace^1 \), with \( j \in \Z \), defined by
	 \begin{equation}
	 	T^j (x_n) \defeq (x_{n+j}) \, .
	 \end{equation}
	 This action is linear and self-adjoint in the sense that the dual action on \( \SequenceSpace^\infty \) is also given by the shift operators \( T^j: \SequenceSpace^\infty \to \SequenceSpace^\infty \).
	 The actions on \( \SequenceSpace^1 \) and \( \SequenceSpace^\infty \) are continuous for the discrete topology on \( \Z \).
	 Recall that all subgroups of \( \Z \) are of the form \( k \Z \) for some integer \( k \geq 0 \).
	 Note that a sequence \( (x_n) \) has stabilizer \( k \Z \) if and only if it is \( k \)-periodic.
	 The subgroups \( \set{0} \) (\( k =0 \)) and \( \Z \) (\( k=1 \)) can easily be realized as stabilizer subgroups of some \( (x_n) \in \SequenceSpace^1 \).
	 However, none of the other subgroups \( k \Z \) with \( k > 1 \) occurs as the stabilizer subgroup of the action on \( \SequenceSpace^1 \).
	 Indeed, if a sequence \( (x_n) \) is non-zero and \( k \)-periodic for \( k > 1 \), then it is divergent and hence never an element of \( \SequenceSpace^1 \).
	 On the other hand, every periodic sequence is bounded and thus all subgroups \( k \Z \) for \( k \geq 0 \) occur as stabilizers of the action on the dual space \( \SequenceSpace^\infty \).
\end{example}
In order to exclude such pathological phenomena that make it impossible to connect the orbit types of \( Q \) with the ones for the lifted action on \( \CotBundle Q \), in the sequel, we will assume that there exist a \( G \)-equivariant vector bundle isomorphism between \( \TBundle Q \) and \( \CotBundle Q \); an assumption that holds in the applications we are interested in.
Note that for a Hilbert manifold \( Q \), a \( G \)-invariant scalar product yields such a \( G \)-equivariant diffeomorphism. 

By \cref{prop:contangentBundle:existenceMomentumMap}, the lifted action of \( G \) on the cotangent bundle \( \CotBundle Q \) preserves the canonical symplectic form \( \omega \).
In order to define the momentum map, we need to specify a dual space to the Lie algebra \( \LieA{g} \) of \( G \).
Similarly to the above strategy for the cotangent bundle, choose a Fréchet space \( \LieA{g}^* \) and a separately continuous non-degenerate bilinear form \( \kappa: \LieA{g}^* \times \LieA{g} \to \R \).
With respect to these data, the momentum map \( J: \CotBundle Q \to \LieA{g}^* \), if it exists, satisfies
\begin{equation}
 	\kappa(J(p), \xi) = \dualPair{p}{\xi \ldot q}, 
\end{equation}
for \( p \in \CotBundle_q Q \) and \( \xi \in \LieA{g} \), see \cref{prop:contangentBundle:existenceMomentumMap}.
Note that the right-hand side, viewed as a functional on \( \LieA{g} \), may not be representable by an element \( J(p) \in \LieA{g}^* \).
In this case, the momentum map for the lifted action does not exist.
The existence of the momentum map mainly depends on the chosen duality \( \kappa \) and not so much on the cotangent bundle.
In fact, we now give an example for an action of an infinite-dimensional Lie group acting on a \emph{finite-dimensional} cotangent bundle that does not possess a momentum map.
\begin{example}
	\label{ex:contangentBundle:actionWithoutMomentumMap}
 	Let \( M \) be a compact finite-dimensional manifold endowed with a volume form \( \vol \).
 	Consider a finite-dimensional Lie group \( G \) that acts on a finite-dimensional manifold \( Q \).
 	Fix a point \( \star \in M \) and let the current group \( \sFunctionSpace(M, G) \) act on \( Q \) via the evaluation group homomorphism \( \ev_\star: \sFunctionSpace(M, G) \to G \), that is,
 	\begin{equation}
 		\sFunctionSpace(M, G) \times Q \to Q, \quad (\lambda, q) \mapsto \lambda(\star) \cdot q \, .
 	\end{equation}
 	The induced action of \( \sFunctionSpace(M, G) \) on the cotangent bundle \( \CotBundle Q \) is symplectic but it  does not possess a momentum map with respect to the natural dual pairing
 	\begin{equation}
 		\kappa: \sFunctionSpace(M, \LieA{g}^*) \times \sFunctionSpace(M, \LieA{g}) \ni (\mu, \xi) \mapsto \int_M \dualPair{\mu}{\xi} \vol \in \R.
 	\end{equation}
 	If the momentum map \( J \) were to exist, then it would necessarily satisfy the relation
 	\begin{equation}
 		\int_M \dualPair{J(p)}{\xi} \vol = \kappa(J(p), \xi) = \dualPair{p}{\xi(\star) \ldot q} = \dualPair{J_G(p)}{\xi(\star)}
 	\end{equation}
 	for all \( p \in \CotBundle Q \) and \( \xi \in \sFunctionSpace(M, \LieA{g}) \), where \( J_G: \CotBundle Q \to \LieA{g}^* \) denotes the momentum map for the \( G \)-action.
 	This is only possible if \( J(p) \) is a delta distribution localized at the point \( \star \).
 	However, by definition, the dual \( \sFunctionSpace(M, \LieA{g}^*) \) only contains regular distributions, and thus there exists no momentum map with values in \( \sFunctionSpace(M, \LieA{g}^*) \).
 	Note that a momentum map does exist if the dual space of \( \sFunctionSpace(M, \LieA{g}) \) is chosen in such a way that it contains the delta distributions.
	
	From a different point of view, the reason for the non-existence of the momentum map is that the evaluation map \( \ev_\star: \sFunctionSpace(M, \LieA{g}) \to \LieA{g} \) is \emph{not} weakly continuous with respect to the pairings \( \kappa \) and \( \dualPairDot \), and hence does not have an adjoint \( \LieA{g}^* \to \sFunctionSpace(M, \LieA{g}^*) \).
\end{example}
In the sequel, we always assume that the lifted action on \( \CotBundle Q \) has a momentum map.

Finally, there is another basic problem typical for the infinite-dimensional setting.
As the map \( \ev_\star: \sFunctionSpace(M, \LieA{g}) \to \LieA{g} \) of \cref{ex:contangentBundle:actionWithoutMomentumMap} shows, the adjoint of a linear mapping between Fréchet spaces may not exist.
But, our strategy for the construction of the normal form will involve dualizing. 
Thus, in all constructions of \cref{sec:cotangentBundleReduction} that rely on dualizing, we will need to assume that the corresponding adjoint maps exist.
Moreover, even if the adjoint exists and the original map is injective, the adjoint is in general not surjective but only has a dense image.
We will also assume that such adjoints are actually surjective (see \cref{defn:cotangentBundle:sliceCompatible} below for the precise formulation of this assumption).
As we will see in \cref{sec:yangMillsHiggs}, for Yang--Mills theory these assumptions concerning adjoints are fulfilled, because the maps involved are elliptic differential operators.

Thus, in summary, we make the following assumptions:
\begin{enumerate}
	\item
		The action of \( G \) on \( Q \) is proper, and a slice exists at every point.
	\item
		For a given choice of \( \CotBundle Q \), the \( G \)-action lifts to an action on \( \CotBundle Q \) that has a momentum map \( J: \CotBundle Q \to \LieA{g}^* \) (relative to a chosen dual pairing \( \kappa: \LieA{g}^* \times \LieA{g} \to \R \)).
	\item
		The adjoints of some injective maps relevant for the development of the theory exist and are surjective.
		This assumption will be made precise below in \cref{defn:cotangentBundle:sliceCompatible} by introducing the notion of a slice compatible with the cotangent bundle structures. 
\end{enumerate}
These assumptions will be in effect throughout \cref{sec:cotangentBundleReduction}.
For the sake of completeness, they will be repeated in the statements of the theorems that follow. 

\subsection{Lifted slices and normal form}
\label{sec:cotangentBundle:normalForm}
Our general strategy for the construction of the normal form is the following:
\begin{enumerate}
	\item 
		By using a slice \( S \) at \( q \in Q \), reduce the problem to \( \CotBundle (G \times_{G_q} S) \),
		where \( G \times_{G_q} S \equiv (G \times S) \slash G_q \) is the \( S \)-bundle associated to \( G \to G \slash G_q \).
	\item
		Establish an equivariant diffeomorphism \( \CotBundle (G \times_{G_q} S) \isomorph G \times_{G_q} (\LieA{m}^* \times \CotBundle S) \), where \( \LieA{m} \) is a complement of \( \LieA{g}_q \) in \( \LieA{g} \).
	\item
		Calculate the momentum map under these identifications.
\end{enumerate}

Let \( p \in \CotBundle Q \) be a point in the fiber over \( q \in Q \).
Assume that the \( G \)-action on \( Q \) has a slice \( S \) at \( q \).
According to the strategy outlined above, the first step is to reduce the problem of determining the local structure of the momentum map near \( p \) to a problem on \( \CotBundle (G \times_{G_q} S) \).
This reduction is accomplished by applying the following standard result to the tube map\footnote{Recall that \( \chi \) is a diffeomorphism onto an open neighborhood of \( q \) in \( Q \). Moreover, it is \( G \)-equivariant with respect to the action of \( G \) on \( G \times_{G_q} S \) by left translation on the \( G \)-factor.} \( \chi: G \times_{G_q} S \to Q \), which plainly extends from the finite- to the infinite-dimensional setting.

\begin{prop}[Lifting point transformations]
	Let \( C \) and \( Q \) be Fréchet manifolds and let \( \phi: C \to Q \) be a diffeomorphism.
	Assume that the lift
	\begin{equation}
		\cotangent \phi: \CotBundle Q \to \CotBundle C, \quad p \mapsto \phi^* p,
	\end{equation}
	of \( \phi \) is a diffeomorphism.
	Then, \( \cotangent \phi \) is a symplectomorphism.
	If, moreover, \( \phi \) is \( G \)-equivariant and the lifted action on \( \CotBundle Q \) has a momentum map \( J_Q \), then \( J_C \defeq J_Q \circ \cotangent \phi^{-1}: \CotBundle C \to \LieA{g}^* \) is a momentum map for the lifted \( G \)-action on \( \CotBundle C \).
\end{prop}

According to step two, we now establish a convenient identification of the cotangent bundle of \( G \times_{G_q} S \).
Generalizing the finite-dimensional case, we say that a Lie subgroup \( H \subseteq G \) is \emphDef{reductive} if its Lie algebra \( \LieA{h} \) has an \( \AdAction_H \)-invariant complement in \( \LieA{g} \).
Since we consider only proper actions of \( G \) on \( Q \), the stabilizer \( G_q \) is always compact and, hence, the following lemma shows that \( G_q \) is a reductive Lie subgroup of \( G \).
\begin{lemma}
	\label{prop:cotangentBundle:compactGroupReducitive}
	For every compact Lie subgroup \( H \subseteq G \), there exists an \( \AdAction_H \)-invariant complement \( \LieA{m} \) of \( \LieA{h} \) in \( \LieA{g} \).
	Moreover, there exists a weakly complementary decomposition \( \LieA{g}^* = \LieA{m}^* \oplus \LieA{h}^* \), that is, there exist weakly closed subspaces \( \LieA{m}^* \subseteq \LieA{g}^* \) and \( \LieA{h}^* \subseteq \LieA{g}^* \) such that \( \LieA{g}^* \) is topologically isomorphic to \( \LieA{m}^* \oplus \LieA{h}^* \) relative to the weak topology.
\end{lemma}
\begin{proof}
	Every compact Lie group is finite-dimensional, because every locally compact topological vector space is finite-dimensional, see \parencite[Proposition~8.7.1]{Koethe1983}.
	As a finite-dimensional subspace, \( \LieA{h} \) is automatically closed and has a topological complement by \parencite[Proposition~15.5.2 and~20.5.5]{Koethe1983}.
	The complement can be chosen to be \( \AdAction_H \)-invariant by taking the average over the projection using the invariant Haar measure.
	By \parencite[Proposition~20.5.1]{Koethe1983}, there exists a weakly complementary decomposition \( \LieA{g}^* = \LieA{h}^0 \oplus \LieA{m}^0 \), where the superscript denotes the annihilator.
	Choose \( \LieA{m}^* = \LieA{h}^0 \) and \( \LieA{h}^* = \LieA{m}^0 \).
\end{proof}

Let \( \iota: G \times S \to G \times_{G_q} S \) be the natural projection and, for \( a \in G \) and \( s \in S \), let \( \iota_a: S \to G \times_{G_q} S \) and \( \iota_s: G \to G \times_{G_q} S \) denote the induced embeddings, respectively.
\begin{defn}
	\label{defn:cotangentBundle:sliceCompatible}
	Let \( S \) be a slice at \( q \) and let \( \LieA{m} \) be an \( \AdAction_{G_q} \)-invariant complement of \( \LieA{g}_q \) in \( \LieA{g} \).
	For a given choice of \( \CotBundle S \) and \( \CotBundle (G \times_{G_q} S) \), the slice \( S \) will be called \emphDef{compatible with the cotangent bundle structures} if it fulfills the following additional requirements:
	\begin{thmenumerate}[label={[S\arabic*]}, ref={\textup{[S\arabic*]}}]
		\item
			\label{i:cotangentBundle:sliceCompatible:lift}
			The lift \( \cotangent \chi \) of the tube map \( \chi: G \times_{G_q} S \to Q \) exists and is a diffeomorphism onto its image.
		\item
			\label{i:cotangentBundle:sliceCompatible:adjoints}
			The injective maps
			\begin{equation}
				\tangent_a \iota_s (a \ldot \cdot): \LieA{m} \to \TBundle_{\equivClass{a,s}} (G \times_{G_q} S),
				\qquad
				\tangent_s \iota_a: \TBundle_s S \to \TBundle_{\equivClass{a,s}} (G \times_{G_q} S)
			\end{equation}
			have surjective adjoints.
			\qedhere
	\end{thmenumerate}
\end{defn}
Whether a slice is compatible with the cotangent bundle structures strongly depends on the dual pairings involved, so that one cannot hope to find a general criterion for the existence of such a slice.
Note that, in the finite-dimensional setting, every slice is compatible with the cotangent bundle structures.
In the sequel, we will show that there are two natural choices for \( \CotBundle (G \times_{G_q} S) \).
\begin{lemma}
	For every choice of an \( \AdAction_{G_q} \)-invariant complement \( \LieA{m} \) of \( \LieA{g}_q \) in \( \LieA{g} \), there exists a \( G \)-equivariant diffeomorphism
	\begin{equation+}
		\label{eq:cotangentBundle:diffeoAssociatedBundle}
		\TBundle (G \times_{G_q} S) \isomorph G \times_{G_q} (\LieA{m} \times \TBundle S).
		\qedhere
	\end{equation+}
\end{lemma}
\begin{proof}
	It is well-known that the choice of an \( \AdAction \)-invariant complement \( \LieA{m} \) yields a homogeneous connection in \( G \to G \slash G_q \). 
	Indeed, the horizontal space at \( a \in G \) is, by definition, \( a \ldot \LieA{m} \) and the horizontal lift of a vector \( \equivClass{a, \xi} \in G \times_{G_q} \LieA{m} \isomorph \TBundle (G \slash G_q) \) to the point \( a \in G \) is \( a \ldot \xi \).
	Accordingly, the tangent bundle to the associated bundle \( G \times_{G_q} S \to G \slash G_q \) splits into its horizontal and vertical parts so that the \( G \)-equivariant map defined by
	\begin{equation}\label{eq:cotangentBundle:identificationTBundleTube}
		G \times \LieA{m} \times \TBundle S \to \TBundle (G \times_{G_q} S),
		\qquad
		(a, \varsigma, Y_s) \mapsto \tangent_a \iota_s (a \ldot \varsigma) + \tangent_s \iota_a (Y_s),
	\end{equation}
	is a \( G_q \)-invariant submersion which descends to a \( G \)-equivariant diffeomorphism between \(  G \times_{G_q} (\LieA{m} \times \TBundle S) \) and \( \TBundle (G \times_{G_q} S) \).
\end{proof}
By dualizing the isomorphism~\eqref{eq:cotangentBundle:diffeoAssociatedBundle}, we get a natural choice for the cotangent bundle \( \CotBundle (G \times_{G_q} S) \).
For that purpose, we choose \( \LieA{m}^* \) as constructed in the proof of \cref{prop:cotangentBundle:compactGroupReducitive}.
Moreover, viewing \( S \) as a submanifold of \( Q \), let \( \CotBundle S \) be the image of \( \TBundle S \) under the \( G \)-equivariant diffeomorphism \( \TBundle Q \to \CotBundle Q \).
Now, the cotangent bundle is provided by the \( G \)-equivariant diffeomorphism
\begin{equation}
	\phi: G \times_{G_q} (\LieA{m}^* \times \CotBundle S) \to \CotBundle (G \times_{G_q} S)
\end{equation}
defined by
\begin{equation}
	\label{eq:cotangentBundle:cotTubeDiffeo}
	\dualPair{\phi(\equivClass{a, (\nu, \alpha_s)})}{\tangent_a \iota_s (a \ldot \varsigma) + \tangent_s \iota_a (Y_s)}
		= \kappa(\nu,\varsigma) + \dualPair{\alpha_s}{Y_s}
\end{equation}
for \( \varsigma \in \LieA{m} \) and \( Y_s \in \TBundle_s S \).
With this choice of \( \CotBundle (G \times_{G_q} S) \) the second condition~[S2] in \cref{defn:cotangentBundle:sliceCompatible} is automatically satisfied, because the adjoints are the identical mappings.
\begin{remark}
	Alternatively, one could define \( \CotBundle (G \times_{G_q} S) \) by dualizing the \( G \)-equivariant diffeomorphism \( \tangent \chi: \TBundle (G \times_{G_q} S) \to \TBundle Q \) given by the tube diffeomorphism \( \chi \).
	In this case, the first condition~[S1] in \cref{defn:cotangentBundle:sliceCompatible} is automatically satisfied.
\end{remark}

\begin{remark}
	In the finite-dimensional context, \textcite[Proposition~13]{Schmah2007} established a similar identification of \( \CotBundle (G \times_{G_q} S) \) using regular cotangent bundle reduction.
	The starting point is the cotangent bundle \( \CotBundle (G \times S) \).
	Using left translation, identify \( \CotBundle G \) with \( G \times \LieA{g}^* \).
	Denote points in \( \CotBundle(G \times S) \isomorph G \times \LieA{g}^* \times \CotBundle S \) by tuples \( (a, \mu, \alpha_s) \).
	A straightforward calculation shows that the cotangent lift of the twisted \( G_q \)-action \( h \cdot (g, s) = (gh^{-1}, h \cdot s) \) on \( G \times S \) has the momentum map
	\begin{equation}
		J_{G_q^\T} (a, \mu, \alpha_s) = - \restr{\mu}{\LieA{g}_q} + J_{G_q}(\alpha_s),
	\end{equation}
	where \( J_{G_q}: \CotBundle S \to \LieA{g}_q^* \) is the momentum map for the lifted \( G_q \)-action on \( \CotBundle S \) and \( \restr{\mu}{\LieA{g}_q} \) denotes the restriction of \( \mu \) to \( \LieA{g}_q \).
	Define the map
	\begin{align}
		\varphi: G \times \LieA{g}^* \times \CotBundle S \supseteq J^{-1}_{G_q^{\T}}(0) &\to \CotBundle (G \times_{G_q} S)
	\shortintertext{by} 
		\dualPair{\varphi(a, \mu, \alpha_s)}{\tangent \iota (a \ldot \xi, Y_s)} &= \kappa(\mu, \xi) + \dualPair{\alpha_s}{Y_s}, \qquad \
	\end{align}
	for \( \xi \in \LieA{g} \) and \( Y_s \in \TBundle_s S \).
	The regular cotangent bundle reduction theorem \parencite[Theorem~6.6.1]{OrtegaRatiu2003} shows that \( \varphi \) is a submersion and that it descends to a symplectomorphism \( \check{\varphi} \) of \( J_{{G_q^{\T}}}^{-1}(0) \slash G_q \) with \( \CotBundle (G \times_{G_q} S) \).

	In order to establish the link with the isomorphism \( \phi \) discussed above (taken in the finite-dimensional case), we define
	\begin{align}
		\psi: G \times (\LieA{m}^* \times \CotBundle S) &\to G \times (\LieA{g}^* \times \CotBundle S),
		\\
		(a, \nu, \alpha_s) &\mapsto (a, \nu + J_{G_q}(\alpha_s), \alpha_s).
	\end{align}
	By construction, \( \psi \) takes values in \( J_{{G_q^{\T}}}^{-1}(0) \) and on this set it has a smooth inverse,
	\begin{equation}
		J_{{G_q^{\T}}}^{-1}(0) \to G \times (\LieA{m}^* \times \CotBundle S), \quad (a, \mu, \alpha_s) \mapsto (a, \restr{\mu}{\LieA{m}}, \alpha_s).
	\end{equation}
	Now, it is an easy exercise in chasing identifications to see that the following diagram commutes
	\begin{equationcd}[label=eq:cotangentBundle:relationWithSchmah]
		J_{{G_q^{\T}}}^{-1}(0)
			\to[rd, "\varphi"]
			\to[d]
			& 
			& G \times (\LieA{m}^* \times \CotBundle S)
			\to[ld]
			\to[ll, "\psi" swap]
			\to[d]
			\\
		J_{{G_q^{\T}}}^{-1}(0) \slash G_q
			\to[r, "\check{\varphi}"]
			& \CotBundle(G \times_{G_q} S)
			& G \times_{G_q} (\LieA{m}^* \times \CotBundle S).
			\to[l, "\phi" swap]
	\end{equationcd}
	Since the regular cotangent bundle reduction theorem does not directly generalize to infinite dimensions\footnote{We do prove a regular reduction theorem below in \cref{prop:cotangentBundle:oneOrbitTypeReducedSpaceCotangentBundle}. However, this result relies on the existence of the normal form and thus cannot be used to construct it.}, the approach of \parencite{Schmah2007} is not available to us in our infinite-dimensional context.
	One can, however, read~\eqref{eq:cotangentBundle:relationWithSchmah} as a proof that the regular reduction of \( \CotBundle(G \times S) \) does coincide with \( \CotBundle (G \times_{G_q} S) \), indeed. 
\end{remark}

\begin{lemma}
	Under the diffeomorphism \( \phi \), the momentum map \( J: \CotBundle (G \times_{G_q} S) \to \LieA{g}^* \) for the lifted \( G \)-action is identified with the map
	\begin{equation}
		(J \circ \phi) (\equivClass{a, (\nu, \alpha_s)}) = \CoAdAction_a (\nu + J_{G_q}(\alpha_s)),
	\end{equation}
	where \( J_{G_q}: \CotBundle S \to \LieA{g}_q^* \) is the momentum map for the lifted \( G_q \)-action on \( S \) and \( \CoAdAction \) denotes the coadjoint action defined by \( \dualPair{\CoAdAction_a \nu}{\xi} = \dualPair{\nu}{\AdAction_{a^{-1}} \xi} \).
	In particular, the condition \( (J \circ \phi) (\equivClass{a, (\nu, \alpha_s)}) = 0 \) is equivalent to \( J_{G_q}(\alpha_s) = 0 \) and \( \nu = 0 \).
\end{lemma}
\begin{proof}
	First, we note that the lifted \( G_q \)-action on \( \CotBundle S \) has a momentum map \( J_{G_q}: \CotBundle S \to \LieA{g}_q^* \), because, by properness of the action, \( G_q \) is compact and hence finite-dimensional.
	The canonical \( G \)-action \( g \cdot \equivClass{a,s} = \equivClass{ga,s} \) on \( G \times_{G_q} S \) has the fundamental vector field
	\begin{equation}
	 	\xi \ldot \equivClass{a,s} = \tangent_a \iota_s (\xi \ldot a) = \tangent_a \iota_s \left(a \ldot (\AdAction_a^{-1} \xi)\right), \quad \xi \in \LieA{g}.
	\end{equation}
	Moreover, \( \equivClass{a h, s} = \equivClass{a, h\cdot s} \) for all \( h \in G_q \) implies
	\begin{equation}
		\tangent_a \iota_s (a \ldot \varrho) = \tangent_s \iota_a (\varrho \ldot s), \qquad \varrho \in \LieA{g}_q.
	\end{equation}
	Let \( \beta = \phi(\equivClass{a, (\nu, \alpha_s)}) \in \CotBundle_{\equivClass{a,s}} (G \times_{G_q} S) \).
	Then, for every \( \xi \in \LieA{g} \), the momentum map satisfies 
	\begin{equation}\begin{split}
	 	\kappa(J(\beta), \AdAction_a \xi)
	 		&= \dualPair*{\beta}{(\AdAction_a \xi) \ldot \equivClass{a,s}}
	 		\\
	 		&= \dualPair*{\beta}{\tangent_a \iota_s (a \ldot \xi)}
	 		\\
	 		&= \dualPair*{\beta}{\tangent_a \iota_s \left(a \ldot \xi_{\LieA{g}_q} \right)}
	 			+ \dualPair*{\beta}{\tangent_a \iota_s \left(a \ldot \xi_{\LieA{m}}\right)}
	 		\\
	 		&= \dualPair*{\beta}{\tangent_s \iota_a \left(\xi_{\LieA{g}_q} \ldot s\right)}
	 			+ \dualPair*{\beta}{\tangent_a \iota_s \left(a \ldot \xi_{\LieA{m}}\right)}
	 		\\
	 		&= \dualPair*{\alpha_s}{\xi_{\LieA{g}_q} \ldot s}
	 			+ \kappa(\nu,\xi_{\LieA{m}})
	 		\\
	 		&= \kappa\left(J_{G_q}(\alpha_s), \xi_{\LieA{g}_q}\right)
	 			+ \kappa(\nu, \xi_{\LieA{m}}),
	\end{split}\end{equation}
	where we have decomposed \( \xi = \xi_{\LieA{g}_q} + \xi_{\LieA{m}} \) into \( \xi_{\LieA{g}_q} \in \LieA{g}_q \) and \( \xi_{\LieA{m}} \in \LieA{m} \).
	In the line before the last line we have used~\eqref{eq:cotangentBundle:cotTubeDiffeo}.
	Hence,
	\begin{equation}
		\CoAdAction_{a^{-1}} (J \circ \phi) (\equivClass{a, (\nu, \alpha_s)}) \equiv \CoAdAction_{a^{-1}} J(\beta) = J_{G_q}(\alpha_s) + \nu
	\end{equation}
	and so \( J(\beta) = 0 \) if and only if \( J_{G_q}(\alpha_s) = 0 \) and \( \nu = 0 \).
\end{proof}

Combining the diffeomorphism \( \phi \) with the local tube diffeomorphism
\begin{equation}
 	\cotangent \chi: \CotBundle Q \to \CotBundle (G \times_{G_q} S)
\end{equation} 
yields a convenient normal form for the momentum map of the lifted \( G \)-action on \( \CotBundle Q \).
\begin{thm}[Normal form]\label{prop:cotangentBundle:simpleNormalForm}
	Let \( Q \) be a Fréchet \( G \)-manifold.
	Assume that the \( G \)-action on \( Q \) is proper, that \( \CotBundle Q \) is a Fréchet vector bundle, which is \( G \)-equivariantly isomorphic to \( \TBundle Q \), and that the lifted action on \( \CotBundle Q \), endowed with its canonical symplectic form \( \omega \), has a momentum map \( J \).
	Let \( p \in \CotBundle_q Q \) and assume that the \( G \)-action on \( Q \) admits a slice \( S \) at \( q \) compatible with the cotangent bundle structures.
	Then, the map \( \Phi: G \times_{G_q} (\LieA{m}^* \times \CotBundle S) \to \CotBundle Q \) defined by
	\begin{equation}\label{eq:cotangentBundle:normalFormCotBundle}
		\dualPair{\Phi(\equivClass{a, (\nu, \alpha_s)})}{(\AdAction_a \rho) \ldot (a \cdot s) + a \ldot Y_s} 
			= \kappa(\nu, \rho) + \dualPair{\alpha_s}{Y_s},
	\end{equation}
	for all \( \rho \in \LieA{m} \) and \( Y_s \in \TBundle_s S \), is a diffeomorphism onto an open neighborhood of \( p \) in \( \CotBundle Q \).
	Moreover, \( \Phi \) is \( G \)-equivariant with respect to left translation on the \( G \)-factor and the lifted action on \( \CotBundle Q \).
	Under \( \Phi \), the momentum map \( J: \CotBundle Q \to \LieA{g}^* \) for the lifted \( G \)-action is identified with the map
	\begin{equation}\label{eq:cotangentBundle:normalFormMomentumMap}
		(J \circ \Phi) (\equivClass{a, (\nu, \alpha_s)}) = \CoAdAction_a (\nu + J_{G_q}(\alpha_s)),
	\end{equation}
	where \( J_{G_q}: \CotBundle S \to \LieA{g}_q^* \) is the momentum map for the lifted \( G_q \)-action on \( \CotBundle S \).
\end{thm}
\begin{proof}
	Let the map \( \Phi = \cotangent \chi^{-1} \circ \phi \) be defined as the composition of the diffeomorphism \( \phi: {G \times_{G_q} (\LieA{m}^* \times \CotBundle S)} \to \CotBundle (G \times_{G_q} S) \) with the local diffeomorphism \( {\cotangent \chi^{-1}: \CotBundle (G \times_{G_q} S) \to \CotBundle Q} \).
	The composition of~\eqref{eq:cotangentBundle:identificationTBundleTube} with \( \tangent \chi \) yields the map
	\begin{equation}
		G \times_{G_q} (\LieA{m} \times \TBundle S) \to \TBundle Q,
		\quad
		\equivClass{a, (\rho, Y_s)} \mapsto \tangent_{a,s} (\chi \circ \iota)(a \ldot \rho, Y_s).
	\end{equation}
	Note that we have
	\begin{equation}
		\chi \circ \iota (a, s) = \chi (\equivClass{a,s}) = a \cdot s.
	\end{equation}
	Thus, the above map reads
	\begin{equation}
		\label{eq:cotangentBundle:normalFormTBundle}
		G \times_{G_q} (\LieA{m} \times \TBundle S) \to \TBundle Q,
		\quad
		\equivClass{a, (\rho, Y_s)} \mapsto (\AdAction_a \rho) \ldot (a \cdot s) + a \ldot Y_s,
	\end{equation}
	which is a diffeomorphism onto its image, the latter being a subbundle of \( \TBundle Q \).
	By dualizing, we obtain the expression~\eqref{eq:cotangentBundle:normalFormCotBundle} for \( \Phi \).
	Note that~\eqref{eq:cotangentBundle:normalFormCotBundle} is a defining equation for \( \Phi \), because the map given in~\eqref{eq:cotangentBundle:normalFormTBundle} is a diffeomorphism onto its image.
	The asserted properties follow immediately from the previous discussion. 
\end{proof}

We should note, however, that the semi-global diffeomorphism \( \Phi \) does not yield a \emph{symplectic slice} for the \( G \)-action.
In fact, it is not even a slice for the lifted action, because we have taken the quotient by the `wrong' stabilizer group, \ie, the model space around \( p \in \CotBundle_q Q \) is of the form \( G \times_{G_q} (\LieA{m}^* \times \CotBundle S) \) instead of \( G \times_{G_p} W \) for some submanifold \( W \) which would be a slice at \( p \) for the cotangent lifted action.
In finite dimensions, much work in the study of singular cotangent bundle reduction is devoted to constructing a bona fide symplectic slice adapted to the cotangent bundle structure (see, \eg, \parencite{Schmah2007,Rodriguez-OlmosTeixido-Roman2014}).
Converting the normal form \( G \times_{G_q} (\LieA{m}^* \times \CotBundle S) \) into a symplectic slice requires a detailed analysis of the Witt--Artin decomposition in the cotangent bundle setting and then extending these infinitesimal results to a local statement using the inverse function theorem.
It is not clear if and how these steps generalize to the infinite-dimensional context as they lead to delicate issues of analytic nature.
It turns out, however, that the simple normal form~\eqref{eq:cotangentBundle:normalFormMomentumMap} of the momentum map we have constructed so far is sufficient for most questions concerning singular cotangent bundle reduction.

Recall that the choice of a complement \( \LieA{m} \) of \( \LieA{g}_q \) in \( \LieA{g} \) yields a canonical identification of \( \TBundle_{\equivClass{e}} G \slash G_q \) with \( \LieA{m} \).
\begin{prop}
	\label{prop:cotangentBundle:normalFormSymplecticForm}
	Under the diffeomorphism \( \Phi \) constructed in \cref{prop:cotangentBundle:simpleNormalForm} and under the identification \( \TBundle_{\equivClass{a, (\nu, \alpha_s)}} (G \times_{G_q} (\LieA{m}^* \times \CotBundle S)) \isomorph \LieA{m} \times \LieA{m}^* \times \TBundle_{\alpha_s} (\CotBundle S) \), the canonical symplectic form \( \omega \) on \( \CotBundle Q \) takes the following form:
	\begin{equation}\begin{split}
		\label{eq:cotangentBundle:normalFormSymplecticForm}
		(\Phi^* \omega &)_{\equivClass{a, (\nu, \alpha_s)}} ((\xi^1, \eta^1, Z^1), (\xi^2, \eta^2, Z^2)) 
			\\
			&= \kappa(\eta^1, \xi^2) - \kappa(\eta^2, \xi^1)
				- \kappa\left(\nu + J_{G_q}(\alpha_s), \commutator{\xi^1}{\xi^2}\right)
				+ \omega^S_{\alpha_s}(Z^1, Z^2),
	\end{split}\end{equation}
	where \( \xi^i \in \LieA{m} \), \( \eta^i \in \LieA{m}^* \) and \( Z^i \in \TBundle_{\alpha_s} (\CotBundle S) \) for \( i = 1,2 \) and \( \omega^S \) is the canonical symplectic form on \( \CotBundle S \).
	In particular, at points with zero momentum the third term on the right-hand side vanishes and \( \Phi^* \omega \) is the direct sum of the canonical symplectic forms on \( \LieA{m} \times \LieA{m}^* \) and \( \CotBundle S \).
\end{prop}
\begin{proof}
	Let \( \CotBundleProj: \CotBundle Q \to Q \) be the canonical projection.
	Then, \( \CotBundleProj \circ \Phi(\equivClass{a, (\nu, \alpha_s)}) = a \cdot s \).
	Using~the definition~\eqref{eq:cotangentBundle:normalFormCotBundle} of \( \Phi \), for the pull-back of the canonical \( 1 \)-form \( \theta \) we find:
	\begin{equation}\begin{split}
		(\Phi^* \theta)_{\equivClass{a, (\nu, \alpha_s)}}&(\equivClass{a \ldot \xi, (\eta, Z)})
			\\
			&= \dualPair{\Phi(\equivClass{a, (\nu, \alpha_s)})}{\tangent_{\equivClass{a, (\nu, \alpha_s)}} (\CotBundleProj \circ \Phi)(\equivClass{a \ldot \xi, (\eta, Z)})}, 
			\\
			&= \dualPair{\Phi(\equivClass{a, (\nu, \alpha_s)})}{(\AdAction_a \xi) \ldot (a \cdot s) + a \ldot \tangent_{\alpha_s} \CotBundleProj (Z)}
			\\
			&= \dualPair{\Phi(\equivClass{a, (\nu, \alpha_s)})}{(\AdAction_a \xi_{\LieA{m}}) \ldot (a \cdot s) + a \ldot (\xi_{\LieA{g}_q} \ldot s) + a \ldot \tangent_{\alpha_s} \CotBundleProj (Z)}
			\\
			&= \kappa(\nu, \xi_{\LieA{m}}) + \dualPair{\alpha_s}{\xi_{\LieA{g}_q} \ldot s}+ \dualPair{\alpha_s}{\tangent_{\alpha_s} \CotBundleProj (Z)}
			\\
			&= \kappa(\nu, \xi_{\LieA{m}}) + \kappa(J_{G_q}(\alpha_s), \xi_{\LieA{g}_q}) + \theta^S_{\alpha_s} (Z),
	\end{split}\end{equation}
	where \( \xi = \xi_{\LieA{m}} + \xi_{\LieA{g}_q} \in \LieA{g} \), \( \eta \in \LieA{m}^* \) and \( Z \in \TBundle_{\alpha_s} (\CotBundle S) \), and \( \theta^S \) is the canonical \( 1 \)-form on \( \CotBundle S \).
	If we introduce the left Maurer--Cartan form \( \vartheta \in \DiffFormSpace^1(G, \LieA{g}) \) by \( \vartheta_a (a \ldot \xi) = \xi \), then we get
	\begin{equation}
		(\Phi^* \theta)_{\equivClass{a, (\nu, \alpha_s)}} = \kappa\left(\nu + J_{G_q}(\alpha_s), \vartheta_a (\cdot)\right) + \theta^S_{\alpha_s} (\cdot).
	\end{equation}
	The Maurer--Cartan equation \( \dif \vartheta_a = - \frac{1}{2} \wedgeLie{\vartheta_a}{\vartheta_a} \) yields
	\begin{equation}\begin{split}
		(\dif \Phi^* \theta)_{\equivClass{a, (\nu, \alpha_s)}} &(\equivClass{a \ldot \xi^1, (\eta^1, Z^1)}, \equivClass{a \ldot \xi^2, (\eta^2, Z^2)})
			\\
			&= \kappa(\eta^1, \xi^2_{\LieA{m}}) - \kappa(\eta^2, \xi^1_{\LieA{m}})
				+ \kappa(\nu + J_{G_q}(\alpha_s), (\dif \vartheta)_a (a \ldot \xi^1, a \ldot \xi^2))
				\\
				&\quad + \kappa(\tangent_{\alpha_s} J_{G_q} (Z^1), \xi^2_{\LieA{g}_q}) - \kappa(\tangent_{\alpha_s} J_{G_q} (Z^2), \xi^1_{\LieA{g}_q})
				\\
				&\quad + (\dif \theta^S)_{\alpha_s}(Z^1, Z^2)
			\\
			&= \kappa(\eta^1, \xi^2_{\LieA{m}}) - \kappa(\eta^2, \xi^1_{\LieA{m}})
				- \kappa\left(\nu + J_{G_q}(\alpha_s), \commutator{\xi^1}{\xi^2}\right)
				\\
				&\quad 
					+ \omega^S_{\alpha_s}(Z^1, \xi^2_{\LieA{g}_q} \ldot \alpha_s) - \omega^S_{\alpha_s}(Z^2, \xi^1_{\LieA{g}_q} \ldot \alpha_s)
					+ \omega^S_{\alpha_s}(Z^1, Z^2).
	\end{split}\end{equation}
	Now the identification \( \TBundle_{\equivClass{a, (\nu, \alpha_s)}} (G \times_{G_q} (\LieA{m}^* \times \CotBundle S)) \isomorph \LieA{m} \times \LieA{m}^* \times \TBundle_{\alpha_s} (\CotBundle S) \) amounts to setting the \( \LieA{g}_q \)-component of \( \xi^i \) to zero and we thus arrive at the claimed formula~\eqref{eq:cotangentBundle:normalFormSymplecticForm}.
\end{proof}

\subsection{Symplectic stratification}
In finite dimensions, it is well known that the reduced phase space of a Hamiltonian \( G \)-system is stratified by symplectic manifolds.
The general theory of singular symplectic reduction in the infinite-dimensional context is worked out in \parencite{DiezThesis}.
In contrast to the finite-dimensional picture, the general infinite-dimensional reduction scheme needs additional assumptions, because one has to handle weakly symplectic forms, the Darboux theorem fails and a simple inverse function theorem is not available beyond Banach spaces.
The aim of this section is to show that the cotangent bundle structure and the normal form established in the previous section yield a symplectic reduction theorem without relying on the heavy machinery that is used in the general theory in \parencite{DiezThesis}.

As usual, for a \( G \)-manifold \( M \) and a subgroup \( H \subseteq G \), we denote the subset of isotropy type \( H \) by \( M_H \) and the subset of orbit type \( (H) \) by \( M_{(H)} \).
For the definition of the notion of stratification we refer to \cref{sec:calculus:groupActionsSlices}.

\begin{thm}
	\label{prop:cotangentBundle:singularSympRed}
	Let \( Q \) be a Fréchet \( G \)-manifold.
	Assume that the \( G \)-action on \( Q \) is proper, that it admits at every point a slice compatible with the cotangent bundle structures and that the decomposition of \( Q \) into orbit types satisfies the frontier condition.
	Moreover, assume that \( \CotBundle Q \) is a Fréchet vector bundle, which is \( G \)-equivariantly isomorphic to \( \TBundle Q \), and that the lifted action on \( \CotBundle Q \), endowed with its canonical symplectic form \( \omega \), has a momentum map \( J \).
	Then, the following holds.
	\begin{thmenumerate}
		\item
			The set of orbit types of \( P \defeq J^{-1}(0) \) with respect to the lifted \( G \)-action coincides with the set of orbit types for the \( G \)-action on \( Q \). 
	\item
			The reduced phase space \( \check{P} \defeq J^{-1}(0) \slash G \) is stratified into orbit type manifolds \( \check{P}_{(K)} \defeq (J^{-1}(0))_{(K)} \slash G \).
	\item 
			Assume, additionally, that every orbit is symplectically closed, that is, the symplectic double orthogonal \( (\LieA{g} \ldot p)^{\omega \omega} \) coincides with \( \LieA{g} \ldot p \) for all \( p \in P \).
			Then, for every orbit type \( (K) \), the manifold \( \check{P}_{(K)} \) carries a symplectic form \( \check{\omega}_{(K)} \) uniquely determined by
			\begin{equation}
				\pi_{(K)}^* \check{\omega}_{(K)} = \restr{\omega}{P_{(K)}},
			\end{equation}
			where \( \pi_{(K)}: P_{(K)} \to \check{P}_{(K)} \) is the natural projection.
			\qedhere
	\end{thmenumerate}
\end{thm}
In the sequel, we will prove this theorem by means of a series of lemmas.

Let us start with the observation that the orbit types of the lifted action are tightly connected to the ones for the action on the base manifold.
The following results extend the determination of orbit types of the lifted action \parencites{Rodriguez-Olmos2006}[Theorem~5]{PerlmutterRodriguez-OlmosSousa-Dias2007} to the infinite-dimensional setting.
\begin{lemma}
	\label{prop:cotangentBundle:orbitTypes}
	Under the assumptions of \cref{prop:cotangentBundle:singularSympRed}, for every orbit type \( (K) \) of \( \CotBundle Q \) there exist a triple \( (H, \tilde{H}, L) \) such that \( (K) \) is represented by \( K = \tilde{H} \intersect L \), where \( \tilde{H} \subseteq H \) are stabilizer subgroups of \( Q \) and \( L \) is a stabilizer subgroup of the \( H \)-action on \( (\LieA{g} \slash \LieA{h})^* \) with \( \LieA{h} \) being the Lie algebra of \( H \).

	Conversely, for every triple \( (H, (\tilde{H}), L) \), where \( H \) is a stabilizer subgroup of \( Q \), \( (\tilde{H}) \) is an orbit type of \( Q \) fulfilling \( (\tilde{H}) \leq (H) \) and \( L \) is a stabilizer subgroup of the \( H \)-action on \( (\LieA{g} \slash \LieA{h})^* \), there exists a representative \( \tilde{H} \) of \( (\tilde{H}) \) and a stabilizer subgroup \( K \) of \( \CotBundle Q \) such that \( K = \tilde{H} \intersect L \).   
\end{lemma}
\begin{proof}
	Let \( p \in \CotBundle_q Q \) and denote the stabilizer of \( p \) and \( q \) by \( K \) and \( H \), respectively.
	Choose an \( \AdAction_H \)-invariant complement \( \LieA{m} \) of \( \LieA{h} \equiv \LieA{g}_q \) in \( \LieA{g} \), which is possible because \( H \) is compact.
	We will first show that there exists \( \tilde{q} \in Q \) (close to \( q \)) and \( \nu \in \LieA{m}^* \) such that
	\begin{equation}
		\label{eq:cotangentBundle:orbitTypes:relation}
		K = G_{\tilde{q}} \intersect H_\nu.
	\end{equation}
	For this purpose, let \( S \) be a slice at \( q \). 
	Denote the model space of \( S \) by \( X \).
	By \cref{prop:cotangentBundle:simpleNormalForm}, there exists \( \nu \in \LieA{m}^* \) and \( \alpha \in \CotBundle_q S \) such that \(  \Phi(\equivClass{e, (\nu, (q, \alpha))}) = p \).
	Since \( \Phi \) is \( G \)-equivariant, the common stabilizer of \( \nu \) and \( \alpha \) under the \( H \)-action on \( \LieA{m}^* \times \CotBundle_q S \) is \( K \), that is, \( K = H_\alpha \intersect H_\nu \).
	By assumption, there exists a \( G \)-equivariant diffeomorphism between \( \TBundle Q \) and \( \CotBundle Q \).
	The latter induces an \( H \)-equivariant diffeomorphism between \( \TBundle S \) and \( \CotBundle S \).
	Thus, there exists \( x \in X \isomorph \TBundle_q S \) whose stabilizer under the \( H \)-action is \( H_\alpha \).
	Since the topology of \( X \) is generated by absorbent sets and \( S \) is (diffeomorphic to) an open neighborhood of \( 0 \) in \( X \), there exists \( r \in \R \) such that \( x = r \tilde{q} \) for some \( \tilde{q} \in S \).
	By linearity of the \( H \)-action, both \( x \) and \( \tilde{q} \) have the same stabilizer \( H_\alpha \).
	As \( \tilde{q} \in S \), \iref{i::slice:SliceDefOnlyStabNotMoveSlice} of \cref{defn:slice:slice} shows that \( G_{\tilde{q}} \subseteq H \) and so \( G_{\tilde{q}} = H_{\tilde{q}} \).
	Hence, we have found \( \tilde{q} \) and \( \nu \) satisfying~\eqref{eq:cotangentBundle:orbitTypes:relation}.
	
	In the converse direction, let \( H \) be a stabilizer subgroup of \( Q \), \( (\tilde{H}) \) be an orbit type of \( Q \) fulfilling \( (\tilde{H}) \leq (H) \) and \( L \) be a stabilizer subgroup of the \( H \)-action on \( (\LieA{g} \slash \LieA{h})^* \).
	Let \( q \in Q \) be such that \( G_q = H \).
	Choose a slice \( S \) at \( q \).
	Since \( (\tilde{H}) \leq (H) \), the frontier condition and \iref{i::slice:SliceDefLocallyProduct} of \cref{defn:slice:slice} imply that there exists \( \tilde{q} \in S \) with \( (H_{\tilde{q}}) = (G_{\tilde{q}}) = (\tilde{H}) \).
	Choose an \( \AdAction_H \)-invariant complement \( \LieA{m} \) of \( \LieA{h} \) in \( \LieA{g} \) and \( \nu \in \LieA{m}^* \) with \( H_\nu = L \).
	Since \( \Phi \) is equivariant, the point \( p = \Phi(\equivClass{e, (\nu, (\tilde{q}, 0))}) \) has stabilizer 
	\begin{equation}
		G_p = H_{\tilde{q}} \intersect H_\nu,
	\end{equation}
	which completes the proof.
\end{proof}

\begin{coro}
	\label{prop:cotangentBundle:orbitTypeLevelSetSameAsBase}
	Under the assumptions of \cref{prop:cotangentBundle:singularSympRed}, the set of orbit types of \( P = J^{-1}(0) \) with respect to the lifted \( G \)-action coincides with the set of orbit types for the \( G \)-action on \( Q \), that is, a subgroup \( K \) is the stabilizer \( K \) of some \( p \in P \) if and only if there exists a point \( \tilde{q} \in Q \) such that \( G_{\tilde{q}} = K \).
\end{coro}
\begin{proof}
	One direction is clear, as the zero section of \( \CotBundle Q \) has the same orbit types as \( Q \) and as it is contained in the zero level set of \( J \).
	The conclusion in the converse direction follows by similar arguments as in the proof of \cref{prop:cotangentBundle:orbitTypes}.
	Indeed, \( J(p) = 0 \) implies that \( \nu = 0 \) by~\eqref{eq:cotangentBundle:normalFormMomentumMap} and hence \( G_{\tilde{q}} = K \) by~\eqref{eq:cotangentBundle:orbitTypes:relation}.
\end{proof}

\begin{lemma}
	\label{prop:cotangentBundle:orbitTypeStrataAreManifolds}
	Under the assumptions of \cref{prop:cotangentBundle:singularSympRed} the following holds.
	For every orbit type \( (K) \) of \( \CotBundle Q \), the orbit type stratum \( (\CotBundle Q)_{(K)} \) is a submanifold of \( \CotBundle Q \) and the quotient \( (\CotBundle Q)_{(K)} \slash G \) is a Fréchet manifold.
\end{lemma}
\begin{proof}
	Let \( p \in (\CotBundle Q)_{(K)} \) and denote its base point by \( q \in Q \).
	By \cref{prop:cotangentBundle:simpleNormalForm}, it is enough to show that
	\begin{equation}
		\label{eq:cotangentBundle:orbitTypeStrataLocalModel}
		\left(G \times_{G_q} (\LieA{m}^* \times \CotBundle S)\right)_{(G_p)} = G \times_{G_q} \left( (\LieA{m}^* \times \CotBundle S)_{(G_p)} \right)
	\end{equation}
	is a submanifold of \( G \times_{G_q} (\LieA{m}^* \times \CotBundle S) \).
	In this expression, \( (G_p) \) still denotes the conjugacy class of \( G_p \) in \( G \) and not in \( G_q \).
	Since the action is linear and \( G_q \) compact, \cref{prop:slice:sliceTheoremLinearAction} implies that the \( G_q \)-action on \( \LieA{m}^* \times \CotBundle S \) admits a slice at every point.
	The existence of these slices and \cref{prop::compactLieSubgroup:conjugatedSubgroupEqual} show that \( (\LieA{m}^* \times \CotBundle S)_{(G_p)} \) is a submanifold of \( \LieA{m}^* \times \CotBundle S \) (this follows from a similar argument as in the proof of \cref{prop:slice:orbitTypeSubsetIsSubmanifold}).
	Finally, the \( G \)-quotient is a smooth manifold, because it is locally identified with \( (\LieA{m}^* \times \CotBundle S)_{(G_p)} \slash G_q\).
\end{proof}

As we will see next, the existence of lifted slices will reduce the study of the local structure of the strata in the reduced space to the analysis of the linear action of a compact group on a symplectic vector space.
It is quite remarkable that, in this linear setting, the usual finite-dimensional result about symplectic strata directly generalizes to the infinite-dimensional realm without any further assumptions.
\begin{lemma}
	\label{prop:singularReduction:linear}
	Let \( (Y, \omega) \) be a symplectic Fréchet space and let \( G \) be a compact Lie group that acts linearly and symplectically on \( Y \).
	Then, the \( G \)-action has an unique (up to a constant) equivariant momentum map \( J: Y \to \LieA{g}^* \) given by
	\begin{equation}
		\label{eq:singularReduction:linear:momentumMap}
		\kappa(J(y), \xi) = \frac{1}{2} \omega(y, \xi \ldot y), \quad y \in Y, \xi \in \LieA{g}.
	\end{equation}	
	Moreover, for every orbit type \( (H) \), the subset \( Y_{(H), 0} \defeq Y_{(H)} \intersect J^{-1}(0) \) is a smooth submanifold of \( Y \) and there exists a unique smooth manifold structure on \( \check{Y}_{(H), 0} \defeq Y_{(H), 0} \slash G \) such that the natural projection \( \pi_{(H)}: Y_{(H), 0} \to \check{Y}_{(H), 0} \) is a smooth submersion.
	Furthermore, \( \check{Y}_{(H), 0} \) carries a symplectic form \( \check{\omega}_{(H)} \) uniquely determined by
	\begin{equation+}
		\label{eq:singularReduction:linear:reducedSymplectic}
		\pi_{(H)}^* \check{\omega}_{(H)} = \restr{\omega}{Y_{(H), 0}}.
		\qedhere
	\end{equation+}
\end{lemma}
\begin{proof}
	We only sketch the main ideas and refer the reader to \parencite{DiezThesis} for the complete proof.
	To prove that~\eqref{eq:singularReduction:linear:momentumMap} defines an equivariant momentum map is a simple exercise which we leave to the reader.
	The proof of the remaining part is based on the fact that a linear action of a compact Lie group on a Fréchet space admits a slice at every point, see \cref{prop:slice:sliceTheoremLinearAction}, and on the Bifurcation \cref{prop:bifurcationLemma}.
	Using these tools, one shows that \( Y_{(H), 0} \) near a point \( y \in Y_H \intersect J^{-1}(0) \) is modeled on \( G \times V_H \), where \( V = \TBundle_y S \intersect \ker \tangent_y J \) and \( S \) is a slice at \( y \).
	Correspondingly, \( V_H \) is the model space of \( \check{Y}_{(H), 0} \) near \( \equivClass{y} \).
	Finally,~\eqref{eq:singularReduction:linear:reducedSymplectic} uniquely defines a closed \( 2 \)-form \( \check{\omega}_{(H)} \) on \( \check{Y}_{(H), 0} \) because \( \pi_{(H)} \) is a submersion.	
	By \cref{prop:bifuractionLemma:symplecticSubspace}, \( V_H \) is a symplectic subspace of \( (Y, \omega) \) and so \( \check{\omega}_{(H)} \) is symplectic.
\end{proof}

With this preparation at hand, we now pass to the cotangent bundle case.
\begin{lemma}
	\label{prop:cotangentBundle:orbitTypeMomentumLevelSetIsManifold}
	Under the assumptions of \cref{prop:cotangentBundle:singularSympRed} the following holds.
	For every orbit type \( (K) \) of the lifted \( G \)-action, the set
	\begin{equation}
		P_{(K)} \defeq (\CotBundle Q)_{(K)} \intersect J^{-1}(0) = (J^{-1}(0))_{(K)}
	\end{equation}
	is a smooth submanifold of \( (\CotBundle Q)_{(K)} \).
	Moreover, there exists a unique smooth manifold structure on
	\begin{equation}
		\check{P}_{(K)} \defeq P_{(K)} \slash G
	\end{equation}
	such that the natural projection \( \pi_{(K)}: P_{(K)} \to \check{P}_{(K)} \) is a smooth submersion.
	%For \( \equivClass{p} \in \check{P}_{(K)} \) with \( p \in P_{K} \), we have
	%\begin{equation}
	%	\TBundle_{\equivClass{p}} \check{P}_{(K)} = \tangent_p \pi_{(K)} (A_{K} \intersect \ker \tangent_p J),
	%\end{equation}
	%where \( A \) is a closed subspace of \( \TBundle_p (\CotBundle Q) \) complementary to \( \LieA{g} \ldot p \).
\end{lemma}
\begin{proof}
	Let \( p \in P_{(K)} \) and let \( q \in Q \) be its base point.
	Under the local diffeomorphism \( \Phi \) established in \cref{prop:cotangentBundle:simpleNormalForm}, the set \( P_{(K)} \) is identified with
	\begin{equation}
		\left(G \times_{G_q} (\LieA{m}^* \times \CotBundle S)\right)_{(K)} \intersect (J \circ \Phi)^{-1}(0) 
			= G \times_{G_q} \left( \set{0} \times (J^{-1}_{G_q}(0))_{(K)} \right).
	\end{equation}
	This equality is a direct consequence of~\eqref{eq:cotangentBundle:orbitTypeStrataLocalModel} and~\eqref{eq:cotangentBundle:normalFormMomentumMap}. 
	By \cref{prop:singularReduction:linear}, the orbit type subset \( (J^{-1}_{G_q}(0))_{(K)} = (\CotBundle S)_{(K)} \intersect J^{-1}_{G_q}(0) \) is a smooth submanifold of \( \CotBundle S \).
	By the same \lcnamecref{prop:singularReduction:linear}, the quotient space \( (J^{-1}_{G_q}(0))_{(K)} \slash G_q \), which is the model space of \( \check{P}_{(K)} \), is a smooth manifold, too.
\end{proof}

\begin{lemma}
	\label{prop:cotangentBundle:frontierConditionDerivedFromQ}
	Under the assumptions of \cref{prop:cotangentBundle:singularSympRed},	the orbit type decompositions of \( \check{Q} \), \( P \) and \( \check{P} \) satisfy the frontier condition.
\end{lemma}
\begin{proof}
	The quotient \( \check{Q} \) inherits the frontier condition from \( Q \) by \parencite[Theorem~4.6]{DiezSlice}.
	Similarly, \( \check{P} \) inherits the frontier condition from \( P \).
	Let \( p \in P \).
	We have to show that \( p \in \closureSet{P_{(K)}} \) if and only if \( (G_p) \geq (K) \).
	Denote the base point of \( p \) by \( q \), abbreviate \( H \equiv G_q \) and let \( \Phi \) be the local diffeomorphism constructed in \cref{prop:cotangentBundle:simpleNormalForm}.
	As we have seen in the proof of \cref{prop:cotangentBundle:orbitTypeMomentumLevelSetIsManifold}, \( P_{(K)} \) is locally identified with
	\begin{equation}
		 G \times_{H} \left(\set{0} \times (J^{-1}_{H} (0))_{(K)} \right).
	\end{equation}
	Moreover, \( \closureSet{P_{(K)}} \) is locally identified with
	\begin{equation}
		 G \times_{H} \left(\set{0} \times \closureSet{(J^{-1}_{H} (0))_{(K)}} \right),
	\end{equation}
	because the quotient map \( G \times \CotBundle S \to G \times_H \CotBundle S \) is open and \( f^{-1}(\closureSet{A}) = \closureSet{f^{-1}(A)} \) for every open continuous map \( f: Y \to Z \) and every subset \( A \subseteq X \).
	Write \( p = \Phi(\equivClass{e, (0, \alpha)}) \), where \( \alpha \in \CotBundle_q S \) satisfies \( J_{H} (\alpha) = 0 \).
	Since \( G_p = H_\alpha \), it is enough to show that \( \alpha \) lies in the closure of \( (J^{-1}_{H} (0))_{(K)} \) if and only if \( (H_\alpha) \geq (K) \).

	First, suppose that \( \alpha \in \closureSet{(J^{-1}_{H} (0))_{(K)}} \).
	Since \( H \) is compact, there exists a slice at \( \alpha \) and hence a neighborhood of \( \alpha \) in \( \CotBundle S \) such that every point in this neighborhood has a stabilizer subconjugate to \( H_\alpha \), see \parencite[Proposition~2.4]{DiezSlice}.
	However, by assumption, \( (J^{-1}_{H} (0))_{(K)} \) has to intersect this neighborhood and thus \( (K) \leq (H_\alpha) \).

	For the converse direction, we first need to establish a result about the orbit type decomposition of \( S \).
	Since the orbit type stratification of \( Q \) satisfies the frontier condition, we have \( \bigUnion_{(H) \geq (K)} Q_{(H)} \subseteq \closureSet{Q_{(K)}} \) for every orbit type \( (K) \) of \( Q \).
	We will show now that we get a similar approximation property in the slice.
	For this purpose, let \( (H) > (K) \) be orbit types of the action on \( Q \) and let \( s \in S_{(H)} \).
	For every open neighborhood \( V \) of \( s \) in \( S \), the image \( \chi^S(U \times V) \) under the local slice diffeomorphism \( \chi^S: U \times S \to Q \) is an open neighborhood of \( s \) in \( Q \).
	Since \( Q_{(H)} \subseteq \closureSet{Q_{(K)}} \), the intersection
	\begin{equation}
		\chi^S(U \times V) \intersect Q_{(K)} = \chi^S(U \times V_{(K)})
	\end{equation}
	is non-empty.
	Thus, \( V_{(K)} \) is non-empty and we have shown that \( S_{(H)} \subseteq \closureSet{S_{(K)}} \) for all orbit types \( (H) > (K) \).
	
	Now, suppose that \( (K) \) is an orbit type of \( P \) with \( (K) \leq (H_\alpha) \).
	By \cref{prop:cotangentBundle:orbitTypeLevelSetSameAsBase}, \( (K) \) is also an orbit type of the \( G \)-action on \( Q \).
	We will show that every open neighborhood \( W \) of \( \alpha \) in \( \CotBundle_q S \) has non-empty intersection with \( (J^{-1}_H (0))_{(K)} \). 
	Note that the momentum map \( J_H \) vanishes on the whole fiber \( \CotBundle_q S \), because \( q \) has stabilizer \( H \).
	Thus, it suffices to show that \( \alpha \) lies in the closure of \( (\CotBundle_q S)_{(K)} \).
	Since \( \CotBundle Q \) is \( G \)-equivariantly diffeomorphic to \( \TBundle Q \), there exists an \( H \)-equivariant diffeomorphism of \( \CotBundle_q S \) and \( \TBundle_q S \).
	Let \( X \in \TBundle_q S \) be the image of \( \alpha \) under this diffeomorphism.
	Since \( S \) is (diffeomorphic to) an open subset of \( \TBundle_q S \), there exists a non-zero \( r \in \R \) such that \( r X \in S \).
	Note that scaling by \( r \) is an \( H \)-equivariant diffeomorphism of \( \TBundle_q S \).
	In particular, the stabilizer of \( rX \) coincides with \( H_{\alpha} \).
	Thus, in summary, we have reduced the problem to showing that \( rX \) lies in the closure of \( (\TBundle_q S)_{(K)} \).
	But, as \( (H_{rX}) = (H_\alpha) \geq (K) \), we have
	\begin{equation}
		rX \in \closureSet{S_{(K)}} \subseteq \closureSet{(\TBundle_q S)_{(K)}}\, ,
	\end{equation}
	using the approximation property \( S_{(H_\alpha)} \subseteq \closureSet{S_{(K)}} \).
\end{proof}

\begin{remark}
	In the paper \parencite{PerlmutterRodriguez-OlmosSousa-Dias2007} it was silently taken for granted that the decomposition of \( Q \) into orbit types always satisfies the frontier condition.
	However, in examples, there may be an orbit type whose closure contains some but not all fixed points; and hence the frontier condition is violated in these cases (see \eg, \parencite[Example~17 and Remark~13]{CrainicMestre2017}).
\end{remark}

It remains to show the last point in \cref{prop:cotangentBundle:singularSympRed}.
\begin{lemma}
	Assume that every orbit is symplectically closed, that is, the symplectic double orthogonal \( (\LieA{g} \ldot p)^{\omega \omega} \)coincides with \( \LieA{g} \ldot p \).
	Then, for every orbit type \( (K) \), there exists a symplectic form \( \check{\omega}_{(K)} \) on \( \check{P}_{(K)} \) uniquely determined by
	\begin{equation+}
		\label{eq:cotangentBundle:symplecticStrata:symplecticFormDescends}
		\pi_{(K)}^* \check{\omega}_{(K)} = \restr{\omega}{P_{(K)}}.
		\qedhere
	\end{equation+}
\end{lemma}
\begin{proof}
	Let \( p \in P_{(K)} \).
	By the Bifurcation \cref{prop:bifurcationLemma}, \( (\LieA{g} \ldot p)^\omega = \ker \tangent_p J \).
	Since \( \TBundle_p P_{(K)} \subseteq \ker \tangent_p J \), we have \( \omega(\xi \ldot p, Y) = 0 \) for all \( \xi \in \LieA{g}, Y \in \TBundle_p P_{(K)} \).
	Thus, the restriction of \( \omega \) to \( P_{(K)} \) is a \( G \)-invariant form horizontal with respect to the \( G \)-orbits and, therefore, it descends to a \( 2 \)-form \( \check{\omega}_{(K)} \) on \( \check{P}_{(K)} \).
	By construction, \( \check{\omega}_{(K)} \) satisfies~\eqref{eq:cotangentBundle:symplecticStrata:symplecticFormDescends}.
	Moreover, it is uniquely determined by this equation, because \( \pi_{(K)} \) is a surjective submersion.
	Taking the exterior differential of~\eqref{eq:cotangentBundle:symplecticStrata:symplecticFormDescends} shows that \( \check{\omega}_{(K)} \) is closed.
	It remains to show that \( \check{\omega}_{(K)} \) is non-degenerate.
	For every choice of a topological complement \( A \) of \( \LieA{g} \ldot p \) in \( \TBundle_p (\CotBundle Q) \), the projection \( \tangent_p \pi_{(K)} \) yields an isomorphism of \( A_{G_p} \intersect \ker \tangent_p J \) with \( \TBundle_{\equivClass{p}} \check{P}_{(K)} \).
	Under this isomorphism, \( \check{\omega}_{(K)} \) coincides with the restriction of \( \omega \) to \( A_{G_p} \intersect \ker \tangent_p J \).
	\Cref{prop:bifuractionLemma:symplecticSubspace} shows that this restricted form is symplectic and thus \( \check{\omega}_{(K)} \) is non-degenerate.
\end{proof}

For the special case, when the \( G \)-action on \( Q \) has only one orbit type\footnote{This assumption includes, of course, also the case of a free action.}, we obtain the infinite-dimensional counterpart to the well-known cotangent bundle reduction theorem for one orbit type \parencite[Theorem~1]{EmmrichRoemer1990}.
\todo{Added footnote and changed proof to reflect that there is no canonical choice for the cotangent bundle of \( Q \slash G \).}
\begin{thm}
	\label{prop:cotangentBundle:oneOrbitTypeReducedSpaceCotangentBundle}
	In the setting of \cref{prop:cotangentBundle:singularSympRed}, assume additionally that the \( G \)-action on \( Q \) has only one orbit type.
	Then, \( \check{Q} \equiv Q \slash G \) is a smooth manifold, and \( \check{P} \equiv J^{-1}(0) \slash G \) is symplectomorphic\footnote{As there is no canonical choice of the cotangent bundle \( \CotBundle \check{Q} \), `symplectomorphic' should be understood in the sense that there exists a non-degenerate pairing of \( \check{P} \) with \( \TBundle \check{Q} \) and that the reduced symplectic form on \( \check{P} \) coincides with the canonical symplectic structure defined by this pairing as in \cref{Cot-Bundle}.} to \( \CotBundle \check{Q} \) with its canonical symplectic structure. 	
\end{thm}  
\begin{proof}
	Let \( (H) \) denote the orbit type of \( Q \).
	By assumption, we have \( Q_{(H)} = Q \) and thus the existence of slices ensures that the quotient \( \check{Q} = Q \slash G \) is a smooth manifold, see \cref{prop:slice:orbitTypeSubsetIsSubmanifold}.
	\Cref{prop:cotangentBundle:singularSympRed} entails that every point of \( J^{-1}(0) \) has orbit type \( (H) \) and that, moreover, the reduced space \( \check{P} = J^{-1}(0) \slash G \) is a smooth manifold and carries a closed \( 2 \)-form\footnote{A priori, \( \check{\omega} \) may be degenerate as we have not assumed that the orbits be symplectically closed.} \( \check{\omega} \).
	Let \( \pi: Q \to \check{Q} \) denote the natural projection.
	In finite dimensions, the proof proceeds by showing that the map \( \CotBundle \check{Q} \to \check{P} \) defined by \( \alpha_{\equivClass{q}} \mapsto \equivClass{\pi^* \alpha_{\equivClass{q}}} \) is a symplectomorphism.
	As there is no canonical choice of the cotangent bundle \( \CotBundle \check{Q} \) in our infinite-dimensional setting, we need to dualize the argument.
	First, we claim that \( \check{P} \) is a vector bundle over \( \check{Q} \).
	Let \( q \in Q \) with \( G_q = H \) and choose a slice \( S \) at \( q \).
	In slice coordinates, we have \( Q \isomorph G \times_H S \) and so \( \check{Q} \isomorph S \), because the assumption that \( Q \) has only the orbit type \( (H) \) implies \( S = S_H \).
	On the other hand, in the proof of \cref{prop:cotangentBundle:orbitTypeMomentumLevelSetIsManifold} we have seen that \( \check{P} \) is locally identified with \( J^{-1}_H (0)_{(H)} \slash H \), where \( J_H: \CotBundle S \to \LieA{h}^* \) is the momentum map for the \( H \)-action on \( \CotBundle S \).
	Now, \( S = S_H \) implies that \( J^{-1}_H (0)_{(H)} \slash H \) coincides with \( \CotBundle S \), and that the projection \( \check{P} \to \check{Q} \) corresponds to the natural projection \( \CotBundle S \to S \).
	This shows that \( \check{P} \) is a vector bundle over \( \check{Q} \), indeed.
	Second, there is a natural pairing of \( \check{P} \) and \( \TBundle \check{Q} \) defined by
	\begin{equation}
		\dualPair{\equivClass{\beta}}{Y_{\equivClass{q}}} \defeq \dualPair{\beta_q}{\tilde{Y}_q},
	\end{equation}
	where \( q \in Q \) with \( \pi(q) = \equivClass{q} \), \( \beta_q \in J^{-1}(0) \intersect \CotBundle_q Q \), and \( \tilde{Y}_q \in \TBundle_q Q \) satisfies \( \tangent_q \pi (\tilde{Y}_q) = Y_{\equivClass{q}} \).
	The fact that \( \beta_q \in J^{-1}(0) \) ensures that this pairing is well-defined.
	Moreover, in the local coordinates discussed above, the pairing correspond to the natural pairing of \( \CotBundle S \) with \( \TBundle S \), which shows that it is smooth and non-degenerate.
	Thus, \( \check{P} \) serves as a model for the cotangent bundle of \( \check{Q} \).
	Finally, by \cref{prop:cotangentBundle:normalFormSymplecticForm}, under these identifications, the canonical symplectic form defined by the pairing (\cf \cref{Cot-Bundle}) and the reduced symplectic form \( \check{\omega} \) both are equal to the canonical symplectic form on \( \CotBundle S \).
\end{proof}

\subsection{Secondary stratification}
Examples of singular cotangent bundle reduction (such as the ones coming from lattice gauge theory \parencite{FischerRudolphSchmidt2007}) show that the natural projection from the reduced phase space \( \check{P} = \CotBundle Q \sslash_0 G \) to the reduced configuration space \( \check{Q} = Q \slash G \) is not a morphism of stratified spaces, \ie, there exist strata in \( \check{P} \) that project onto different strata in \( \check{Q} \).
In the finite-dimensional context, \textcite{PerlmutterRodriguez-OlmosSousa-Dias2007} have refined the symplectic stratification of \( \check{P} \) in such a way that the projection \( \check{P} \to  \check{Q} \) becomes a morphism of stratified spaces.
This so-called secondary stratification has also the advantage of identifying certain strata in the reduced phase space as cotangent bundles.
In this section, we construct this refined stratification in our infinite-dimensional setting.
In particular, we show that the secondary strata are submanifolds of the symplectic strata in \( \hat{P} \) (for finite-dimensional manifolds this is shown in \parencite[Theorem~7]{PerlmutterRodriguez-OlmosSousa-Dias2007}, but the proof there does not directly translate to the infinite-dimensional setting).
Moreover, we study how the secondary strata behave with respect to the ambient symplectic structure and investigate the frontier condition.
This discussion culminates in \cref{prop:cotangentBundle:singularCotangentBundleRed}.
To get a feeling for the secondary stratification, the reader might find it instructive to consult the example of the harmonic oscillator discussed in \cref{ex:cotangentBundle:harmonicOss} while the general theory is developed.
In this subsection, we continue in the setting of \cref{prop:cotangentBundle:singularSympRed}.

In \parencite{PerlmutterRodriguez-OlmosSousa-Dias2007}, a secondary stratification of the reduced phase space \( \check{P} = \CotBundle Q \sslash_0 G \) has been introduced.
Inspired by this construction, we introduce the following refined decomposition of \( \CotBundle Q \).
For orbit types \( (K) \) and \( (H) \) of \( \CotBundle Q \) and \( Q \), respectively, consider the subset
\begin{equation}
	\seam{(\CotBundle Q)}^{(K)}_{(H)} \defeq \set{p \in \CotBundle_q Q \given q \in Q_{(H)}, p \in (\CotBundle Q)_{(K)}}
\end{equation}
of the orbit type stratum \( (\CotBundle Q)_{(K)} \).
Since the projection \( \CotBundle Q \to Q \) is \( G \)-equivariant, \( \seam{(\CotBundle Q)}^{(K)}_{(H)} \) is non-empty only if \( (K) \leq (H) \).
Moreover, the union of \( \seam{(\CotBundle Q)}^{(K)}_{(H)} \) over all orbit types \( (H) \) fulfilling this condition yields the orbit type stratum \( (\CotBundle Q)_{(K)} \).
Note that the stabilizer \( G_p \) of \( p \in \CotBundle_q Q \) under the lifted \( G \)-action on \( \CotBundle Q \) coincides with the stabilizer \( [G_q]_p \) of \( p \) under the dual isotropy action of \( G_q \) on the fiber \( \CotBundle_q Q \).
Whence, we equivalently have
\begin{equation}
	\seam{(\CotBundle Q)}^{(K)}_{(H)} \defeq \set{p \in \CotBundle_q Q \given q \in Q_{(H)}, [G_q]_p \sim K}.
\end{equation}

\begin{lemma}
	\label{prop:cotangentBundle:secondaryStrataAreManifolds}
	Under the assumptions of \cref{prop:cotangentBundle:singularSympRed} the following holds.
	Let \( (K) \) be an orbit type of \( \CotBundle Q \).
	Then, for every orbit type \( (H) \) of \( Q \) fulfilling \( (H) \geq (K) \), the sets \( \seam{(\CotBundle Q)}^{(K)}_{(H)} \) and \( \seam{(\CotBundle Q)}^{(K)}_{(H)} \slash G \) are submanifolds of \( (\CotBundle Q)_{(K)} \) and \( (\CotBundle Q)_{(K)} \slash G \), respectively.
\end{lemma}
We call the sets \( \seam{(\CotBundle Q)}^{(K)}_{(H)} \) the \emphDef{secondary strata} and the decomposition of \( \CotBundle Q \) into these secondary strata is referred to as the \emphDef{secondary orbit type stratification}.
As far as we know, the submanifold structure of \( \seam{(\CotBundle Q)}^{(K)}_{(H)} \) and \( \seam{(\CotBundle Q)}^{(K)}_{(H)} \slash G \) is novel even for the finite-dimensional case.
\begin{proof}
	Let \( p \in \seam{(\CotBundle Q)}^{(K)}_{(H)} \) and denote its base point by \( q \in Q \).
	Without loss of generality we may assume that \( G_p = K \) and \( G_q = H \) with \( G_p \subseteq G_q \).
	By \cref{prop:cotangentBundle:simpleNormalForm}, it is enough to show that the corresponding subset
	\begin{equation}
		\seam{\left(G \times_{G_q} (\LieA{m}^* \times \CotBundle S)\right)}^{(G_p)}_{(G_q)} \subseteq G \times_{G_q} (\LieA{m}^* \times \CotBundle S)
	\end{equation}
	is a submanifold.
	By definition, \( S \) is diffeomorphic to an open subset of a Fréchet space \( X \) and, hence, we may identify \( \CotBundle S \) with \( S \times X^* \).
	By \cref{prop::compactLieSubgroup:conjugatedSubgroupEqual}, for every point \( s \) in the slice, \( G_s \) is conjugate to \( G_q \) if and only if \( G_s = G_q \).
	We thus find
	\begin{equation}
		\seam{\left(G \times_{G_q} (\LieA{m}^* \times \CotBundle S)\right)}^{(G_p)}_{(G_q)} = G \times_{G_q} \left( (\LieA{m}^* \times X^*)_{(G_p)} \times S_{G_q} \right).
	\end{equation}
	In this expression, \( (G_p) \) clearly denotes the conjugacy class of \( G_p \) in \( G \) and not in \( G_q \).
	Now, since \( G_q \) is compact, \( (\LieA{m}^* \times X^*)_{(G_p)} \) is a submanifold of \( \LieA{m}^* \times X^* \) (the proof follows by the same arguments as in the proof of \cref{prop:cotangentBundle:orbitTypeStrataAreManifolds}).
	Finally, the \( G \)-quotient is again a smooth manifold, because it is locally identified with \( (\LieA{m}^* \times X^*)_{(G_p)} \slash G_q \times S_{G_q} \).
\end{proof}
For the study of the interaction of the secondary orbit type stratification with the momentum map geometry we need the following basic result about linear cotangent bundle reduction.
\begin{lemma}
	\label{prop:cotangentBundle:seamManifoldInLinear}
	Let \( \dualPairDot: X^* \times X \to \R \) be a dual pair of Fréchet spaces and let \( G \) be a compact Lie group acting linearly on \( X \) and, by duality, also on \( X^* \).
	Then, the lifted \( G \)-action on \( \CotBundle X \) has an equivariant momentum map \( J: \CotBundle X \to \LieA{g}^* \).
	Moreover, for every orbit type \( (K) \),
	\begin{equation}
		\label{eq:contangentBundle:seamManifoldInLinear:iso}
		\seam{(\CotBundle X)}^{(K)}_{(G)} \intersect J^{-1}(0) \isomorph X_G \times X^*_{(K)}
	\end{equation}
	is a submanifold of \( \CotBundle X \).
\end{lemma}
\begin{proof}
	Under the identification \( \CotBundle X \isomorph X \times X^* \), the canonical \( 1 \)-form takes the form
	\begin{equation}
		\theta_{x, \alpha}(y, \beta) = \dualPair{\alpha}{y}, \qquad x, y \in X, \alpha, \beta \in X^*.
	\end{equation}
	Since \( G \) is compact and, hence, finite-dimensional, the linear action has a momentum map \( J \) defined by
	\begin{equation}
		\kappa(J(x, \alpha), \xi) = \dualPair{\alpha}{\xi \ldot x}, \qquad \xi \in \LieA{g}.
	\end{equation}
	Note that \( J(x, \alpha) = 0 \) if \( x \in X_G \).
	Since, by definition of the secondary strata,
	\begin{equation}
		\seam{(\CotBundle X)}^{(K)}_{(G)} \isomorph X_{(G)} \times X^*_{(K)} = X_G \times X^*_{(K)}
	\end{equation}
	holds, we obtain
	\begin{equation}
		\seam{(\CotBundle X)}^{(K)}_{(G)} \intersect J^{-1}(0) \isomorph X_G \times X^*_{(K)}.
	\end{equation}
	Since \( G \) is compact, the action on \( X^* \) admits a slice at every point according to \cref{prop:slice:sliceTheoremLinearAction} and thus the orbit type manifold \( X^*_{(K)} \) is a submanifold of \( X^* \), see \cref{prop:slice:orbitTypeSubsetIsSubmanifold}.
	Therefore, \( \seam{(\CotBundle X)}^{(K)}_{(G)} \intersect J^{-1}(0) \) is a submanifold of \( \CotBundle X \).
\end{proof}

We now return to the general non-linear setting.
Given two orbit types \( (K) \leq (H) \), following \parencite{PerlmutterRodriguez-OlmosSousa-Dias2007}, we call the set
\begin{equation}
	\seam{P}^{(K)}_{(H)} \defeq \seam{(\CotBundle Q)}^{(K)}_{(H)} \intersect J^{-1}(0)
\end{equation}
a \emphDef{preseam} and the quotient \( \seam{\check{P}}^{(K)}_{(H)} \defeq \seam{P}^{(K)}_{(H)} \slash G \) a \emphDef{seam}.
\begin{lemma}
	\label{prop:cotangentBundle:seamsAreManifolds}
	Under the assumptions of \cref{prop:cotangentBundle:singularSympRed} the following holds.
	For every pair of orbit types \( (H) \geq (K) \), the preseam \( \seam{P}^{(K)}_{(H)} \) is a smooth submanifold of \( \CotBundle Q \) and the seam \( \seam{\check{P}}^{(K)}_{(H)} \) is a smooth submanifold of \( \check{P}_{(K)} \) and of \( (\CotBundle Q)_{(K)} \slash G \).
	Moreover, \( \seam{\check{P}}^{(K)}_{(H)} \) is a smooth fiber bundle over \( \check{Q}_{(H)} \).
\end{lemma}
In the finite-dimensional setting, the smooth structure of \( \seam{\check{P}}^{(K)}_{(H)} \) and its fibration over \( \check{Q}_{(H)} \) has been established in \parencite[Theorem~7]{PerlmutterRodriguez-OlmosSousa-Dias2007}.
To the best of our knowledge, the submanifold structure of \( \seam{P}^{(K)}_{(H)} \) is novel even for the finite-dimensional case.
\begin{proof}
	Let \( p \in \seam{(\CotBundle Q)}^{(K)}_{(H)} \) and denote its base point by \( q \in Q \).
	Let \( S \) be a slice at \( q \) and let \( X \) be the model space of \( S \).
	Using the local diffeomorphism \( \Phi \) of \cref{prop:cotangentBundle:simpleNormalForm} and the isomorphism of \cref{eq:contangentBundle:seamManifoldInLinear:iso}, in a neighborhood of \( p \) we can identify the preseam with the submanifold
	\begin{equation}\begin{split}
		\seam{\left(G \times_{G_q} (\LieA{m}^* \times \CotBundle S)\right)}^{(G_p)}_{(G_q)} \intersect J^{-1}(0) 
			&= G \times_{G_q} \left( \set{0} \times \seam{J^{-1}_{G_q}(0)}^{(G_p)}_{(G_q)} \right)
			\\
			&\isomorph G \times_{G_q} \left( \set{0} \times S_{G_q} \times X^*_{(G_p)} \right)
	\end{split}\end{equation}
	of \( G \times_{G_q} (\LieA{m}^* \times \CotBundle S) \).
	Similarly, the seam \( \seam{\check{P}}^{(K)}_{(H)} \) has locally the same structure as the smooth manifold \( S_{G_q} \times (X^*_{(G_p)} \slash G_q) \).
	Under these identifications, the quotient map \( \seam{\check{P}}^{(K)}_{(H)} \to \check{Q}_{(H)} \) corresponds to the projection onto the first factor and is thus a locally trivial submersion.
\end{proof}

We now come to the interaction of the seams with the symplectic geometry.
For this purpose, denote by \( \check{\omega}^{(K)}_{(H)} \) the restriction of \( \check{\omega}_{(K)} \) to \( \seam{\check{P}}^{(K)}_{(H)} \subseteq \check{P}_{(K)} \).
The injection \( \TBundle (Q_{(H)}) \to \restr{(\TBundle Q)}{Q_{(H)}} \) induces a surjective map \( \pr: \restr{(\CotBundle Q)}{Q_{(H)}} \to \CotBundle (Q_{(H)}) \) and, thereby, a map \( \bar{\pi}: \seam{P}^{(K)}_{(H)} \to \CotBundle (Q_{(H)}) \).
Let \( \bar{\omega}_{(H)} \) denote the canonical symplectic form on \( \CotBundle (Q_{(H)}) \).
With this notation, we can give a characterization of \( \check{\omega}^{(K)}_{(H)} \) similar to the one for the reduced symplectic form.
\begin{lemma}
	\label{prop:cotangentBundle:reducedFormOnSeam}
	The restriction \( \check{\omega}^{(K)}_{(H)} \) of \( \check{\omega}_{(K)} \) to \( \seam{\check{P}}^{(K)}_{(H)} \subseteq \check{P}_{(K)} \) is uniquely characterized by 
	\begin{equation}
		\left(\pi^{(K)}_{(H)}\right)^* \check{\omega}^{(K)}_{(H)} = \bar{\pi}^* \, \bar{\omega}_{(H)},
	\end{equation}
	where \( \pi^{(K)}_{(H)}: \seam{P}^{(K)}_{(H)} \to \seam{\check{P}}^{(K)}_{(H)} \) is the canonical projection.
\end{lemma}
\begin{proof}
	We first note that the restriction of \( \omega \) to \( \restr{(\CotBundle Q)}{Q_{(H)}} \) coincides with the pull-back \( \pr^* \, \bar{\omega}_{(H)} \).
	In fact, the commutative diagram
	\begin{equationcd}
		\CotBundle (Q_{(H)}) \to[d]
			& \restr{(\CotBundle Q)}{Q_{(H)}} \to[l, swap, "\pr"] \to[r] \to[d]
			& \CotBundle Q \to[d]
		\\
		Q_{(H)}
			& Q_{(H)} \to[l, swap, "\id"] \to[r]
			& Q
	\end{equationcd}
	and a straightforward calculation show that the pull-back \( \pr^* \, \bar{\theta}_{(H)} \) of the canonical \( 1 \)-form on \( \CotBundle (Q_{(H)}) \) coincides with the restriction of \( \theta \) to \( \restr{(\CotBundle Q)}{Q_{(H)}} \).
	Now the claim follows by chasing along the following commutative diagram
	\begin{equationcd}
		\CotBundle (Q_{(H)})
			& \restr{(\CotBundle Q)}{Q_{(H)}} \to[l, swap, "\pr"] \to[r]
			& \CotBundle Q
		\\
			& \seam{P}^{(K)}_{(H)} \to[lu, swap, "\bar{\pi}"] \to[u] \to[r] \to[d, "\pi^{(K)}_{(H)}"]
			& P_{(K)} \to[d] \to[u]
		\\
			& \seam{\check{P}}^{(K)}_{(H)} \to[r]
			& \check{P}_{(K)}
	\end{equationcd}
\end{proof}

The construction above provides additional insight into the structure of the seam \( \seam{\check{P}}^{(H)}_{(H)} \).
To see this, note that \( \bar{\pi} \) takes values in the zero level set of the momentum map
\begin{equation}
	\bar{J}_{(H)}: \CotBundle (Q_{(H)}) \to \LieA{g}^* \, .
\end{equation}
Moreover, \( \bar{\pi} \) is \( G \)-equivariant and thus descends to a map 
\begin{equation}
	\check{\bar{\pi}}: \seam{\check{P}}^{(K)}_{(H)} \to \bar{J}_{(H)}^{-1}(0) \slash G.
\end{equation}
By \cref{prop:cotangentBundle:oneOrbitTypeReducedSpaceCotangentBundle}, the target space \( \bar{J}_{(H)}^{-1}(0)_{(H)} \slash G \) is symplectomorphic to \( \CotBundle (\check{Q}_{(H)}) \).
\begin{prop}
	\label{prop:cotangentBundle:restrictionSymplecticForm}
	Under the assumptions of \cref{prop:cotangentBundle:singularSympRed}, for every orbit type \( (H) \), the restriction of \( \check{\omega}_{(H)} \) to \( \seam{\check{P}}^{(H)}_{(H)} \) is symplectic and \( \check{\bar{\pi}} \) is a symplectomorphism\footnote{As in \cref{prop:cotangentBundle:oneOrbitTypeReducedSpaceCotangentBundle}, `symplectomorphic' should be understood in the sense that there exists a non-degenerate pairing of \( \seam{\check{P}}^{(H)}_{(H)} \) with \( \TBundle (\check{Q}_{(H)}) \) and that the reduced symplectic form on \( \seam{\check{P}}^{(H)}_{(H)} \) coincides with the canonical symplectic structure defined by this pairing as in \cref{Cot-Bundle}.} between \( \seam{\check{P}}^{(H)}_{(H)} \) and \( \CotBundle (\check{Q}_{(H)}) \) with its canonical symplectic structure\footnote{One might expect that \( \seam{(\CotBundle Q)}^{(H)}_{(H)} \) is symplectomorphic to \( \CotBundle (Q_{(H)}) \). However, simple examples like \( G = \SOGroup(n) \) acting on \( \R^n \) show that this is not the case.}.
\end{prop}
We call \( \seam{\check{P}}^{(H)}_{(H)} \) the \emphDef{principal seam}.
\begin{proof}
	As we have seen in the proof of \cref{prop:cotangentBundle:seamsAreManifolds}, the seam \( \seam{\check{P}}^{(H)}_{(H)} \) is locally identified with \( S_H \times (X^*_{(H)} \slash H) = S_H \times X^*_H \isomorph \CotBundle (S_H) \).
	On the other hand, \( \CotBundle (\check{Q}_{(H)}) \) is locally equivalent to \( \CotBundle (S_H) \), see the proof of \cref{prop:cotangentBundle:oneOrbitTypeReducedSpaceCotangentBundle}.
	It is straightforward to see that, in these coordinates, \( \check{\bar{\pi}} \) is the identity map on \( \CotBundle (S_H) \) and hence a diffeomorphism.
	Finally, by \cref{prop:cotangentBundle:reducedFormOnSeam}, \( \check{\bar{\pi}} \) intertwines the closed \( 2 \)-form \( \check{\omega}^{(H)}_{(H)} \) with the canonical symplectic form on \( \CotBundle (\check{Q}_{(H)}) \).
	Since \( \check{\bar{\pi}} \) is a diffeomorphism, \( \check{\omega}^{(H)}_{(H)} \) is symplectic.
\end{proof}

In finite dimensions, one can show that the seams \( \seam{\check{P}}^{(K)}_{(H)} \) are coisotropic with respect to the reduced symplectic form \( \check{\omega}_{(K)} \) on \( \check{P}_{(K)} \), see \parencite[Corollary~9]{PerlmutterRodriguez-OlmosSousa-Dias2007}.
The proof, however, relies on counting dimensions and thus does not generalize to the infinite-dimensional setting.
A different idea to show that the seams are coisotropic is to use a Witt--Artin decomposition adapted to the cotangent bundle case.
In the finite-dimensional context, such a decomposition was established in \parencite{PerlmutterRodriguez-OlmosSousa-Dias2005}, but its extension to infinite dimensions is not immediate and will be left to future work.

Let us now discuss the frontier condition.
For this purpose, we endow the set of pairs of orbit types \( ((K), (H)) \) satisfying \( (H) \geq (K) \) with the partial ordering
\begin{equation}
 	((K), (H)) \leq ((K'), (H')) \quad \text{if and only if } \quad (K) \leq (K') \text{ and } (H) \leq (H').
\end{equation}
\begin{lemma}
	In the setting of \cref{prop:cotangentBundle:singularSympRed}, for every pair of orbit types \( (K) \) and \( (H) \) fulfilling \( (H) \geq (K) \), we have
	\begin{equation}
		\closureSet{\seam{\check{P}}^{(K)}_{(H)}} = \bigUnion_{((K'), (H')) \geq ((K), (H))} \seam{\check{P}}^{(K')}_{(H')}\, ,
	\end{equation}
	so that the decomposition of \( \check{P} \) into seams satisfies the frontier condition. 
\end{lemma}
\begin{proof}
	Since the orbit type decomposition of \( Q \) satisfies the frontier condition, the orbit type decomposition of \( P \) and \( \check{P} \) share this property by \cref{prop:cotangentBundle:frontierConditionDerivedFromQ}.
	Let \( \seam{\check{P}}^{(K')}_{(H')} \) be a seam that has a non-empty intersection with the closure of \( \seam{\check{P}}^{(K)}_{(H)} \).
	In particular, \( \seam{\check{P}}^{(K')}_{(H')} \) intersects the closure of \( \check{P}_{(K)} \) and thus \( (K') \geq (K) \) as the orbit type decomposition satisfies the frontier condition.
	Since the canonical projection \( \check{P} \to \check{Q} \) is continuous, a similar argument shows that \( (H') \geq (H) \).

	For the converse direction, let \( \seam{\check{P}}^{(K)}_{(H)} \) and \( \seam{\check{P}}^{(K')}_{(H')} \) be seams with \( (K') \geq (K) \) and \( (H') \geq (H) \).
	We have to show that \( \seam{\check{P}}^{(K')}_{(H')} \) lies in the closure of \( \seam{\check{P}}^{(K)}_{(H)} \).
	Let \( \equivClass{p} \in \seam{\check{P}}^{(K')}_{(H')} \) and choose \( p \in \seam{P}^{(K')}_{(H')} \).
	Denote the base point of \( p \) by \( q \).
	We will show that every neighborhood of \( p \) in \( P \) has a non-trivial intersection with the preseam \( \seam{P}^{(K)}_{(H)} \).
	Since this is a local question, it is enough to consider it in a tubular neighborhood of the form \( G \times_{G_q} (\LieA{m}^* \times \CotBundle S) \).
	That is, it is enough to show that every neighborhood of \( \restr{p}{\TBundle_q S} \in \CotBundle_q S \) in \( J^{-1}_{G_q}(0) \) contains a point \( (s, \alpha) \in S \times X^* \isomorph \CotBundle S \) with \( (G_s) = (H) \) and \( (G_\alpha) = (K) \).
	The existence of such a point follows from the fact that \( P \) and \( Q \) satisfy the frontier condition.
	Indeed, since \( (H) \leq (H') = (G_q) \), the frontier condition implies that every neighborhood of \( q \) in \( Q \) contains a point \( q' \) such that \( (G_{q'}) = (H) \).
	Without loss of generality, we may assume that we work in slice coordinates and that there are, thus, \( g \in G \) and \( s \in S \) with \( q' = g \cdot s \).
	By equivariance of the stabilizer, we have \( (G_s) = (G_{q'}) = (H) \).
	The construction of \( \alpha \in J^{-1}_{G_q}(0) \) with \( (G_\alpha) = (K) \) follows from similar arguments using the frontier condition for \( \check{P} \).
\end{proof}

\begin{coro}
	Under the assumptions of \cref{prop:cotangentBundle:singularSympRed}, for every pair of orbit types \( (K) \) and \( (H) \) with \( (H) \geq (K) \), the closure of \( \seam{\check{P}}^{(K)}_{(H)} \) in \( \check{P}_{(K)} \) is the union of \( \seam{\check{P}}^{(K)}_{(H')} \) over all orbit types \( (H') \geq (H) \).
	In particular, the decomposition of \( \check{P}_{(K)} \) into seams satisfies the frontier condition. 
\end{coro}

Let us summarize.
\begin{thm}[Singular cotangent bundle reduction at zero]
	\label{prop:cotangentBundle:singularCotangentBundleRed}
	Let \( Q \) be a Fréchet \( G \)-manifold.
	Assume that the \( G \)-action on \( Q \) is proper, that it admits at every point a slice compatible with the cotangent bundle structures and that the decomposition of \( Q \) into orbit types satisfies the frontier condition.
	Moreover, assume that \( \CotBundle Q \) is a Fréchet vector bundle, which is \( G \)-equivariantly isomorphic to \( \TBundle Q \), and that the lifted action on \( \CotBundle Q \), endowed with its canonical symplectic form \( \omega \), has a momentum map \( J \).
	Assume, additionally, that every orbit is symplectically closed, that is, the symplectic double orthogonal \( (\LieA{g} \ldot p)^{\omega \omega} \) coincides with \( \LieA{g} \ldot p \) for all \( p \in J^{-1}(0) \).
	Then the following holds:
	\begin{thmenumerate}
		\item
			The set of orbit types of \( J^{-1}(0) \) with respect to the lifted \( G \)-action coincides with the set of orbit types for the \( G \)-action on \( Q \). 
		\item
			The reduced phase space \( \check{P} = J^{-1}(0) \slash G \) is stratified into orbit type subsets \( \check{P}_{(K)} = (J^{-1}(0))_{(K)} \slash G \).
			For every orbit type \( (K) \), the set \( \check{P}_{(K)} \) is a smooth manifold and carries a symplectic form \( \check{\omega}_{(K)} \).
		\item
			Every symplectic stratum \( \check{P}_{(K)} \) is further stratified as
			\begin{equation}
			 	\check{P}_{(K)} = \bigDisjUnion_{(H) \geq (K)} \seam{\check{P}}^{(K)}_{(H)} \,,
			\end{equation}
			where each seam \( \seam{\check{P}}^{(K)}_{(H)} \) is a smooth fiber bundle over \( \check{Q}_{(H)} \).
		\item
			For every orbit type \( (H) \), the principal seam \( \seam{\check{P}}^{(H)}_{(H)} \) endowed with the restriction of the symplectic form \( \check{\omega}_{(H)} \) is symplectomorphic to \( \CotBundle (\check{Q}_{(H)}) \) endowed with its canonical symplectic structure.
		\item
			The decomposition
			\begin{equation}
				\check{P} = \bigDisjUnion_{(H) \geq (K)} \seam{\check{P}}^{(K)}_{(H)}
			\end{equation}
			is a stratification of \( \check{P} \) called the \emph{secondary stratification}.
			Moreover, the projection \( \CotBundle Q \to Q \) induces a stratified surjective submersion \( \check{P} \to \check{Q} \) with respect to the secondary stratification of \( \check{P} \) and the orbit type stratification of \( \check{Q} \).
			\qedhere
	\end{thmenumerate}
\end{thm}

\subsection{Dynamics}
Let us now pass from the kinematic picture presented so far to dynamics.

In finite dimensions, every Hamiltonian \( h \) on a symplectic manifold \( (M, \omega) \) induces a Hamiltonian flow.
This no longer holds in an infinite-dimensional context.
For one thing, the symplectic form on \( M \) is in general only weakly symplectic so that
\begin{equation}
	\omega^\flat: \TBundle M \to \CotBundle M
\end{equation}
is only injective and not surjective.
Hence, a Hamiltonian vector field \( X_h \) associated to the \( 1 \)-form \( \dif h \) may not exist.
Even if \( X_h \) exists, it may not have a unique flow.
The construction of a flow requires the solution of an ordinary differential equation on \( M \), which a priori is not guaranteed to exists and to be unique in infinite dimensions.
For example, in the gauge theory context studied in \cref{sec:yangMillsHiggs} below, existence and uniqueness of the Hamiltonian flow is equivalent to the well-posedness of the Cauchy problem for the Yang--Mills-Higgs theory.

Let us return to the setting of \cref{prop:cotangentBundle:singularCotangentBundleRed}.
\begin{prop}
	Let \( h \) be a \( G \)-invariant Hamiltonian on \( \CotBundle Q \).
	Assume that the associated Hamiltonian vector field \( X_h \) exists and that it has a unique flow \( \flow^h_t \).
	Let \( (K) \) be an orbit type.
	Then,
	\begin{thmenumerate}
		\item
			the flow \( \flow^h_t \) is \( G \)-equivariant and leaves \( P_{(K)} \) invariant and, hence, it projects to a flow \( \check{\flow}^h_t \) on \( \check{P}_{(K)} \),
		\item
			the projected flow \( \check{\flow}^h_t \) is Hamiltonian with respect to the smooth function \( \check{h}_{(K)} \) on \( \check{P}_{(K)} \) defined by
			\begin{equation+}
				\pi_{(K)}^* \check{h}_{(K)} = \restr{h}{P_{(K)}}.
				\qedhere
			\end{equation+}
	\end{thmenumerate}
\end{prop}
\begin{proof}
	Since \( h \) is \( G \)-invariant, the associated Hamiltonian vector field \( X_h \) is invariant, too.
	The calculation
	\begin{equation}
		\difFracAt{}{t}{t} \flow^h_t (g \cdot m) = (X_h)_{g \cdot m} = g \ldot (X_h)_m = \difFracAt{}{t}{t} g \cdot \flow^h_t (m)
	\end{equation}
	shows that the flow \( \flow^h_t \) is \( G \)-equivariant (since, by assumption, it exists and is unique).
	Moreover, the Noether theorem also holds in our infinite-dimensional setting (see \parencite{DiezRatiuAutomorphisms}) and thus the flow \( \flow^h_t \) leaves \( P_{(K)} \) invariant.
	Hence \( \flow^h_t \)  projects onto a flow \( \check{\flow}^h_t \) on \( \check{P}_{(K)} \).
	Denote the induced vector field on \( \check{P}_{(K)} \) by \( \check{X}_{(K)} \).

	Since \( h \) is \( G \)-invariant and \( \pi_{(K)} \) a surjective submersion, \( h \) descends to a smooth function \( \check{h}_{(K)} \) on \( \check{P}_{(K)} \).
	That \( \check{X}_{(K)} \) is  Hamiltonian with respect to \( \check{h}_{(K)} \),  indeed, is verified by a routine calculation.
\end{proof}

The interaction of dynamics with the secondary stratification is more complicated.
The seams are in general not preserved by the Hamiltonian flow.
The following example suggests that, for each orbit type \( (K) \), the singular seams \( \seam{\check{P}}^{(K)}_{(H)} \) with \( (H) > (K) \) stitch together the dynamics in the cotangent bundle \( \seam{\check{P}}^{(K)}_{(K)} \isomorph \CotBundle (\check{Q}_{(K)}) \).
\begin{example}
	\label{ex:cotangentBundle:harmonicOss}
	Consider a two-dimensional isotropic harmonic oscillator, whose coordinates are \( q = (q_1, q_2) \) and the corresponding momenta are \( p = (p_1, p_2) \).
	Consequently, the phase space is \( \CotBundle \R^2 \) and the Hamiltonian of the system is given by
	\begin{equation}
		H(q, p) = \frac{1}{2}\norm{p}^2 + \frac{1}{2}\norm{q}^2.
	\end{equation}
	Note that \( \UGroup(1) \) acting by rotation in the \(( q_1, q_2) \)-plane is a symmetry of \( H \).
	The angular momentum
	\begin{equation}
		J(q,p) = q_1 p_2 - q_2 p_1
	\end{equation}
	is the momentum map for the lift of this action to \( \CotBundle \R^2 \).
	Hence, \( J(q,p) = 0 \) if and only if \( q \) and \( p \) are parallel.
	The \( \UGroup(1) \)-action on \( Q \) is free except at the origin, which has stabilizer \( \UGroup(1) \).
	Consequently, the secondary stratification of \( P = J^{-1}(0) \) is
	\begin{equation}
		P = \underbrace{\set{(0,0)}}_{\seam{P}^{\UGroup(1)}_{\UGroup(1)}}
			\union 
			\underbrace{\set{(0, p \neq 0)}}_{\seam{P}^{\set{e}}_{\UGroup(1)}}
			\union
			\underbrace{\set{(q \neq 0, p) \given q \parallel p}}_{\seam{P}^{\set{e}}_{\set{e}}}.
	\end{equation}
	In order to identify the reduced phase space, consider the map
	\begin{equation}
		\label{eq:cotangentBundle:harmonicOss:Kmap}
		K: \CotBundle \R^2 \to \R^3, \quad (q,p) 
		\mapsto 
			\Vector{E_+ \\ E_- \\ H} 
			\defeq 
			\Vector{\frac{1}{2}\norm{p}^2 - \frac{1}{2}\norm{q}^2 \\ q \cdot p \\ \frac{1}{2}\norm{p}^2 + \frac{1}{2}\norm{q}^2}.
	\end{equation}
	The reason for this notation will become clear in a moment.
	On the way, we note that the combination of \( K \) and \( J \) yields the momentum map for the \( \UGroup(2) \)-symmetry\footnote{
	To be more precise, to arrive at the momentum map of \textcite[\nopp I.3.3]{CushmanBates1997} one has to exchange the coordinates \( q_2 \) and \( p_1 \).
	The \( \UGroup(2) \) symmetry is a good starting point for the qualitative discussion of the dynamics using the energy-momentum map; a topic we will not further develop here.}, see \parencite[\nopp I.3.3]{CushmanBates1997}.
	The Hamiltonian \( H \) is clearly non-negative and a direct calculation shows that \( H^2 - J^2 = E_+^2 + E_-^2 \).
	Hence, the image of \( P = J^{-1}(0) \) under \( K \) is the upper cone (with origin included), see \cref{fig:cotangentBundle:harmonicOss:integralCurves}.
	\begin{figure}[th]
		\centering
		\begin{subfigure}[t]{.5\textwidth}
			\centering
			\begin{tikzpicture}[scale=1]
				% Axis
				\draw[->] (-0.05, 0) -- (2.6, 0) node[right] {$\bar{q}$};
				\draw[->] (0, -2.4) -- (0, 2.6) node[above] {$\bar{p}$};

				\begin{scope}
					\clip (0.05,-2.35) rectangle (2.35,2.35);
					% Curves in light gray
					\foreach \hzero in {0.3,0.5,0.8,1.5,1.9,2.3,2.7,3.1,3.5,3.9}{% hzero = energy
						\draw[lightgray,postaction={
								decorate,
								decoration={
									markings,
									mark=between positions 0.42 and 0.58 step 0.16 with \arrow{stealth};
								}}] 
							plot[domain=-1.5:1.5, samples=50, smooth] 
							({\hzero*cos(\x r)^2},{- tan(\x r)});
					}

					% Curve in black
					\pgfmathsetmacro{\hzero}{1.2}
					\draw[postaction={
							decorate,
							decoration={
								markings,
								mark=between positions 0.47 and 0.54 step 0.06 with \arrow{stealth};
							}}] 
							plot[domain=-1.5:1.5, samples=50] 
							({\hzero*cos(\x r)^2},{- tan(\x r)});
				\end{scope}
			\end{tikzpicture}
			\caption{Integral curves of \( Y_{\bar{H}} \) in \( \CotBundle \R_{> 0} \).}
		\end{subfigure}%
		\begin{subfigure}[t]{.4\textwidth}
			\centering
			\begin{tikzpicture}[scale=0.8]
				% Curves in light gray
				\foreach \hzero in {0.8,1.5,3.1,3.9,4.5}{% hzero = energy
					\draw[lightgray,postaction={
							decorate,
							decoration={
								markings,
								mark=between positions 0.2 and 0.65 step 0.45 with \arrow{stealth};
							}}] 
						(0,-\hzero) ellipse [x radius={-\hzero / 3 + 2}, y radius={0.3 - \hzero / 20}];
				}

				% Curve in black
				\pgfmathsetmacro{\hzero}{2.3}
				\draw[postaction={
						decorate,
						decoration={
							markings,
							mark=between positions 0.2 and 0.65 step 0.45 with \arrow{stealth};
						}}] 
					(0,-\hzero) ellipse [x radius={-\hzero / 3 + 2}, y radius={0.3 - \hzero / 20}];

				% Outer cone
				\draw (0,0) ellipse [x radius=2, y radius=0.3];
				\draw (2,0) -- (0,-6) -- (-2,0);

				% Strata labels
				\node (generic) at (3.5, 0) {$\seam{\check{P}}^{\set{e}}_{\set{e}}$};
				\node (seam) at (2.2, -3) {$\seam{\check{P}}^{\set{e}}_{\UGroup(1)}$};
				\node (singular) at (1.5, -5.8) {$\seam{\check{P}}^{\UGroup(1)}_{\UGroup(1)}$};
				\draw[->] (generic) to [in=40, out=160] (1,-0.5);
				\draw[->] (seam) to [in=-25, out=180] (1.1,-2.8);
				\draw[->] (singular) to [in=-25, out=190] (0.1,-6);
			\end{tikzpicture}
			\caption{Integral curves of \( X_{H} \) in \( \check{P} \).}
		\end{subfigure}
		\caption{Comparison of the Hamiltonian flows in \( \CotBundle \R_{>0} \) and \( \check{P} \).}
		\label{fig:cotangentBundle:harmonicOss:integralCurves}
	\end{figure}
	Moreover, \( K \) is \( \UGroup(1) \)-invariant and descends to a homeomorphism
	\begin{equation}
		\check{K}: \check{P} \to C
	\end{equation}
	of the reduced phase space \( \check{P} = J^{-1}(0) \slash \UGroup(1) \) with the upper cone \( C \subseteq \R^3 \).
	The image of \( \seam{\check{P}}^{\UGroup(1)}_{\UGroup(1)} \) under \( \check{K} \) is the origin and the seam \( \seam{\check{P}}^{\set{e}}_{\UGroup(1)} \) gets mapped onto the line \( L \) determined by \( E_- = 0 \), \( H = E_+ \) and \( H > 0 \).
	The remaining part of the cone corresponds to \( \seam{\check{P}}^{\set{e}}_{\set{e}} \).
	
	We now pass to the symplectic structure.
	The Poisson brackets of the components \( E_\pm \) and \( H \) of \( K \) (viewed as real-valued functions on \( \CotBundle \R^2 \)) are given by
	\begin{equation}
		\label{eq:cotangentBundle:harmonicOss:commutatorSL}
		\poisson{H}{E_\pm} = \mp \, 2 E_\mp \qquad \poisson{E_+}{E_-} = 2 H.
	\end{equation}
	These relations are identical with the commutation relations of \( \SLAlgebra(2, \R) \).
	Hence, \( K \) is a Poisson map from \( \CotBundle \R^2 \) to \( \R^3 \isomorph \SLAlgebra(2, \R) \), where the latter space carries the usual Lie--Poisson structure, \ie, the one given by the bivector field
	\begin{equation}
		\label{eq:cotangentBundle:harmonicOss:poissonBivector}
		\Pi = - 2 E_- \, \partial_H \wedge \partial_{E_+} + 2 E_+ \, \partial_H \wedge \partial_{E_-} + 2 H \, \partial_{E_+} \wedge \partial_{E_-}.
	\end{equation}
	The symplectic top stratum \( \check{P}_{\set{e}} \), \ie the cone without the origin, is a coadjoint orbit of \( \SLGroup(2, \R) \) and thus carries the Kostant--Kirillov--Souriau symplectic form.
	As the singular stratum \( \check{P}_{\UGroup(1)} \) is zero-dimensional, its symplectic form vanishes.

	Recall from \cref{prop:cotangentBundle:restrictionSymplecticForm} the construction of the symplectomorphism \( \check{\bar{\pi}} \) between \( \seam{\check{P}}^{\set{e}}_{\set{e}} \) and \( \CotBundle (\check{Q}_{\set{e}}) \).
	Moreover, the map \( \equivClass{q} \mapsto \frac{1}{2} \norm{q}^2 \) identifies \( \check{Q}_{\set{e}} \) with \( \R_{> 0} \) and thus yields a symplectomorphism of \( \CotBundle (\check{Q}_{\set{e}}) \) with \( \CotBundle \R_{>0} \isomorph \R_{>0} \times \R \).
	A straightforward calculation shows that the combined symplectomorphism 
	\begin{equation}
		\psi: \seam{\check{P}}^{\set{e}}_{\set{e}} \to \CotBundle (\check{Q}_{\set{e}}) \to \CotBundle \R_{>0}
	\end{equation}
	is given by
	\begin{equation}
		\psi(\equivClass{q, p}) = \left(\frac{1}{2} \norm{q}^2, \frac{q \cdot p}{\norm{q}^2}\right) = \left(\frac{1}{2} (H - E_+), \frac{E_-}{H - E_+}\right).
	\end{equation}
	\todo{The canonical one form on the open cone has thus the form 
	\(
		\bar{\theta} = \frac{1}{2} \frac{E_-}{H - E_+} \dif(H - E_+) = \frac{E_-}{2} \dif \ln(H - E_+).
	\)
	It appears that the symplectic form has a log-divergence at \( H = E_- \).
	However, this cannot be the case as the whole cone is a symplectic manifold.
	How to resolve this paradox?
	}
	Consider the map
	\begin{equation}
		I: \CotBundle \R_{> 0} \to \R^3, \quad 
			(\bar{q}, \bar{p}) \mapsto \Vector{\bar{q} (\bar{p}^2 - 1) \\ 2 \bar{q} \bar{p} \\ \bar{q} (\bar{p}^2 + 1)}. 
	\end{equation}
	The image of \( I \) is the upper cone \( C \) without the line \( L \) and the following diagram commutes
	\begin{equationcd}
		\seam{\check{P}}^{\set{e}}_{\set{e}} \to[r, "\check{K}"] \to[d, "\check{\bar{\pi}}"] \to[rd, "\psi"]
			& C \setminus L
		\\
		\CotBundle (\check{Q}_{\set{e}}) \to[r]
			& \CotBundle \R_{> 0} \, .\to[u, "I"]
	\end{equationcd}
	Moreover, a direct calculation shows that the components of \( I \) again satisfy the commutation relations~\eqref{eq:cotangentBundle:harmonicOss:commutatorSL} and hence \( I \) is a Poisson immersion of \( \CotBundle \R_{> 0} \) into \( (\SLAlgebra(2, \R), \Pi) \).
	To summarize the kinematic picture, we have decomposed the reduced phase space into the symplectic strata \( \check{P}_{\UGroup(1)} \) and \( \check{P}_{\set{e}} \).
	The symplectic stratum \( \check{P}_{\set{e}} \) further decomposes into the cotangent bundle \( \CotBundle \R_{> 0} \) and the line \( L \).
	This decomposition is in accordance with \cref{prop:cotangentBundle:singularCotangentBundleRed}.

	Let us now discuss the dynamics.
	Using~\eqref{eq:cotangentBundle:harmonicOss:poissonBivector}, the Hamiltonian vector field \( X_H = \dif H \contr \Pi \) on \( \R^3 \) generated by \( H \) is given by
	\begin{equation}
		X_H = -2 E_- \, \partial_{E_+} + 2 E_+ \, \partial_{E_-}.
	\end{equation}
	Hence, the time evolution is given by rotation in the \( (E_+ , E_-) \)-plane with \( H = \const \).
	In particular, the flow periodically hits the line \( L \), \ie, the seam \( \seam{\check{P}}^{\set{e}}_{\UGroup(1)} \).
	It is interesting to compare this behavior to the Hamiltonian flow on \( \CotBundle \R_{> 0} \) generated by 
	\begin{equation}
		\bar{H} \defeq I^* H = \bar{q} (\bar{p}^2 + 1).
	\end{equation}
	The associated Hamiltonian vector field has the form
	\begin{equation}
	 	Y_{\bar{H}} = 2 \bar{q} \bar{p} \, \difp_{\bar{q}} - (\bar{p}^2 + 1) \, \difp_{\bar{p}}
	\end{equation}
	and, hence, the integral curves are given by
	\begin{equation}
		\bar{q}(t) = \bar{H}_0 \cos^2 (t + t_0), \qquad \bar{p}(t) = - \tan(t + t_0),
	\end{equation}
	where \( \bar{H}_0 \) and \( t_0 \) are determined by the initial conditions.
	Under the map \( I \), they read
	\begin{equation}
		t \mapsto \bar{H}_0 \Vector{-\cos\bigl(2 (t + t_0)\bigr) \\ -\sin\bigl(2 (t + t_0)\bigr) \\ 1}.
	\end{equation}
	Note that in \( \CotBundle \R_{> 0} \) the flow is not defined at times \( t_c = \frac{\pi}{2} + k \pi - t_0 \) with \( k \in \N \).
	At these times, \( \bar{p} \) explodes, \ie, the trajectory periodically tries to quickly leave the configuration space \( \R_{>0} \).
	On the other hand, the flow under \( I \) continuously extends to \( t \in \R \).
	In other words, the map \( I \) plays the role of regularizing the dynamics in \( \CotBundle \R_{>0} \).
	In this sense, the singular seam \( \seam{\check{P}}^{\set{e}}_{\UGroup(1)} \) stitches together the singular solution in \( \seam{\check{P}}^{\set{e}}_{\set{e}} \isomorph \CotBundle \R_{> 0} \) to a nice periodic flow.
	See \cref{fig:cotangentBundle:harmonicOss:integralCurves} for a visual comparison of the flows in \( \CotBundle \R_{>0} \) and \( \check{P} \).
\end{example}

%%%%%%%%%%%%%%%%%%%%%%%%%%%%%%%%%%%%%%%%%%%%%%%%%%%%%%%%%%%%%%%%%%%%%%%%%%%%%%%%%%%%%%%%

\section{Yang--Mills--Higgs theory}
\label{sec:yangMillsHiggs}

%%%%%%%%%%%%%%%%%%%%%%%%%%%%%%%%%%%%%%%%%%%%%%%%%%%%%%%%%%%%%%%%%%%%%%%%%%%%%%%%%%%%%%%%

In the sequel, we will investigate the stratified structure of the reduced phase space of the Yang--Mills--Higgs theory.
This important class of gauge field theories plays a pivotal role in physics as well as in mathematics.
First, it provides deep insights into the Standard Model of particle physics.
The role of the Higgs field is to partially break the gauge symmetry and endow some of the gauge bosons with a mass.
Second, the major interest in understanding the classical Yang--Mills--Higgs equations is due to the existence of instantons and magnetic monopole solutions (see \eg \parencite[Chapter~7]{RudolphSchmidt2014} and references therein).
Third, it inspired a wealth of deep results in geometric analysis starting from the fundamental work of \textcite{Hitchin1987,Taubes1985,Taubes1982a} and others.

For the basics we refer to \parencite{DiezRudolphClebsch}, where we have studied the Hamiltonian structure of this model based on the \( (3+1) \)-splitting of the configuration space and a geometric constraint analysis.
In order to fix the notation and the conventions, let us briefly recall the main points.

Let \( (M, g) \) be a \( 3 \)-dimensional compact oriented manifold without boundary endowed with a time-dependent Riemannian metric, which plays the role of a Cauchy surface\footnote{To be more precise, \( M \) should be viewed as the compactification of a Cauchy surface. The choice of compactification corresponds to certain boundary conditions on the fields at spatial infinity.} in the \(( 3+1 )\)-splitting.
Denote the induced volume form by \( \vol_g \).
The geometry underlying Yang--Mills-Higgs theory is that of a principal \( G \)-bundle \( P \to M \), where \( G \) is a connected compact Lie group.
A connection in \( P \) is a splitting of the tangent bundle \( \TBundle P = \VBundle P \oplus \HBundle P \) into the canonical vertical distribution \( \VBundle P \) and a horizontal distribution \( \HBundle P \).
Recall that \( \VBundle P \) is spanned by the fundamental vector fields \( p \mapsto \xi \ldot p \) for \( \xi \in \LieA{g} \).
In particular, \( \VBundle P \slash G \) is isomorphic to the adjoint bundle \( \AdBundle P = P \times_G \LieA{g} \).
The covariant derivative with respect to a connection \( A \) is denoted by \( \dif_A \) and the curvature of \( A \) is written as \( F_A \).
A bosonic matter field is a smooth section \( \varphi \) of the associated vector bundle \( F = P \times_G \FibreBundleModel{F} \), where the typical fiber \( \FibreBundleModel{F} \) carries a unitary \( G \)-representation.
Thus, the space of configurations \( \SectionSpaceAbb{Q} \) of Yang--Mills--Higgs theory consists of pairs \( (A, \varphi) \).
It obviously is the product of the infinite-dimensional affine Fréchet space \( \ConnSpace \) of smooth connections and the Fréchet space \( \SectionSpaceAbb{F} \) of smooth sections of \( F \).
Let \(  \FibreBundleModel{V}: \FibreBundleModel{F} \to \R \) be a \( G \)-invariant function and denote the induced function on \( F \) by \( V \) (the Higgs potential).

We assume that \( \FibreBundleModel{F} \) carries a \( G \)-invariant Riemannian metric, which provides an identification of \( F \) with its dual bundle \( F^* \).
In order to underline that the Hodge dual is defined in terms of a linear functional on the space of differential forms, we nonetheless use the convention that the Hodge dual of a vector-valued differential form \( \alpha \in \DiffFormSpace^k(M, F) \) is the \emph{dual-valued} differential form \( \hodgeStar \alpha \in \DiffFormSpace^{3-k}(M, F^*) \).
Moreover, we use the diamond product\footnote{This is the natural extension to differential forms of the diamond product \( \diamond: F \times F^* \to \LieA{g}^* \) that plays an important role in the study of Lie--Poisson systems.}
\begin{equation}
	\diamond: \DiffFormSpace^k(M, F) \times \DiffFormSpace^{3-r-k}(M, F^*) \to \DiffFormSpace^{3-r}(M, \CoAdBundle P)\, ,
\end{equation}
which is defined by
\begin{equation}
	\wedgeDual{\xi}{(\alpha \diamond \beta)} = \wedgeDual{(\xi \wedgeldot \alpha)}{\beta} \in \DiffFormSpace^{\dim M}(M) \quad \text{for all } \xi \in \DiffFormSpace^r(M, \AdBundle P),
\end{equation}
where \( \wedgeldot: \DiffFormSpace^r(M, \AdBundle P) \times \DiffFormSpace^k(M, F) \to \DiffFormSpace^{r+k}(M, F) \) is the natural operation obtained by combining the Lie algebra action \( \LieA{g} \times \FibreBundleModel{F} \to \FibreBundleModel{F}, (\xi, f) \mapsto \xi \ldot f \) with the wedge product operation.
Moreover, for \( \alpha \in \DiffFormSpace^k(M, F) \) and \( \beta \in \DiffFormSpace^{3-k}(M, F^*) \), we have denoted by \( \wedgeDual{\alpha}{\beta} \) the real-valued top-form which arises from combining the wedge product with the natural pairing \( \dualPairDot: F \times F^* \to \R \).

Since \( \SectionSpaceAbb{Q} \) is an affine space, its tangent bundle is trivial with fiber \( \DiffFormSpace^1(M, \AdBundle P) \times \sSectionSpace(F) \).
We will denote points in \( \TBundle \SectionSpaceAbb{Q} \) by tuples \( (A, \alpha, \varphi, \zeta) \) with \( \alpha \in \DiffFormSpace^1(M, \AdBundle P) \) and \( \zeta \in \sSectionSpace(F) \).
A natural choice for the cotangent bundle $\CotBundle \SectionSpaceAbb{Q}$ is the trivial bundle over $\SectionSpaceAbb{Q}$ with fiber \( 	\DiffFormSpace^2(M, \CoAdBundle P) \times \DiffFormSpace^3(M, F^*)  \).
We denote elements of this fiber by pairs \( (D, \Pi) \).
Then, the natural pairing with \( \TBundle \SectionSpaceAbb{Q} \) is given by integration over \( M \), 
\begin{equation}
	\dualPair*{(D, \Pi)}{(\alpha, \zeta)} = \int_M \wedgeDual{D}{\alpha} + \int_M \wedgeDual{\Pi}{\zeta}.
\end{equation}

The equations of motion are derived from their \( 4 \)-dimensional covariant counterpart in the temporal gauge, see \parencite{DiezRudolphClebsch}.
In this \( (3+1) \) formulation, \( M \) is a Cauchy surface and the Lorentzian metric on \( \R \times M \) is of the form \( - \ell(t)^2 \dif t^2 + g(t) \), where \( \ell \) is the lapse function, that is, \( \ell(t) \in \sFunctionSpace(M) \).
With this notation at hand, the evolutionary form of the Yang--Mills--Higgs equations is given by:
\begin{subequations}\label{eq:yangMillsHiggs:equationsOfMotion}\begin{align+}
	\partial_t D &= - \dif_A (\ell \hodgeStar F_A) - \ell \varphi \diamond \hodgeStar (\dif_A \varphi),
	\label{eq:yangMillsHiggs:ampere}
	\\
	\partial_t A &= \ell \hodgeStar D,
	\label{eq:yangMillsHiggs:Dmomentum}
	\\
	\partial_t \Pi &= \dif_A (\ell \hodgeStar \dif_A \varphi) - \ell \, V'(\varphi) \vol_g,
	\label{eq:yangMillsHiggs:higgs}
	\\
	\partial_t \varphi &= \ell \hodgeStar \Pi,
	\label{eq:yangMillsHiggs:pimomentum}
	\\
	\dif_A D &+ \varphi \diamond \Pi = 0 \, .
	\label{eq:yangMillsHiggs:GaussConstraint}
\end{align+}\end{subequations}
By \parencite[Theorem~3.4]{DiezRudolphClebsch}, the evolution \cref{eq:yangMillsHiggs:ampere,eq:yangMillsHiggs:Dmomentum,eq:yangMillsHiggs:higgs,eq:yangMillsHiggs:pimomentum} are Hamiltonian with respect to 
\begin{equation}\label{eq:yangMillsHiggs:hamiltonian}\begin{split}
	\SectionSpaceAbb{H}(A, D, \varphi, \Pi)
		&= \int_M \frac{\ell}{2} \Bigl(\wedgeDual{D}{\hodgeStar D} +  \wedgeDual{F_A}{\hodgeStar F_A}
		+ \wedgeDual{\Pi}{\hodgeStar \Pi} + \wedgeDual{\dif_A \varphi}{\hodgeStar \dif_A \varphi}
		\\
		&\qquad+ 2 \, V(\varphi) \vol_g \Bigr).
\end{split}\end{equation}
Moreover, \cref{eq:yangMillsHiggs:GaussConstraint} is the Gauß constraint.
In terms of the cotangent bundle geometry, it has the following interpretation, \cf \parencite{DiezRudolphClebsch,Sniatycki1999,ArmsMarsdenEtAl1981}.
On \( \SectionSpaceAbb{Q} = \ConnSpace \times \SectionSpaceAbb{F} \) we have a left action of the group \( \GauGroup = \sSectionSpace(P \times_G G) \) of local gauge transformations, 
\begin{equation}
\label{LocGTr}
	A \mapsto \AdAction_{\lambda} A + \lambda \dif \lambda^{-1},
	\quad
	\varphi \mapsto \lambda \cdot \varphi,
\end{equation}
for \( \lambda \in \GauGroup \).
The Hamiltonian \( \SectionMapAbb{H} \) is invariant under the lift of the action to \( \CotBundle \SectionSpaceAbb{Q} \).
A straightforward calculation shows that 
\begin{equation}\label{eq:yangMillsHiggs:momentumMap}
	\SectionMapAbb{J}(A, D, \varphi, \Pi) = \dif_A D + \varphi \diamond \Pi
\end{equation}
is the momentum map for the lifted action with respect to the natural choice \( \GauAlgebra^* = \DiffFormSpace^3(M, \CoAdBundle P) \), see \parencite[Equation~3.10]{DiezRudolphClebsch}.
Hence, the Gauß constraint~\eqref{eq:yangMillsHiggs:GaussConstraint} is equivalent to the momentum map constraint \( \SectionMapAbb{J} = 0 \).
\begin{remark}
	In \parencite{DiezRudolphClebsch}, we have accomplished a unification of the Hamiltonian evolution \cref{eq:yangMillsHiggs:ampere,eq:yangMillsHiggs:Dmomentum,eq:yangMillsHiggs:higgs,eq:yangMillsHiggs:pimomentum} with the constraint~\eqref{eq:yangMillsHiggs:GaussConstraint} by developing a novel variational principle (called the Clebsch--Lagrange principle).
	Besides the variation of configuration variables, the latter includes also the variation of the symmetry generators of the system.
	Here, these generators coincide with the time-component of the gauge potential of the \( 4 \)-dimensional theory.
	In this language, the choice of the temporal gauge has the interpretation of being the first step in symplectic reduction by stages.
	In the sequel, we discuss the reduction of the remaining symmetry of the Cauchy problem.
	A version of the reduction by stages theorem thus shows that the reduced phase space we obtain coincides with the reduced phase space of the \( 4 \)-dimensional theory.
\end{remark}

Note that the action of gauge transformations on \( \SectionSpaceAbb{Q} \) is usually not free.
Hence, the model under consideration fits into the general setting of infinite-dimensional \emph{singular} cotangent bundle reduction as discussed in \cref{sec:cotangentBundleReduction}.
We now show that all assumptions made in the general discussion are met for the Yang--Mills--Higgs system:
\begin{enumerate}
	\item 
		\( \SectionSpaceAbb{Q} \) is a Fréchet manifold, because it is an affine space modeled on the Fréchet vector space \( \DiffFormSpace^1(M, \AdBundle P) \times \sSectionSpace(F) \).
	\item
		\( \GauGroup \) is a Fréchet Lie group, because it is realized as the space of smooth sections of the group bundle \( P \times_G G \), see \parencite{CirelliMania1985} for details.
	\item
		The cotangent bundle \( \CotBundle \SectionSpaceAbb{Q} = \SectionSpaceAbb{Q} \times \DiffFormSpace^2(M, \CoAdBundle P) \times \DiffFormSpace^3(M, F^*) \) is clearly a Fréchet vector bundle.
		The Hodge operator yields a fiber-preserving \( \GauGroup \)-equivariant isomorphism between \( \TBundle \SectionSpaceAbb{Q} \) and \( \CotBundle \SectionSpaceAbb{Q} \).
	\item 
		The \( \GauGroup \)-action on \( \SectionSpaceAbb{Q} \) is affine and thus smooth.
		Moreover, it is proper, see \parencite{Diez2013,RudolphSchmidtEtAl2002}.
	\item 
		The \( \GauGroup \)-action on \( \SectionSpaceAbb{Q} \) admits a slice at every point.
		First, the slice \( \SectionSpaceAbb{S}_{A_0} \) at \( A_0 \in \ConnSpace \) is given by the Coulomb gauge condition.
		That is\footnote{Here, as usual, \( \dif^*_{A} \alpha \defeq (-1)^k \hodgeStar \dif_{A} \hodgeStar \alpha \) for a \( k \)-form \( \alpha \).},
		\begin{equation}
			\SectionSpaceAbb{S}_{A_0} \defeq \set{A \in \SectionSpaceAbb{U} \given \dif_{A_0}^* (A - A_0) = 0},
		\end{equation}
		where \( \SectionSpaceAbb{U} \) is an open neighborhood of \( A_0 \) in \( \ConnSpace \).
		In order to see that \( \SectionSpaceAbb{S} \) is a slice, indeed, one uses the Nash--Moser inverse function theorem \parencite{Hamilton1982}, which amongst other things relies on the fact that \( \SectionSpaceAbb{Q} \) and \( \GauGroup \) are in fact tame Fréchet.
		The details can be found in \parencite{Diez2013,AbbatiCirelliEtAl1989}.
		
		Note that this slice for the \( \GauGroup \)-action on \( \ConnSpace \) fixes the gauge transformations up to elements of the stabilizer \( \GauGroup_{A_0} \) of \( A_0 \).
		Thus, we are left with the \( \GauGroup_{A_0} \)-action on \( \SectionSpaceAbb{F} \).
		This is a linear action of a finite-dimensional compact group on a Fréchet space and hence has a slice \( \SectionSpaceAbb{S}_{\varphi_0} \) at every point \( \varphi_0 \in \SectionSpaceAbb{F} \), see \cref{prop:slice:sliceTheoremLinearAction}.
		By \cref{prop:slice:sliceForProduct}, the product \( \SectionSpaceAbb{S}_{A_0} \times \SectionSpaceAbb{S}_{\varphi_0} \) is a slice at \( (A_0, \varphi_0) \) for the \( \GauGroup \)-action on \( \SectionSpaceAbb{Q} \).
		This slice is compatible with the cotangent bundle structures in the sense of \cref{defn:cotangentBundle:sliceCompatible} as will be discussed in \cref{sec:yangMillsHiggs:normalForm}. 
	\item
		That every orbit of \( \GauGroup \) is symplectically closed will be shown in \cref{prop:yangMillsHiggs:reduction:orbitsSymplecticallyClosed} below.
		Thus, in particular, the strong version of the Bifurcation \cref{prop:bifurcationLemma} holds.
\end{enumerate}
As a consequence, \cref{prop:cotangentBundle:simpleNormalForm} holds.
For \cref{prop:cotangentBundle:singularSympRed,prop:cotangentBundle:singularCotangentBundleRed} to hold we assume, additionally, that the frontier condition for the decomposition of \( \SectionSpaceAbb{Q} \) into gauge orbit types is satisfied.
The orbit type decomposition of \( \ConnSpace \) satisfies the frontier condition \parencite[Theorem~4.3.5]{KondrackiRogulski1986}.
However, including matter fields is a rather delicate issue.
As we will see, the stabilizer of the Higgs field is given in terms of a series of intersections of stabilizer groups of the \( G \)-action on \( \FibreBundleModel{F} \) and this intersection is hard to analyze in full generality.
For the Glashow--Weinberg--Salam model the frontier condition can be verified by direct inspection, see \cref{prop:yangMillsHiggs:GWS:approximationHiggs}.

Subsequently, we analyze the content of these theorems for the case under consideration.

\subsection{Normal form}
\label{sec:yangMillsHiggs:normalForm}

As for the classical Hodge--de Rham complex, elliptic theory\footnote{These Hodge-type decompositions are usually derived in a Sobolev context but elliptic regularity implies that the decompositions below hold in the \( \sFunctionSpace \)-setting.} gives topological isomorphisms (\cf \parencite[Theorem~6.1.9]{RudolphSchmidt2014}):
\begin{subequations}
\begin{align+}
	\label{eq:yangMillsHiggs:decomposeLie}
	\DiffFormSpace^0(M, \AdBundle P) &= \img \dif^*_{A} \oplus \ker \dif_{A},
	\\
	\label{eq:yangMillsHiggs:decomposeTangent}
	\DiffFormSpace^1(M, \AdBundle P) &= \img \dif_{A} \oplus \ker \dif^*_{A},
\intertext{where \( \dif^*_{A} \) is the codifferential. By applying the Hodge star operator, we obtain the dual decompositions:}
	\label{eq:yangMillsHiggs:decomposeLieDual}
	\DiffFormSpace^3(M, \CoAdBundle P) &= \img \dif_{A} \oplus \ker \dif^*_{A},
	\\
	\DiffFormSpace^2(M, \CoAdBundle P) &= \img \dif^*_{A} \oplus \ker \dif_{A}.
\end{align+}
\end{subequations}
For clarity of presentation, we first derive the normal form of \cref{prop:cotangentBundle:simpleNormalForm} for \( \CotBundle \ConnSpace \) and thereby ignore the Higgs part for a moment (see \cref{rem:yangMillsHiggs:normalFormWithHiggs}).
By the above decompositions, we have a splitting of \( \TBundle \ConnSpace \),
\begin{equation}
	\TBundle_{A} \ConnSpace 
		\isomorph \DiffFormSpace^1(M, \AdBundle P) 
		= \img \dif_{A} \oplus \ker \dif^*_{A},
\end{equation}
into the canonical vertical distribution \( \img \dif_{A} \) and the \( L^2 \)-orthogonal complement \( \ker \dif^*_{A} \).
This decomposition is basic for the study of the geometry of the stratification, see \parencite[Sections~8.3 and~8.4]{RudolphSchmidt2014}.
As one expects from Hodge theory, these decompositions satisfy the annihilation relations
\begin{align}
	\left(\img \dif^*_{A}\right)^\polar &= \ker \dif^*_{A},
	\qquad
	\left(\ker \dif_{A}\right)^\polar = \img \dif_{A},
	\\
	\left(\img \dif_{A}\right)^\polar &= \ker \dif_{A},
	\qquad
	\left(\ker \dif^*_{A}\right)^\polar = \img \dif^*_{A},
\end{align}
with respect to the natural  pairing between \( \DiffFormSpace^k(M, \AdBundle P) \) and \( \DiffFormSpace^{3-k}(M, \CoAdBundle P) \).
Thus, for the spaces involved in \cref{prop:cotangentBundle:simpleNormalForm} we obtain:
\begin{alignat}{2}
	\GauAlgebra_{A_0} &= \ker \dif_{A_0}: \DiffFormSpace^0(M, \AdBundle P) \to \DiffFormSpace^1(M, \AdBundle P),
	\\
	\GauAlgebra_{A_0}^* &= \ker \dif^*_{A_0}: \DiffFormSpace^3(M, \CoAdBundle P) \to \DiffFormSpace^2(M, \CoAdBundle P),
	\\
	\LieA{m} &= \img \dif^*_{A_0}: \DiffFormSpace^1(M, \AdBundle P) \to \DiffFormSpace^0(M, \AdBundle P),
	\\
	\LieA{m}^* &= \img \dif_{A_0}: \DiffFormSpace^2(M, \CoAdBundle P) \to \DiffFormSpace^3(M, \CoAdBundle P),
	\\
	\TBundle_A \SectionSpaceAbb{S}_{A_0} &= \ker \dif^*_{A_0}: \DiffFormSpace^1(M, \AdBundle P) \to \DiffFormSpace^0(M, \AdBundle P),
	\\
	\CotBundle_A \SectionSpaceAbb{S}_{A_0} &= \ker \dif_{A_0}: \DiffFormSpace^2(M, \CoAdBundle P) \to \DiffFormSpace^3(M, \CoAdBundle P).
\end{alignat}
In the sequel, we use \( \parallel \) or \( \perp \) to denote objects that are parallel or orthogonal to the gauge orbits, respectively.
Accordingly, we can write every \( E \in \TBundle_A \ConnSpace \) as \( E = \AdAction_\lambda (E^\perp + E^\parallel) \) with \( E^\perp \in \TBundle_A \SectionSpaceAbb{S}_{A_0} \) and \( E^\parallel = - \dif_{A} \chi \) for some \( \chi \in \LieA{m} \).
As a consequence of the above identifications, the local diffeomorphism~\eqref{eq:cotangentBundle:normalFormTBundle} takes the form
\begin{equation}
	\label{eq:yangMillsHiggs:normalFormTBundle}
	\GauGroup \times_{\GauGroup_{A_0}} (\LieA{m} \times \TBundle \SectionSpaceAbb{S}_{A_0}) \to \TBundle \ConnSpace,
	\quad
	\equivClass{\lambda, (\chi, A, E^\perp)} \mapsto \left(\lambda \cdot A, \AdAction_\lambda(E^\perp - \dif_{A} \chi)\right).
\end{equation}
In order to determine the dual map \( \Phi \), we introduce the Faddeev--Popov operator \( \laplace_{A A_0} \) for every \( A \in \SectionSpaceAbb{S}_{A_0} \) as the composition
\begin{equationcd}
	\img \dif_{A_0} \toInject{r} 
		& \DiffFormSpace^3(M, \CoAdBundle P) \to[r, "\dif_{A} \dif^*_{A_0}"] 
		& \DiffFormSpace^3(M, \CoAdBundle P) \to[r, "\pr"]
		& \img \dif_{A_0}.
\end{equationcd}
By possibly shrinking the slice \( \SectionSpaceAbb{S}_{A_0} \), we may assume that \( \laplace_{A A_0} \) is an invertible operator on \( \img \dif_{A_0} \) for all \( A \in \SectionSpaceAbb{S}_{A_0} \), see \parencite[p.~655]{RudolphSchmidt2014} or \parencite[p.~22]{DiezHuebschmann2017}.
In particular, we have \( \dif_{A} \dif^*_{A_0} \circ \laplace^{-1}_{A A_0} = \id_{\img \dif_{A_0}} \).
Thus, the operator
\begin{equation}
	T_{A}: \CotBundle_A \SectionSpaceAbb{S}_{A_0} \times \LieA{m}^* \to \DiffFormSpace^2(M, \CoAdBundle P),
	\quad
	(D^\perp, \nu) \mapsto D^\perp + \dif^*_{A_0} \laplace^{-1}_{A A_0} \nu
\end{equation}
satisfies 
\begin{equation}
	\pr_{\ker \dif_{A}} \circ T_{A} (D^\perp, \nu) = D^\perp,
	\qquad
	\pr_{\img \dif_{A}} \circ \dif_{A} T_{A} (D^\perp, \nu) = \nu.
\end{equation}
Thus, \( T_A \) is an isomorphism of Fréchet spaces for every \( A \in \SectionSpaceAbb{S}_{A_0} \).
Moreover, we find
\begin{equation}
	\dualPair{E^\perp - \dif_{A} \chi}{T_{A}(D^\perp, \nu)}
		= \dualPair{E^\perp}{T_{A}(D^\perp, \nu)} + \dualPair{\chi}{\dif_{A} T_{A}(D^\perp, \nu)}
		= \dualPair{E^\perp}{D^\perp} + \dualPair{\chi}{\nu}.
\end{equation}
In summary, here, the local diffeomorphism \( \Phi \) defined in~\eqref{eq:cotangentBundle:normalFormCotBundle}  has the form
\begin{equation}
	\label{eq:yangMillsHiggs:normalFormCotBundle}
	\GauGroup \times_{\GauGroup_{A_0}} (\LieA{m}^* \times \CotBundle \SectionSpaceAbb{S}_{A_0}) \to \CotBundle \ConnSpace,
	\quad
	\equivClass{\lambda, (\nu, A, D^\perp)} \mapsto \left(\lambda \cdot A, \CoAdAction_\lambda(D^\perp + \dif^*_{A_0} \laplace^{-1}_{A A_0} \nu)\right).
\end{equation}
Thus, every \( D \in \CotBundle_A \ConnSpace \) decomposes as \( D = \CoAdAction_\lambda(D^\perp + D^\parallel) \) with \( D^\perp \in \CotBundle_A \SectionSpaceAbb{S}_{A_0} \) and \( D^\parallel = \dif^*_{A_0} \laplace^{-1}_{A A_0} \nu \) for some \( \nu \in \LieA{m}^* \).
That is, \( \nu = \dif_{A_0} D^\parallel \).
\begin{remark}
	As we have seen, for every \( A \in \SectionSpaceAbb{S}_{A_0} \), the map
	\begin{equation}
		\LieA{m} \times \TBundle_A \SectionSpaceAbb{S}_{A_0} \to \TBundle_A \ConnSpace,
		\quad
		(\chi, E^\perp) \mapsto E^\perp - \dif_{A} \chi
	\end{equation}
	has the adjoint \( T_A \).
	Since \( T_A \) is an isomorphism, we see that the slice \( \SectionSpaceAbb{S} \) is compatible with the cotangent bundle structures.
\end{remark}
In these local coordinates, the momentum map \( \SectionMapAbb{J} \) of~\eqref{eq:yangMillsHiggs:momentumMap} is expressed as follows, (\cf~\eqref{eq:cotangentBundle:normalFormMomentumMap}):
\begin{equation}
	\label{eq:yangMillsHiggs:normalFormMomentumMap}
	\SectionMapAbb{J}(\Phi(\equivClass{\lambda, (\nu, A, D^\perp)}), \varphi, \Pi) = \CoAdAction_\lambda(\dif_{A} D^\perp + \nu) + \varphi \diamond \Pi,
\end{equation}
where \( A \in \SectionSpaceAbb{S}_{A_0} \) and \( D^\perp \in \ker \dif_{A_0} \).
We denote the matter charge density \( \varphi \diamond \Pi \) by \( \rho \) and decompose it with respect to~\eqref{eq:yangMillsHiggs:decomposeLieDual},
\begin{equation}
	\rho = \CoAdAction_\lambda (\rho^\parallel + \rho^\perp),
\end{equation}
where \( \rho^\parallel \in \img \dif_{A_0} \) and \( \rho^\perp \in \ker \dif^*_{A_0} \).
Note that \( \dif_{A} D^\perp \in \ker \dif^*_{A_0} \), because by upper semi-continuity of the kernel of a semi-Fredholm operator \parencite[Corollary~19.1.6]{Hoermander2007} we have
\begin{equation}
	\ker \dif^*_{A_0} \dif_{A} \subseteq \ker \dif^*_{A_0} \dif_{A_0} = \ker \dif_{A_0}.
\end{equation}
Thus, with respect to the decomposition \( \GauAlgebra^* = \LieA{m}^* \oplus \GauAlgebra_{A_0}^* \), the Gauß constraint \( \SectionMapAbb{J} \circ \Phi = 0 \) is equivalent to the following equations:
\begin{subequations}\label{eq:yangMillsHiggs:decompositionGauss}\begin{align}
	\label{eq:yangMillsHiggs:decompositionGauss:linear}
	\nu + \rho^\parallel &= 0 \, ,
	\\
	\label{eq:yangMillsHiggs:decompositionGauss:nonlinear}
	\dif_{A} D^\perp + \rho^\perp &= 0 \, .
\end{align}\end{subequations}
In summary, we have decomposed the non-linear Gauß constraint into a linear equation and finitely-many non-linear equations according to the fact that the stabilizer \( \GauAlgebra_{A_0} \) is finite-dimensional.
In spirit, this splitting is similar to the Kuranishi method or the Lyapunov--Schmidt construction, where one also chooses convenient coordinates to reduce a non-linear equation to a non-linear equation in finite dimensions.
However, in contrast to these procedures, we find the coordinates by exploiting the cotangent bundle geometry of the problem and do not directly use the inverse function theorem.
Note that the construction of the local diffeomorphism \( \Phi \) involves the non-local operator \( \laplace^{-1}_{A A_0} \).
Hence, the reconstruction of the solution \( (A, D) \) of the Gauß constraint from a solution \( (\nu, A, D^\perp) \) of~\eqref{eq:yangMillsHiggs:decompositionGauss} is a non-local and non-linear operation.
In the physics literature, one usually only considers the case of trivial stabilizer \( \GauAlgebra_{A_0} \) (\ie, irreducible connections, as we will see below).
In this case, the above construction reduces the Gauß constraint to the linear equation~\eqref{eq:yangMillsHiggs:decompositionGauss:linear}.
However, if one wants to include non-generic configurations, that is, connections with a non-trivial stabilizer,~\eqref{eq:yangMillsHiggs:decompositionGauss:nonlinear} must be taken into account as well.

On the quantum level, our observations suggest that the standard quantization methods, which only take the linear constraint~\eqref{eq:yangMillsHiggs:decompositionGauss:linear} into account, fail in the neighborhood of reducible connections and need to be supplemented by the non-linear constraint~\eqref{eq:yangMillsHiggs:decompositionGauss:nonlinear}.
For a quantization program for \emph{lattice} gauge theories, where non-generic gauge orbit strata are included, we refer to \parencite[Chapter~9]{RudolphSchmidt2014} and to \parencite{HuebschmannRudolphEtAl2009} for a case study.

\begin{remark}
	\label{rem:yangMillsHiggs:normalFormWithHiggs}
	In the above construction of the normal form, we have considered only \( \CotBundle \ConnSpace \) and, thereby, we have ignored the Higgs part.
By passing to the normal form  of the full cotangent bundle \( \CotBundle \SectionSpaceAbb{Q} \) according to \cref{prop:cotangentBundle:simpleNormalForm}, we note that the stabilizer \( \GauAlgebra_{A_0} \) further decomposes into the stabilizer \( \GauAlgebra_{(A_0, \varphi_0)} = \GauAlgebra_{A_0} \intersect \GauAlgebra_{\varphi_0} \) of \( (A_0, \varphi_0) \in \SectionSpaceAbb{Q} \) and some complement \( \LieA{r} \).
	Accordingly, the non-linear part~\eqref{eq:yangMillsHiggs:decompositionGauss:nonlinear} of the Gauß law further decomposes into a linear equation in \( \LieA{r}^* \) and a non-linear equation in \( \GauAlgebra_{(A_0, \varphi_0)}^* \).
\end{remark}

\subsection{Orbit types}
\label{sec:yangMillsHiggs:orbitTypes}

Next, let us consider \cref{prop:cotangentBundle:singularSympRed}.
By point (i) of this theorem, the set of orbit types of \( \SectionSpaceAbb{P} = \SectionMapAbb{J}^{-1}(0) \) with respect to the lifted \(\GauGroup \)-action coincides with the set of orbit types for the \(\GauGroup \)-action on \( \SectionSpaceAbb{Q} \).
We will now determine these orbit types. 
As above, first we limit our attention to the case \( \SectionSpaceAbb{Q} = \ConnSpace \).
As already mentioned in the introduction, the classification of orbit types for the \( \GauGroup \)-action on \( \ConnSpace \) was accomplished in \parencite{RudolphSchmidtEtAl2002b,RudolphSchmidtEtAl2002a,RudolphSchmidtEtAl2002,HertschRudolphSchmidt2010,HertschRudolphSchmidt2011}.

\newcommand{\lpnt}{p_0}
Choose a point \( \lpnt \in P \).
Evaluating gauge transformations at \( \lpnt \) yields a Lie group homomorphism \( \ev_{\lpnt}: \GauGroup \to G \).
A gauge transformation \( \lambda \in \GauGroup \) leaves the connection \( A \) invariant if and only if \( \lambda \) is constant on every \( A \)-horizontal curve, that is, if it is constant on the holonomy bundle \( P_A \) of \( A \) (based at \( \lpnt \)).
By \( G \)-equivariance, \( \lambda \in \GauGroup_A \) is completely determined by its value at some point in \( P_A \) and thus the evaluation map \( \ev_{\lpnt} \) restricts to an isomorphism of Lie groups between the stabilizer \( \GauGroup_A \) and the centralizer \( \centralizer_G (\HolGroup_A) \) of the holonomy group of \( A \) based at \( \lpnt \) (\cf \parencite[Theorem~2.1]{RudolphSchmidtEtAl2002}).
We usually suppress the evaluation map in our notation and view \( \GauGroup_A \) directly as a Lie subgroup of \( G \).
Recall that a subgroup that can be written as a centralizer is called a Howe subgroup.
In particular, 
\begin{equation}
	H = \centralizer_G (\GauGroup_A) = \centralizer^2_G (\HolGroup_A) 
\end{equation}
is a Howe subgroup of \( G \).
Consider the \( H \)-principal bundle
\begin{equation}
	P_H 
		\defeq P_A \times_{\HolGroup_A} H 
		\equiv P_A \cdot H 
		\subseteq P 
\end{equation}
associated to the holonomy bundle \( P_A \).
The bundle \( P_H \) consists of all \( p \in P \) obeying \( \lambda(p) = \lambda(\lpnt) \) for every \( \lambda \in \GauGroup_A \).
Conversely, the stabilizer subgroup can be recovered from \( P_H \) as the subgroup
\begin{equation}
	\GauGroup_A = \set{\lambda \in \GauGroup \given \restr{\lambda}{P_H} = \const}.
\end{equation}
A bundle reduction of $P$ to a Howe subgroup is called a Howe subbundle. 
A bundle reduction $ Q$ of $P$ is called holonomy-induced if there exists a connected reduction $\tilde Q \subseteq P$ to a subgroup $\tilde H$ such that $Q = \tilde Q \cdot \centralizer^2_G (\tilde H) $.
The following proposition (\cf \parencite[Theorem~3.3]{RudolphSchmidtEtAl2002}) is basic for the classification procedure. 
\begin{prop}
	\label{prop:yangMillsHiggs:orbitTypesIsomorphicToHoweSubbundles}
	Let $P$ be a principal $G$-bundle over $M$.
	Let $M$ be compact with $\dim M \geq 2$.
	Then, the mapping
	\begin{equation}
		\equivClass{\GauGroup_A} \mapsto \equivClass{P_A \cdot \centralizer^2_G (\HolGroup_A)}
	\end{equation}
	from the set of gauge orbit types to the set of isomorphism classes of holonomy-induced Howe subbundles of $P$ (factorized by the action of $G$) is a bijection.  
\end{prop}

\begin{remark}
	In \parencite{RudolphSchmidtEtAl2002}, this theorem is proved in the context of Sobolev spaces, but the result clearly holds true in the smooth category as well.
\end{remark}

\begin{remark}
	By \parencite[Theorem~6.2]{RudolphSchmidtEtAl2002b}, every Howe subbundle of a principal \( \SUGroup(n) \)-bundle is automatically holonomy induced.
	However, for other classical groups this is not always the case, see the counterexample after Theorem~6.2 in \parencite{RudolphSchmidtEtAl2002b}.
\end{remark}

Thus, to enumerate the gauge orbit types for a given principal \( G \)-bundle $P \to M$, one has to work through the following program:
\begin{enumerate}
	\item Classify the Howe subgroups up to conjugacy. 
	\item Classify the Howe subbundles up to isomorphy. 
	\item
		Extract the Howe subbundles which are holonomy-induced. 
	\item
		Factorize by the principal action. 
	\item Determine the natural partial ordering.
\end{enumerate}

For the convenience of the reader, we recall the result for the case of $G= \SUGroup(n)$ as accomplished in \parencite{RudolphSchmidtEtAl2002b}, see also \parencite{RudolphSchmidtEtAl2002a} for the discussion of the partial ordering.
The classification for the other classical Lie groups was presented in \parencite{HertschRudolphSchmidt2010,HertschRudolphSchmidt2011}. 
\begin{thm}
	\label{prop:yangMillsHiggs:classificationOrbitTypesSU}
	Let $P$ be a principal $\SUGroup(n)$-bundle over a compact manifold $M$ of dimension $2,3$, or $4$.
	The gauge orbit types of the space of connections on $P$ are in one-to-one correspondence with symbols $\equivClass{(I;\alpha, \xi)}$, where
	\begin{thmenumerate}
		\item
			$I = \bigl((k_1, \ldots, k_r), (m_1, \ldots, m_r)\bigr)$ is a pair of sequences of positive integers obeying
			\begin{equation}
				\sum_{i=1}^r k_i m_i = n \, ,
			\end{equation}
		\item
			$\alpha = (\alpha_1, \ldots, \alpha_r)$ is a sequence of elements $\alpha_i = 1 + \alpha_{i,1} + \alpha_{i,2} + \dotsb \in \sCohomology^\ast (M, \Z)$ with \( \alpha_{i,j} \in \sCohomology^{2j}(M, \Z) \) representing admissible values of the Chern classes \( \chernClass_j \) of $\UGroup(k_i)$-bundles over $M$,
		\item 
			$\xi \in \sCohomology^1 (M, {\mathbbm Z}_d)$, where $d$ is the greatest common divisor of $(m_1, \ldots, m_r)$.
	\end{thmenumerate}
	The cohomology elements $\alpha_i$ and $\xi$ are subject to the relations 
	\begin{equation}
		\sum_{i=1}^r \frac{m_i}{d} \alpha_{i, 1} = \beta_d (\xi), \quad \alpha_1^{m_1} \cup \ldots \cup \alpha_r^{m_r}  = \chernClass(P),
	\end{equation}
	 where $\chernClass(P)$ is the total Chern class of $P$ and 	$\beta_d : \sCohomology^1(M, \Z_d) \to \sCohomology^2(M, \Z)$ is the Bockstein homomorphism associated with the short exact sequence of coefficient groups $0 \to \Z \to \Z \to \Z_d \to 0$.
	 For any permutation $\sigma$ of $\{1, \ldots r\}$, the symbols  $\equivClass{(I;\alpha, \xi)}$ and $\equivClass{(\sigma I; \sigma \alpha, \xi)}$ have to be identified. 
\end{thm}

It is now straightforward to include the matter fields \( \varphi \) and to pass, thereby, from \( \ConnSpace \) to \( \SectionSpaceAbb{Q} \).
Indeed, a gauge transformation \( \lambda \) leaves \( \varphi \) invariant if and only if \( \lambda(p) \in G_{\varphi(p)} \) for all \( p \in P \), where \( G_{\varphi(p)} \) is the stabilizer of \( \varphi(p) \in \FibreBundleModel{F} \) under the \( G \)-action.
The equivariance properties \( \lambda(p \cdot g) = g^{-1} \lambda(p) g \) and 
\begin{equation}
	G_{\varphi(p \cdot g)} = G_{g^{-1} \cdot \varphi(p)} = g^{-1} G_{\varphi(p)} g
\end{equation}
show that it is actually enough to test \( \lambda(p) \in G_{\varphi(p)} \) only for one point per fiber.
In particular, it suffices to let \( p \) range over points in \( P_A \).
Since a gauge transformation \( \lambda \) in the stabilizer of \( A \) is necessarily constant on \( P_A \), the evaluation map \( \ev_{\lpnt} \) restricts to an isomorphism of Lie groups
\begin{equation}
	\label{eq:yangMillsHiggs:stabilizerAPhi}
 	\GauGroup_{A, \varphi}
 		= \GauGroup_{A} \intersect \GauGroup_{\varphi}
 		\isomorph \centralizer_G (\HolGroup_A) \intersect \bigIntersection_{p \in P_A} G_{\varphi(p)} \,.
\end{equation}
To summarize, by point (i) of \cref{prop:cotangentBundle:singularSympRed}, we have completely determined the orbit types of \( \SectionSpaceAbb{P} = \SectionMapAbb{J}^{-1}(0) \subseteq \CotBundle \SectionMapAbb{Q} \) with respect to the lifted \(\GauGroup \)-action.

Finally, let us outline two procedures to determine the orbit types of \( \CotBundle \SectionMapAbb{Q} \).
First, using the equivariant identification \( \CotBundle \SectionMapAbb{Q} \isomorph \TBundle \SectionMapAbb{Q} \) and following \parencite{Hertsch2008}, we obtain
\begin{equation}
	\GauGroup_{A, D} \isomorph \centralizer_G(\HolGroup_A \union \HolGroup_{A + D}).
\end{equation}
Thus, by an analogue of \cref{prop:yangMillsHiggs:orbitTypesIsomorphicToHoweSubbundles}, the orbit types of \( \TBundle \ConnSpace \) are in one-to-one correspondence with the isomorphism classes of bundles
\begin{equation}
	P_A \cdot \centralizer_G^2(\HolGroup_A \union \HolGroup_{A + D}).
\end{equation}
We note that these bundles are Howe but in general they are not holonomy-induced (\ie of the form \( Q \cdot \centralizer_G^2(H) \) where \( Q \) is a connected \( H \)-bundle).
This is in accordance with the general fact that \( \CotBundle \SectionMapAbb{Q} \) may have more orbit types than \( \SectionMapAbb{Q} \).

Second, more explicitly, a point \( (A, D) \in \ConnSpace \times \Omega^2(M, \CoAdBundle P) \) in the cotangent bundle \( \CotBundle \ConnSpace \) has stabilizer
\begin{equation}
	\GauGroup_{A, D} = \GauGroup_{A} \intersect \GauGroup_D,
\end{equation}
where \( \GauGroup_{A} \) and \( \GauGroup_D \) denote the stabilizers of \( A \) and \( D \) under the action of the gauge group, respectively.
	
To analyze \( \GauGroup_{A, D} \), for a moment, consider the stabilizer of a general \( k \)-form \( \alpha \) with values in \( F = P \times_G \FibreBundleModel{F} \), where \( \FibreBundleModel{F} \) carries a \( G \)-representation.
For \( \alpha \in \DiffFormSpace^k(M, F) \) and \( p \in P \), let \( V_\alpha (p) \) be the subspace of \( \FibreBundleModel{F} \) spanned by all elements of the form \( \alpha_p (X_1, \dotsc, X_k) \), where \( X_i \in \TBundle_p P \).
Similar arguments as above for \( \alpha = \varphi \) yield the following identification
\begin{equation}
	\label{eq:yangMillsHiggs:stabConnAlpha}
	\GauGroup_{A, \alpha}
		\isomorph \centralizer_G (\HolGroup_A) \intersect \bigIntersection_{\substack{p \in P_A \\ \mathclap{f \in V_\alpha (p)}}} G_f \,.
\end{equation}
Now let us return to the stabilizer of \( D \).
Denote \( K \equiv \HolGroup_A \).
Since the adjoint action of \( \centralizer_G (K) \) on the Lie algebra \( \LieA{k} \) of \( K \) is trivial, \( \LieA{k} \subseteq \LieA{g} \) is \( \AdAction_{\centralizer_G (K)} \)-invariant and thus has a \( \centralizer_G (K) \)-invariant complement \( \LieA{p} \) in \( \LieA{g} \).
For \( \nu \in \LieA{k}^* \), we have 
\begin{equation}
	\dualPair{\CoAdAction_g \nu}{\xi} 
		= \dualPair{\nu}{\AdAction^{-1}_g \xi}
		= \dualPair{\nu}{\xi}
\end{equation}
for every \( g \in \centralizer_G (K) \) and \( \xi \in \LieA{k} \).
Hence, \( G_\nu \subseteq \centralizer_G (K) \) for every \( \nu \in \LieA{k}^* \).
In summary, we have shown:
\begin{prop}
	Let \( A \in \ConnSpace \) and \( D \in \CotBundle_A \ConnSpace \).
	Then,
	\begin{equation}
		\label{eq:yangMillsHiggs:stabAD}
		\GauGroup_{A, D}
			\isomorph \centralizer_G (\HolGroup_A) \intersect \bigIntersection_{\substack{p \in P_A \\ \mathclap{\nu \in V_{D}(p) \intersect \LieA{p}^*}}} G_\nu \,,
	\end{equation}
	where \( V_D(p) \) is the subspace of \( \LieA{g}^* \) spanned by all elements of the form \( D_p (X_1, X_2) \) for \( X_i \in \TBundle_p P \).
\end{prop}

\subsection{Description of orbit types after symmetry breaking}
Now, let us pass to a description of the model after symmetry breaking.
As usual in the physics literature, let us assume that \( \varphi \) takes values in one fixed orbit type \( \FibreBundleModel{F}_{(K)} \) of the \( G \)-action for some stabilizer subgroup \( K \).
Assume, moreover, that the bundle \( \FibreBundleModel{F}_{(K)} \to \check{\FibreBundleModel{F}}_{(K)} = \FibreBundleModel{F}_{(K)} \slash G \) is trivial and choose a smooth section \( f_0: \check{\FibreBundleModel{F}}_{(K)} \to \FibreBundleModel{F}_{(K)} \) which takes values in the subset \( \FibreBundleModel{F}_K \) of isotropy type \( K \).
With respect to this choice, we can write the Higgs field \( \varphi \), viewed as a \( G \)-equivariant map \( P \to \FibreBundleModel{F} \), as
\begin{equation}
	\label{eq:yangMillsHiggs:decompositionHiggsField}
	\varphi(p) = \phi(p) \cdot f_0 (\eta(p)), \quad p \in P,
\end{equation}
where \( \phi: P \to G \) and \( \eta: P \to \check{\FibreBundleModel{F}}_{(K)} \) are smooth maps.
Since \( f_0 \) takes values in \( \FibreBundleModel{F}_K \), the map \( \phi \) is uniquely defined when viewed as a map with values in \( G \slash K \).
Moreover, by \( G \)-equivariance of \( \varphi \), we identify \( \phi \) as a section of \( P \times_G G \slash K \) and \( \eta \) as a smooth map \( M \to \check{\FibreBundleModel{F}}_{(K)} \).
It is straightforward to see that the decomposition~\eqref{eq:yangMillsHiggs:decompositionHiggsField} depends smoothly on \( \varphi \) and hence establishes a diffeomorphism
\begin{equation}
	\label{eq:yangMillsHiggs:decompositionHiggsFieldSpace}
	\SectionSpaceAbb{F} \to \sSectionSpace(P \times_G G \slash K) \times \sFunctionSpace(M, \check{\FibreBundleModel{F}}_{(K)}),
	\qquad \varphi \mapsto (\phi, \eta)
\end{equation}
of Fréchet manifolds.
In geometric terms, \( \phi \) yields a reduction of \( P \) to the principal \( K \)-bundle
\begin{equation}
	\hat{P} \defeq \set{p \in P \given \phi(p) = \equivClass{e}}.
\end{equation}
Next, recall the following geometric version of the Higgs mechanism (\cf \parencite[Proposition~7.3.4]{RudolphSchmidt2014}):
Since \( K \) is compact, there exists an \( \AdAction_K \)-invariant decomposition \( \LieA{g} = \LieA{k} \oplus \LieA{p} \).
Accordingly, the restriction of \( A \) to the \( K \)-bundle \( \hat{P} \) splits into\footnote{In the physics language, \( \hat{A} \) is the reduced gauge field and \( \tau \) is the intermediate vector boson. The surviving Higgs field is \( \eta \), or rather \( \eta \) shifted by the Higgs vacuum.} 
\begin{equation}
	\label{eq:yangMillsHiggs:fieldsAfterSymBreak}
	\restr{A}{\hat{P}} = \hat{A} + \tau,
\end{equation}
where \( \hat{A} \) and \( \tau \) take values in \( \LieA{k} \) and \( \LieA{p} \), respectively.
It is an easy exercise to verify that \( \hat{A} \) is a principal \( K \)-connection in \( \hat{P} \) and \( \tau \) a horizontal \( 1 \)-form of type \( \AdBundle_K \LieA{p} \) on \( \hat{P} \), \ie, \( \tau \in \DiffFormSpace^1(M, \hat{P} \times_K \LieA{p}) \).
Let us combine this decomposition of \( A \) with the diffeomorphism of~\eqref{eq:yangMillsHiggs:decompositionHiggsFieldSpace}.
For this purpose, consider the smooth Fréchet bundle \( \SectionSpaceAbb{E} \to \sSectionSpace(P \times_G G \slash K) \), whose fiber over \( \phi \) is \( \ConnSpace(\hat{P}) \times \DiffFormSpace^1(M, \hat{P} \times_K \LieA{p}) \).
The bundle \( \SectionSpaceAbb{E} \) carries a natural action\footnote{Note that the action of \( \GauGroup \) on \( \SectionSpaceAbb{E} \) mixes the variables \( \hat{A} \) and \( \tau \), because the decomposition \( \LieA{g} = \LieA{k} \oplus \LieA{p} \) is not \( \AdAction_G \)-invariant.} of \( \GauGroup \) as \( \hat{P} \) is a subbundle of \( P \).
To summarize we obtain the following.
\begin{prop}
	\label{prop:yangMillsHiggs:configSpaceSymmBreaking}
	The map \( (A, \varphi) \mapsto (\phi, \hat{A}, \tau, \eta) \) defines a \( \GauGroup \)-equivariant diffeomorphism of \( \SectionSpaceAbb{Q} \) with \( \SectionSpaceAbb{E} \times \sFunctionSpace(M, \check{\FibreBundleModel{F}}_{(K)}) \).
	In particular, we get an isomorphism of stratified spaces between the gauge orbit space \( \check{\SectionSpaceAbb{Q}} \) and \( \SectionSpaceAbb{E} \slash \GauGroup \times \sFunctionSpace(M, \check{\FibreBundleModel{F}}_{(K)}) \).
\end{prop}
When the bundle \( \hat{P} \) is non-trivial, the field \( \phi \) representing \( \hat{P} \) carries topological data, which may be encoded in various ways, \eg in terms of Chern classes.

Using this proposition, the gauge orbit types of \( \SectionSpaceAbb{Q} \) may be characterized in the following more explicit way. 
First, note that \( \eta \) does not contribute to the orbit type structure.
Next, recall that \( \lambda \in \GauGroup \) preserves \( \varphi \) if and only if \( \lambda(p) \cdot \varphi(p) = \varphi(p) \) for all \( p \in P \).
Thus, every \( \lambda \in \GauGroup_\varphi \) restricts to a \( K \)-gauge transformation on \( \hat{P} \).
In fact, a moment's reflection shows that every \( K \)-gauge transformation on \( \hat{P} \) can be obtained in that way (use \( P \times_G K \isomorph \hat{P} \times_K K \)). 
This yields the following.
\begin{lemma}
	\label{prop:yangMillsHiggs:stabPhi}
	For every \( \varphi \) with associated \( K \)-bundle \( \hat{P} \), we have
	\begin{equation}
		\GauGroup_\varphi = \hat{\GauGroup},
	\end{equation}
	where \( \hat{\GauGroup} \) denotes the group of gauge transformations of \( \hat{P} \).
\end{lemma}
Let \( (A, \varphi) \in \SectionSpaceAbb{Q} \).
By equivariance and~\eqref{eq:yangMillsHiggs:stabilizerAPhi}, we find\footnote{In the sequel, we assume that the base point \( p_0 \) defining \( P_A \) lies in \( \hat{P} \), which can be always assured by translation with a constant \( g \in G \).}
\begin{equation}
	\label{eq:yangMillsHiggs:stabilizerAPhiOneOrbitType}
	\GauGroup_{A, \varphi} 
		\isomorph \centralizer_G (\HolGroup_A) \intersect \bigIntersection_{p \in P_A} \phi(p) K \phi(p)^{-1}.
\end{equation}
It is interesting to compare the gauge orbit types \( \GauGroup_{A, \varphi} \) to the orbit types of the theory after symmetry reduction.
Using \cref{prop:yangMillsHiggs:configSpaceSymmBreaking}, we have
\begin{equation}
	\label{eq:yangMillsHiggs:stabAfterSymBreak}
	\GauGroup_{A, \varphi}
		= \hat{\GauGroup}_{\hat{A}, \tau}
		= \hat{\GauGroup}_{\hat{A}} \intersect \hat{\GauGroup}_{\tau}, 
\end{equation}
where \( \hat{\GauGroup}_{\tau} \) denotes the stabilizer of \( \tau \) under the \( \hat{\GauGroup} \)-action on \( \DiffFormSpace^1(M, \hat{P} \times_K \LieA{p}) \).
Indeed, by \cref{prop:yangMillsHiggs:stabPhi}, \( \GauGroup_{\varphi} \) is isomorphic to \( \hat{\GauGroup} \) and it is straightforward to see that a gauge transformation \( \lambda \in \GauGroup_{\varphi} \) leaves \( A \) invariant if and only if \( \restr{\lambda}{\hat{P}} \in \hat{\GauGroup} \) leaves \( \restr{A}{\hat{P}} = \hat{A} + \tau \) invariant. 
Furthermore, using~\eqref{eq:yangMillsHiggs:stabConnAlpha}, we obtain a more explicit description of \( \hat{\GauGroup}_{\tau} \).
Hence, in summary, we get the following.
\begin{prop}
	\label{prop:yangMillsHiggs:stabAPhi}
	The stabilizer of \( (A, \varphi) \) under the \( \GauGroup \)-action on \( \SectionSpaceAbb{Q} \) is given by
	\begin{equation+}
		\label{eq:yangMillsHiggs:stabAPhi}
		\GauGroup_{A, \varphi}
			\isomorph 
				\centralizer_K (\HolGroup_{\hat{A}})
				\intersect 
				\bigIntersection_{\substack{p \in P_{\hat{A}} \\ \mathclap{\xi \in V_\tau (p)}}} G_\xi \,.
			\qedhere
	\end{equation+}
\end{prop}

Finally, let us consider a special case which is important for instance in the theory of magnetic monopoles \parencite[Chapter~7]{RudolphSchmidt2014}.
It is defined by the additional assumption that \( \phi \) be covariantly constant with respect to \( A \).
Then, \( A \) reduces to a connection \( \hat{A} \) on \( \hat{P} \) and so \( \tau = 0 \).
Thus, in this case,~\eqref{eq:yangMillsHiggs:stabAPhi} simplifies to
\begin{equation}
	\label{eq:yangMillsHiggs:stabWithMatterCovConst}
	\GauGroup_{A, \varphi}
		\isomorph \hat{\GauGroup}_{\hat{A}}
		\isomorph \centralizer_K (\HolGroup_{\hat{A}}).
\end{equation}
As an immediate consequence, we obtain the following analogue of \cref{prop:yangMillsHiggs:orbitTypesIsomorphicToHoweSubbundles}.
\begin{prop}
	The orbit types of the action of \( \GauGroup \) on \( \SectionSpaceAbb{Q} \) at points \( (A, \varphi) \) with \( \dif_A \phi = 0 \) are in one-to-one correspondence with isomorphism classes of holonomy-induced Howe subbundles of \( K \)-reductions \( \hat{P} \) of \( P \) via the map
	\begin{equation+}
		\equivClass{\GauGroup_{A, \varphi}} \mapsto \equivClass{\hat{P}_{\hat{A}} \cdot \centralizer_K^2(\HolGroup_{\hat{A}})}.
		\qedhere
	\end{equation+}
\end{prop}
\begin{proof}
	By~\eqref{eq:yangMillsHiggs:stabWithMatterCovConst}, the stabilizer \( \GauGroup_{A, \varphi} \) of a point \( (A, \varphi) \in \SectionSpaceAbb{Q} \) satisfying \( \dif_A \phi = 0 \) is isomorphic to the stabilizer of \( \hat{A} \) under \( \hat{\GauGroup} \).
	Now recall from \cref{prop:yangMillsHiggs:orbitTypesIsomorphicToHoweSubbundles} that orbit types of connections on \( \hat{P} \) are in bijective correspondence with isomorphism classes of holonomy-induced Howe subbundles of \( \hat{P} \).
\end{proof}

\subsection{Reduced symplectic structure}
\label{sec:yangMillsHiggs:symplectic}

Finally, let us come to point (iii) of \cref{prop:cotangentBundle:singularSympRed}.
The symplectic structure \( \Omega \) on \( \CotBundle \SectionSpaceAbb{Q} \) is defined by
\begin{equation}\begin{split}
	\Omega_{A, \varphi, D, \Pi}&\left((\diF A_1, \diF \varphi_1, \diF D_1, \diF \Pi_1), (\diF A_2, \diF \varphi_2, \diF D_2, \diF \Pi_2)\right) 
		\\
		&= \int_M \wedgeDual{\diF D_1}{\diF A_2} - \wedgeDual{\diF D_2}{\diF A_1} + \wedgeDual{\diF \Pi_1}{\diF \varphi_2} - \wedgeDual{\diF \Pi_2}{\diF \varphi_1},
\end{split}\end{equation}
where \( (\diF A_j, \diF D_j, \diF \varphi_j, \diF \Pi_j) \in \TBundle_{(A, D, \varphi, \Pi)} (\CotBundle Q) \) for \( j=1,2 \).
\begin{lemma}
	\label{prop:yangMillsHiggs:reduction:orbitsSymplecticallyClosed}
	The orbit \( \GauAlgebra \ldot (A, D, \varphi, \Pi) \) is symplectically closed in \( \CotBundle \SectionSpaceAbb{Q} \) for every \( (A, D, \varphi, \Pi) \in \CotBundle \SectionSpaceAbb{Q} \), that is,
	\begin{equation+}
		\left(\GauAlgebra \ldot (A, D, \varphi, \Pi)\right)^{\Omega \Omega} = \GauAlgebra \ldot (A, D, \varphi, \Pi).
		\qedhere
	\end{equation+}
\end{lemma}
\begin{proof}
	First, note that
	\begin{equation}
		\SectionSpaceAbb{j}(\diF A, \diF D, \diF \varphi, \diF \Pi) 
			\defeq (- \hodgeStar \diF D, \hodgeStar \diF A, - \hodgeStar \diF \Pi, \hodgeStar \diF \varphi)
	\end{equation}
	defines an almost complex structure \( \SectionSpaceAbb{j} \) on \( \CotBundle Q \), which intertwines the symplectic structure \( \Omega \) with the \( \LTwoFunctionSpace \)-scalar product.
	Clearly, for every vector subspace \( W \subseteq \TBundle_{(A, D, \varphi, \Pi)} (\CotBundle Q) \) we have \( W^\Omega = \SectionSpaceAbb{j} W^\perp \), where \( W^\perp \) is the \( \LTwoFunctionSpace \)-orthogonal complement.
	Hence,
	\begin{equation}
		W^{\Omega \Omega} 
			= \SectionSpaceAbb{j} ( \SectionSpaceAbb{j} W^\perp)^\perp
			= W^{\perp \perp}
	\end{equation}
	and it remains to show that \( W^{\perp \perp} = W \) for \( W = \GauAlgebra \ldot (A, D, \varphi, \Pi) \).
	By the bipolar theorem \parencite[Theorem~4.1.5]{Schaefer1971}, this holds if and only if the orbit is closed with respect to the weak \( \LTwoFunctionSpace \)-topology.
	Note that the infinitesimal action has the character of a multiplication operator in the \( D \)- and \( \varphi \)-direction and hence in these components the orbit is not closed (it may actually be dense).
	Nonetheless, the `diagonal' orbit \( \GauAlgebra \ldot (A, D, \varphi, \Pi) \) is closed as we will show now.

	Let \( (\diF A, \diF D, \diF \varphi, \diF \Pi) \) be an element of \( \TBundle_{(A, D, \varphi, \Pi)} (\CotBundle Q) \) that does not lie on the \( \GauAlgebra \)-orbit.
	It will be convenient to write \( \alpha \equiv \diF A \) and abbreviate \( (\diF D, \diF \varphi, \diF \Pi) \) by \( \beta \).
	We have to construct an \( \LTwoFunctionSpace \)-open neighborhood \( \SectionSpaceAbb{U} \) of \( (\alpha, \beta) \) in \( \TBundle_{(A, D, \varphi, \Pi)} (\CotBundle \SectionSpaceAbb{Q}) \) that is disjoint from the \( \GauAlgebra \)-orbit.
	First, suppose that \( \alpha \) cannot be written as \( \dif_A \xi \) for some \( \xi \in \GauAlgebra \).
	Since the decomposition~\eqref{eq:yangMillsHiggs:decomposeTangent} is orthogonal with respect to the \( \LTwoFunctionSpace \)-scalar product, the orbit \( \GauAlgebra \ldot A \) is closed in \( \TBundle_A \ConnSpace \) with respect to the weak \( \LTwoFunctionSpace \)-topology.
	Thus, there exists an \( \LTwoFunctionSpace \)-open neighborhood \( \SectionSpaceAbb{U}_\alpha \) of \( \alpha \) in \( \DiffFormSpace^1(M, \AdBundle P) \) which is disjoint from \( \GauAlgebra \ldot A \).
	Then, 
	\begin{equation}
		\SectionSpaceAbb{U} \defeq \SectionSpaceAbb{U}_\alpha \times \DiffFormSpace^2(M, \CoAdBundle P) \times \sSectionSpace(F) \times \DiffFormSpace^3(M, F^*)
	\end{equation}
	is an \( \LTwoFunctionSpace \)-open neighborhood of \( (\alpha, \beta) \) in \( \TBundle_{(A, D, \varphi, \Pi)} (\CotBundle Q) \) which does not intersect the \( \GauAlgebra \)-orbit.
	Now suppose that \( \alpha = \dif_A \xi \) for some \( \xi \in \GauAlgebra \).
	Then, \( \xi \) is uniquely determined up to an element of the finite-dimensional stabilizer \( \GauAlgebra_A \).
	Since the orbit \( \GauAlgebra_A \ldot (D, \varphi, \Pi) \) is finite-dimensional, it is automatically closed.
	Hence, there exists an \( \LTwoFunctionSpace \)-open neighborhood \( \SectionSpaceAbb{U}_\beta \) of \( \beta \) and an \( \LTwoFunctionSpace \)-open subset \( \SectionSpaceAbb{V} \) containing the orbit \( \GauAlgebra_A \ldot (D, \varphi, \Pi) \) such that \( \SectionSpaceAbb{U}_\beta \) and \( \SectionSpaceAbb{V} \) have empty intersection.
	Let \( \SectionSpaceAbb{W} \subseteq \GauAlgebra \) denote the inverse image of \( \SectionSpaceAbb{V} \) under the action \( \xi \to \xi \ldot (D, \varphi, \Pi) \).
	Note that \( \GauAlgebra_A \subseteq \SectionSpaceAbb{W} \).
	Now,
	\begin{equation}
		\SectionSpaceAbb{U}_\alpha \defeq \dif_A \SectionSpaceAbb{W} \oplus \ker \dif^*_A
	\end{equation}
	is an \( \LTwoFunctionSpace \)-open neighborhood of \( \alpha \) in \( \DiffFormSpace^1(M, \AdBundle P) \).
	Suppose \( \alpha' \in \SectionSpaceAbb{U}_\alpha \) and \( \beta' \in \SectionSpaceAbb{U}_\beta \) are of the form \( \alpha' = \dif_A \xi' \) and \( \beta' = \xi' \ldot \beta \) for some \( \xi' \in \GauAlgebra \).
	Then, \( \xi' \in \SectionSpaceAbb{W} \) and thus \( \beta' = \xi' \ldot \beta \in \SectionSpaceAbb{V} \).
	However, by assumption, \( \beta' \) is an element of \( \SectionSpaceAbb{U}_\beta \).
	Since the latter is disjoint from \( \SectionSpaceAbb{V} \), we have constructed a contradiction.
	In summary, \( \SectionSpaceAbb{U} \defeq \SectionSpaceAbb{U}_\alpha \times \SectionSpaceAbb{U}_\beta \) is an \( \LTwoFunctionSpace \)-open neighborhood of \( (\alpha, \beta) \) in \( \TBundle_{(A, D, \varphi, \Pi)} (\CotBundle Q) \), which has empty intersection with the \( \GauAlgebra \)-orbit.
\end{proof}
By this lemma, the strata of the reduced phase space inherit a symplectic form from \( (\CotBundle \SectionSpaceAbb{Q}, \Omega) \), according to point (iii) of \cref{prop:cotangentBundle:singularSympRed}.
Moreover, we have shown that \cref{prop:cotangentBundle:singularCotangentBundleRed} holds for Yang--Mills--Higgs theory and thus, in particular, every symplectic stratum further decomposes into seams and a cotangent bundle.
This secondary stratification will be further analyzed below in the concrete example of the Glashow--Weinberg--Salam model.

\section{Example: Glashow--Weinberg--Salam model}
\label{GWS-model}
We now specialize the discussion to the Higgs sector of the standard model of electroweak interactions.
As in the general setting, the configurations of this model are pairs \( (A, \varphi) \) consisting of a connection in a principal bundle \( P \) and a section of an associated vector bundle \( P \times_G \FibreBundleModel{F} \).
For this model, the principal bundle \( P \) is an \( \SUGroup(2) \times \UGroup(1) \)-bundle over \( M \) and the associated bundle \( F \) has typical fiber \( \FibreBundleModel{F} =\C^2 \) carrying the following representation of \( G = \SUGroup(2) \times \UGroup(1) \):
\begin{equation}
	\label{eq:yangMillsHiggs:GWS:representation}
	\rho_{a, \vartheta} (z_1, z_2) = a \cdot \E^{\frac{\I}{2} \vartheta} \cdot (z_1, z_2), \qquad a \in \SUGroup(2), \vartheta \in [0, 4\pi).
\end{equation}
The Higgs potential \( \FibreBundleModel{V}: \C^2 \to \R \) has the form
\begin{equation}
	\FibreBundleModel{V}(f) = \lambda \left(\norm{f}^2 - \frac{\nu^2}{2} \right)^2
\end{equation}
for given \( \lambda > 0 \) and non-zero \( \nu \in \R \).

It is straightforward to see that under the representation \( \rho \) the origin is a fixed point and that all other points have a stabilizer conjugate to
\begin{equation}
	\label{eq:yangMillsHiggs:GWS:stabilizer}
	K \defeq \set*{\left( \Matrix{\E^{\frac{\I}{2} \vartheta} & 0 \\ 0 & \E^{-\frac{\I}{2} \vartheta}}, \E^{\I \vartheta} \right) \given \vartheta \in [0, 4\pi)},
\end{equation}
which is isomorphic to \( \UGroup(1) \) and plays the role of the electromagnetic gauge group after symmetry breaking.
As common in the physics literature, we assume that \( \varphi \) vanishes nowhere, that is, we ignore the singular orbit type in \( \FibreBundleModel{F} \).
The generic orbits in \( \FibreBundleModel{F} \) are three-spheres centered at the origin and hence the quotient \( \check{\FibreBundleModel{F}}_{(K)} = \FibreBundleModel{F}_{(K)} \slash G \) is diffeomorphic to \( \R_{>0} \).
All the points \( \frac{r}{\sqrt{2}} (0, \nu) \) for \( r \in \R_{>0} \) have stabilizer \( K \).
Hence, the map \( f_0: r \mapsto \frac{r}{\sqrt{2}} (0, \nu) \) is a smooth section of \( \FibreBundleModel{F}_{(K)} \to \check{\FibreBundleModel{F}}_{(K)} \) taking values in \( \FibreBundleModel{F}_K \).
Accordingly, the decomposition~\eqref{eq:yangMillsHiggs:decompositionHiggsField} of \( \varphi \) simplifies here to
\begin{equation}
	\label{eq:yangMillsHiggs:GWS:decompositionHiggsField}
	\varphi = \frac{\eta}{\sqrt{2}} \, \phi \cdot \Vector{0 \\ \nu},
\end{equation}
where \( \phi \in \sSectionSpace(P \times_G G \slash K) \) and \( \eta \in \sFunctionSpace(M, \R_{> 0}) \).
Note that, in this presentation, \( V(\varphi) = \frac{\lambda \nu^2}{2} (\eta^2 - 1)^2 \).

\subsection{Orbit types of \texorpdfstring{$ \SectionSpaceAbb{Q} $}{Q}}
\label{sec:yangMillsHiggs:GWS:orbitTypesQ}
According to the general program, we have to determine the Howe subgroups of \( \SUGroup(2) \times \UGroup(1) \).
For that purpose, we use the following elementary result, whose proof is a simple calculation that we leave to the reader.
\begin{lemma}
	Let \( G \) be a group and \( L \) be an abelian group.
	A subgroup \( H \) of \( G \times L \) is Howe if and only if there exists a subgroup \( H'_G \) of \( G \) such that 
	\begin{equation+}
		\label{eq:yangMillsHiggs:howeInDirectProduct}
		H = \centralizer_G (H'_G) \times L.
		\qedhere
	\end{equation+}
\end{lemma}
According to this lemma, we first have to determine the Howe subgroups of \( \SUGroup(2) \).
Clearly, each Howe subgroup of \( \SUGroup(2) \) is conjugate to the centralizer \( \Z_2 \), to \( \SUGroup(2) \) itself or to \( \UGroup(1) \) (seen as a subgroup via the embedding \( \alpha \mapsto \smallMatrix{\alpha & 0 \\ 0 & \bar{\alpha}} \)).
This corresponds to the trivial group or \( \Z_2 \), \( \SUGroup(2) \) and \( \UGroup(1) \) as holonomy groups, respectively, see \cref{table:gaugeTheory:gaugeOrbitTypes:su}.
Correspondingly, the Howe subgroups of \( \SUGroup(2) \times \UGroup(1) \) are conjugate to \( \SUGroup(2) \times \UGroup(1) \), \( \UGroup(1) \times \UGroup(1) \) or \( \Z_2 \times \UGroup(1) \).
\begin{table}[tbp]
	\centering
	\begin{tabular}{l l l}
		\toprule
			\( \HolGroup_A \) &
			\( \GauGroup_A \) &
			\( H \)
			\\
		\midrule
			\( \set{e} \), \( \Z_2 \) &
			\( \SUGroup(2) \) &
			\( \Z_2 \)
			\\
			\( \UGroup(1) \) & 
			\( \UGroup(1) \) & 
			\( \UGroup(1) \)
			\\
			\( \SUGroup(2) \) &
			\( \Z_2 \) &
			\( \SUGroup(2) \)
			\\
		\bottomrule
	\end{tabular}
	\caption{List of all possible holonomy groups for \( \SUGroup(2) \) up to conjugacy with the corresponding stabilizer \( \GauGroup_A = \centralizer_G (\HolGroup_A) \) and the Howe subgroup \( H = \centralizer^2_G (\HolGroup_A) \).}
	\label{table:gaugeTheory:gaugeOrbitTypes:su}
\end{table}

Even in the case when~\eqref{eq:yangMillsHiggs:howeInDirectProduct} uniquely determines \( H'_G \), there is still room for different subgroups \( H' \) of \( G \times L \) with \( H = \centralizer_{G \times L} (H') \).
Indeed, recall that by Goursat's lemma subgroups of \( G \times L \) are in bijection with quintuples \( (G_1, G_2, L_1, L_2, \theta) \), where \( G_2 \normalSubgroupEq G_1 \subseteq G \), \( L_2 \normalSubgroupEq L_1 \subseteq L \) and \( \theta: G_1 \slash G_2 \to L_1 \slash L_2 \) is an isomorphism.
Such a tuple yields the subgroup
\begin{equation}
	H' = \set{(g, l) \in G_1 \times L_1 \given \theta(g \, G_2) = l \, L_2} \subseteq G \times L.
\end{equation}
Note that the projection of \( H' \) to the \( G \)-factor coincides with \( G_1 \).
Hence the knowledge of \( H'_G \) just determines \( G_1 = H'_G \).
We now determine the possible choices for the other elements \( (G_2, L_1, L_2, \theta) \) that generate the Howe subgroups of \( \SUGroup(2) \times \UGroup(1) \), as summarized in \cref{table:yangMillsHiggs:GWS:orbitTypes}.
\todoResearch{\( H' \) is of course the holonomy group. What is the significants of the elements in the tuple \( (G_1, G_2, L_1, L_2, \theta) \)? Given \( H' \) they are defined as follows. \( G_1 = \pr_G (H') \) and \( G_2 = \set{g \in G (g, e) \in H'} \); similarly for \( L_1 \) and \( L_2 \). The isomorphism \( \theta \) is defined by \( \theta(g \, G_2) = l \, L_2 \) for \( (g, l) \in H' \).}
\begin{itemize}
	\item 
		The Howe subgroup \( \SUGroup(2) \) is generated by \( G_1 = \set{e} \) or \( G_1 = \Z_2 \).
		In the first case, we hence have \( G_2 = \set{e} \) and so \( L_1 = L_2 \) with \( \theta \) being trivial.
		In particular, \( L_1 = L_2 \) has to be either \( \set{e} \), \( \UGroup(1) \) or the cyclic group \( \Z_p \) for some \( p \in \N \), since these are the only subgroups of \( \UGroup(1) \).
		For \( G_1 = \Z_2 \) there are two cases:
		First we may choose \( G_2 = \Z_2 \), which then again requires \( L_1 = L_2 \).
		Secondly, also \( G_2 = \set{e} \) is possible.
		Then \( L_1 \slash L_2 \) has to be isomorphic to \( \Z_2 \), which is only possible\footnote{Note that \( \Z_p \slash \Z_q \) is isomorphic to \( \Z_{p \slash q} \).} if \( L_1 = \Z_{2p} \) and \( L_2 = \Z_{p} \) for some \( p \in \N \).
	\item
		The Howe subgroup \( \UGroup(1) \) is generated by \( G_1 = \UGroup(1) \).
		There are two non-trivial choices for \( G_2 \).
		First, \( G_2 = \UGroup(1) \) leads again to \( L_1 = L_2 \).
		The second choice \( G_2 = \Z_q \) for some \( q \in \N \) enforces \( L_1 = \UGroup(1) \) and \( L_2 = \Z_p \).
		Since the map \( z \mapsto z^q \) induces an isomorphism of \( \UGroup(1) \slash \Z_q \) with \( \UGroup(1) \) and the only automorphisms of \( \UGroup(1) \) are of the form \( z \mapsto z^k \) for some \( k \in \N \), isomorphisms \( \UGroup(1) \slash \Z_q \to \UGroup(1) \slash \Z_p \) are necessarily induced by maps of the form \( z \mapsto z^{\frac{k q}{p}} \) for some \( k \in \N \).
	\item 
		The Howe subgroup \( \Z_2 \) is generated by \( G_1 = \SUGroup(2) \) giving rise to a plethora of possible choices for \( G_2 \), \( L_1 \) and \( L_2 \).
		Fortunately, this case will be of no further interest below, so we do not need to dive into details.
\end{itemize}
\begin{table}[tbp]
	\centering
	\begin{tabular}{l  l l l l l}
		\toprule
			&
			\multicolumn{5}{c}{\( H' \)}
			\\
			\cmidrule(r){2-6}
			\( H \) &
			\( G_1 \) &
			\( G_2 \) &
			\( L_1 \) &
			\( L_2 \) &
			\( \theta \)
			\\
		\midrule
			\( \SUGroup(2) \times \UGroup(1) \) &
			\( \set{e} \) &
			\( \set{e} \) &
			\( \UGroup(1) \) &
			\( \UGroup(1) \) &
			trivial
			\\
			&
			\( \set{e} \) &
			\( \set{e} \) &
			\( \Z_p \) &
			\( \Z_p \) &
			trivial
			\\
			&
			\( \set{e} \) &
			\( \set{e} \) &
			\( \set{e} \) &
			\( \set{e} \) &
			trivial
			\\
			&
			\( \Z_2 \) &
			\( \Z_2 \) &
			\( \UGroup(1) \) &
			\( \UGroup(1) \) &
			trivial
			\\
			&
			\( \Z_2 \) &
			\( \Z_2 \) &
			\( \Z_p \) &
			\( \Z_p \) &
			trivial
			\\
			&
			\( \Z_2 \) &
			\( \Z_2 \) &
			\( \set{e} \) &
			\( \set{e} \) &
			trivial
			\\
			&
			\( \Z_2 \) &
			\( \set{e} \) &
			\( \Z_{2p} \) &
			\( \Z_p \) &
			\( \id_{\Z_2} \)
			\\
			\( \UGroup(1) \times \UGroup(1) \) &
			\( \UGroup(1) \) &
			\( \UGroup(1) \) &
			\( \UGroup(1) \) &
			\( \UGroup(1) \) &
			trivial
			\\
			&
			\( \UGroup(1) \) &
			\( \UGroup(1) \) &
			\( \Z_p \) &
			\( \Z_p \) &
			trivial
			\\
			&
			\( \UGroup(1) \) &
			\( \UGroup(1) \) &
			\( \set{e} \) &
			\( \set{e} \) &
			trivial
			\\
			&
			\( \UGroup(1) \) &
			\( \Z_q \) &
			\( \UGroup(1) \) &
			\( \Z_p \) &
			\( z \mapsto z^{\frac{kq}{p}} \) for some \( k \in \N \)
			\\
			&
			\( \UGroup(1) \) &
			\( \set{e} \) &
			\( \UGroup(1) \) &
			\( \set{e} \) &
			\( z \mapsto z^k \) for some \( k \in \N \)
			\\
			\( \Z_2 \times \UGroup(1) \) &
			\multicolumn{5}{c}{many choices}
			\\
		\bottomrule
	\end{tabular}
	\caption{List of all Howe subgroups \( H \) of \( \SUGroup(2) \times \UGroup(1) \) up to conjugacy with the corresponding generator \( H' \) satisfying \( \centralizer_{\SUGroup(2) \times \UGroup(1)} (H') = H \).}
	\label{table:yangMillsHiggs:GWS:orbitTypes}
\end{table}
According to \cref{prop:yangMillsHiggs:orbitTypesIsomorphicToHoweSubbundles}, we have accomplished the algebraic part of the classification of stabilizer subgroups \( \GauGroup_A \).
Depending on the topologies of \( M \) and \( P \), some of the Howe subgroups \( H \) may possibly not occur in the final classification.

According to~\eqref{eq:yangMillsHiggs:stabilizerAPhiOneOrbitType}, in order to calculate the stabilizer \( \GauGroup_{A,\varphi} \), we need to determine which conjugates of \( K \) intersect \( K \) again.
Writing \( a \in \SUGroup(2) \) as \( a = \smallMatrix{\alpha & -\bar{\beta} \\ \beta & \bar{\alpha}} \) with \( \alpha, \beta \in \C \) such that \( \abs{\alpha}^2 + \abs{\beta}^2 = 1 \), we find
\begin{equation}
	a \Matrix{\E^{\frac{\I \vartheta}{2}} & 0 \\ 0 & \E^{-\frac{\I \vartheta}{2}}} a^{-1} 
		= \Matrix{
			\abs{\alpha}^2 \E^{\frac{\I \vartheta}{2}} + \abs{\beta}^2 \E^{- \frac{\I \vartheta}{2}} 
			& - \alpha \bar{\beta} \left(\E^{- \frac{\I \vartheta}{2}} - \E^{\frac{\I \vartheta}{2}}\right)
			\\
			\bar{\alpha} \beta \left(\E^{\frac{\I \vartheta}{2}} - \E^{-\frac{\I \vartheta}{2}}\right)
			& \abs{\alpha}^2 \E^{-\frac{\I \vartheta}{2}} + \abs{\beta}^2 \E^{\frac{\I \vartheta}{2}}
			}.
\end{equation}
Hence \( \left( a \smallMatrix{\E^{\frac{\I \vartheta}{2}} & 0 \\ 0 & \E^{-\frac{\I \vartheta}{2}}} a^{-1}, \E^{\I \vartheta} \right) \) is an element of \( K \) if and only if either \( \beta = 0 \) or \( \E^{\I \vartheta} = 1 \).
Thus,
\begin{equation}
	\bigIntersection_{p \in P} \phi(p) K \phi(p)^{-1}
		= 
		\begin{cases}
			K 		& \text{if } \beta(p) = 0 \text{ for all } p \in P_A,
			\\
			\Z_2	& \text{otherwise},
		\end{cases}	
\end{equation}
where \( \Z_2 \) is viewed as the subgroup
\begin{equation}
	\Z_2 \isomorph \set*{\left(\Matrix{1 & 0 \\ 0 & 1}, 1 \right), \left(\Matrix{-1 & 0 \\ 0 & -1}, 1 \right)}
\end{equation}
of \( K \).
Thus, we obtain the following.
\begin{prop}
	\label{prop:yangMillsHiggs:GWS:stabilizerConnHiggs}
	The common stabilizer \( \GauGroup_{A, \varphi} \) of \( (A, \varphi) \) is either \( K \) if \( \beta(p) = 0 \) for all \( p \in P_A \) and the stabilizer of \( A \) is \( \SUGroup(2) \times \UGroup(1) \) or \( \UGroup(1) \times \UGroup(1) \); or otherwise it is \( \Z_2 \).
\end{prop}

\newcommand{\weinbergAngle}{\theta_{\textup{W}}}
We now describe the structure of the orbit types in terms of the fields \( (\hat{A}, \tau) \) after symmetry breaking, see~\eqref{eq:yangMillsHiggs:fieldsAfterSymBreak}.
For this purpose, choose the following basis of \( \LieA{g} = \SUAlgebra(2) \times \UAlgebra(1) \): 
\begin{equation}
	\set*{t_a = \frac{\I}{2} \sigma_a, \I},
\end{equation}
where \( \sigma_a \) are the Pauli matrices.
In terms of these generators, the induced representation of \( \LieA{g} \), which will also be denoted by \( \rho \), is determined by
\begin{equation}
	\label{eq:yangMillsHiggs:GWS:representationAlgebra}
	\rho_{t_a} = t_a, \qquad \rho_{\I} = \frac{\I}{2} \one.
\end{equation}
Let \( t_\pm = t_3 \pm \I \).
Then, the Lie algebra \( \LieA{k} \) of \( K \) is spanned by \( t_+ \) and the complement \( \LieA{p} \) is spanned by \( \{ t_1, t_2, t_- \} \).
With respect to the basis \( \{t_a, \I \} \), we expand the connection as
\begin{equation}
	\label{eq:yangMillsHiggs:GWS:connectionInComponents}
	A = g W^a t_a + \I g' B,
\end{equation}
where we have introduced the coupling constants \( g \) and \( g' \).
Thus, passing to the basis \( \set{t_1, t_2, t_\pm} \) yields the decomposition \( A = \hat{A} + \tau \), where
\begin{subequations}
	\begin{align+}
		\hat{A} & = \frac{1}{2} (g W^3 + g' B) t_+,
		\\
		\tau & = g W^1 t_1 + g W^2 t_2 + \frac{1}{2} (g W^3 - g' B) t_-.
	\end{align+}
\end{subequations}
As common in the physics literature, it is convenient to introduce the fields
\begin{subequations}
	\label{eq:yangMillsHiggs:GWS:WpmZAgamma}
	\begin{align+}
		W_\pm &= \frac{1}{\sqrt{2}} (W^1 \mp \I W^2),
		\\
		\Vector{Z \\ A_\gamma} &= \frac{1}{\sqrt{g^2 + g'^2}} \Matrix{g & - g' \\ g' & g} \Vector{W^3 \\ B}.
	\end{align+}
\end{subequations}
Physically, \( A_\gamma \) is the electromagnetic field whose quanta are massless photons, \( Z \) describes the neutral but massive \( Z \)-boson, and \( W_\pm \) are massive charged bosons.
The \( Z \) and \( W_\pm \)-bosons mediate the weak interaction, while the photon corresponds to the remaining \( K \isomorph \UGroup(1) \) symmetry of electromagnetism.
We denote the elementary charge \( e \) by \( e = \frac{g g'}{\sqrt{g^2 + g'^2}} \) and the Weinberg angle \( \weinbergAngle \) by \( \tan \weinbergAngle = \frac{g'}{g} \).
Then, 
\begin{subequations}
	\label{eq:yangMillsHiggs:GWS:hatAandTauInComponents}
	\begin{align+}
		\hat{A} &= \left(e  A_\gamma + \frac{g \cos \weinbergAngle - g' \sin \weinbergAngle}{2} Z \right) t_+,
		\\
		 \tau &= g W_+ t + g W_- \bar{t} + \frac{g g'}{2 e} Z t_-
	\end{align+}
\end{subequations}
with \( t \defeq \frac{1}{\sqrt{2}} (t_1 + \I t_2) \) and \( \bar{t} \defeq \frac{1}{\sqrt{2}} (t_1 - \I t_2) \).
For later use, let us record the commutation relations of the new basis:
\begin{subequations}
	\label{eq:yangMillsHiggs:GWS:commutationTBarTTPM}
	\begin{align+}
		\commutator{t}{\bar{t}\,} &= \I t_3,
		&
		\commutator{t}{t_\pm} &= -\I t,
		\\
		\commutator{t_+}{t_-} &= 0,
		&
		\commutator{\bar{t}}{t_\pm} &= \I \bar{t}.
	\end{align+}
\end{subequations}

By~\eqref{eq:yangMillsHiggs:stabAfterSymBreak}, we have \( \GauGroup_{A, \varphi} = \hat{\GauGroup}_{\hat{A}, \tau} \).
Since \( K \) is abelian, the stabilizer \( \hat{\GauGroup}_{\hat{A}} \) is isomorphic to \( K \).
Thus it remains to determine which gauge transformations leave \( \tau \) invariant.
For this purpose, let \( k = \left( \smallMatrix{\E^{\frac{\I}{2} \vartheta} & 0 \\ 0 & \E^{-\frac{\I}{2} \vartheta}}, \E^{\I \vartheta} \right) \in K \).
A straightforward calculation shows that \( k \) acts on the basis \( \set{t_a, \I} \) as follows:
\begin{subequations}
	\label{eq:yangMillsHiggs:GWS:adjointActionOfK}
	\begin{align+}
		\AdAction_k t_1 &= \cos \vartheta \, t_1 - \sin \vartheta \, t_2,
		\\
		\AdAction_k t_2 &= \sin \vartheta \, t_1 + \cos \vartheta \, t_2,
		\\
		\AdAction_k t_3 &= t_3,
		\\
		\AdAction_k \I &= \I.
	\end{align+}
\end{subequations}
That is, \( k \) acts as a rotation in the \( (t_1, t_2) \)-plane and acts trivially on \( t_\pm \).
Thus, by~\cref{prop:yangMillsHiggs:stabAPhi},
\begin{equation}
	\GauGroup_{A, \varphi}
		= \hat{\GauGroup}_{\hat{A}} \intersect \hat{\GauGroup}_{\tau}
		\isomorph
		\begin{cases}
			K		& \text{if } \tau \text{ is proportional to } t_- ,
			\\
			\Z_2	& \text{otherwise}.
		\end{cases}
\end{equation}
In other words, \( (\hat{A}, \tau) \) has non-trivial stabilizer if and only if \( W^1 = 0 = W^2 \), \ie, if the \( W \)-boson vanishes leaving the \( Z \)-boson as the only non-trivial intermediate vector boson.
Hence, on the non-generic stratum only one intermediate boson, namely the \( Z \)-boson, survives.

Finally, we show that the decomposition of \( \SectionSpaceAbb{Q} \) defines a stratification.
\begin{prop}
	\label{prop:yangMillsHiggs:GWS:approximationHiggs}
	The decomposition of \( \SectionSpaceAbb{Q} \) into orbit types satisfies the frontier condition.
\end{prop}
\begin{proof}
	As we have seen above, the \( \GauGroup \)-action on \( \SectionSpaceAbb{Q} \) has only the orbit types \( (K) \) and \( (\Z_2) \).
	The frontier condition thus requires that every pair \( (A, \varphi) \in \SectionSpaceAbb{Q} \) with orbit type \( (K) \) can be approximated by a sequence \( (A_i, \varphi_i) \) with orbit type \( (\Z_2) \).
	By \cref{prop:yangMillsHiggs:GWS:stabilizerConnHiggs}, a pair \( (A, \varphi) \) has a stabilizer conjugate to \( K \) only when \( A \) has a stabilizer conjugate to \( \SUGroup(2) \times \UGroup(1) \) or \( \UGroup(1) \times \UGroup(1) \).
	In both cases, the approximation theorem \parencite[Theorem~4.3.5]{KondrackiRogulski1986} shows that there is a converging sequence \( A_i \to A \) of connections \( A_i \) with stabilizer conjugate to \( \Z_2 \times \UGroup(1) \).
	By \cref{prop:yangMillsHiggs:stabPhi}, the pair \( (A_i, \varphi) \) has stabilizer conjugate to \( (\Z_2 \times \UGroup(1)) \intersect K \isomorph \Z_2 \) and converges to \( (A, \varphi) \) by construction.
\end{proof}

\subsection{Orbit types of \texorpdfstring{$ \CotBundle \SectionSpaceAbb{Q} $}{T* Q}}
Next, we determine the secondary stratification of the cotangent bundle.
For this purpose, we endow \( \LieA{g} \) with the \( \AdAction_G \)-invariant scalar product given as the product of (minus) the Killing form on \( \SUAlgebra(2) \) and the usual scalar product on \( \UAlgebra(1) \).
The normalization is such that the generators \( \set{t_a, \I} \) form an orthonormal basis.
In the sequel, we will always use this scalar product to identify \( \LieA{g}^* \) with \( \LieA{g} \).
Recall the definition~\eqref{eq:yangMillsHiggs:GWS:stabilizer} of the electromagnetic gauge group \( K \subseteq G \).
Now, let \( D \in \DiffFormSpace^2 (M, \CoAdBundle P) \). 
We decompose \( D \) according to \( \LieA{g} = \LieA{k} \oplus \LieA{p} \) into \( \restr{D}{\hat{P}} = D_{\LieA{k}} + D_{\LieA{p}} \), with \( D_{\LieA{k}} \in \DiffFormSpace^2(M, \CoAdBundle \hat{P}) \isomorph \DiffFormSpace^2(M, \LieA{k}^*) \) and \( D_{\LieA{p}} \in \DiffFormSpace^2(M, \hat{P} \times_K \LieA{p}^*) \).
Recall from the discussion above that \( \hat{\GauGroup}_{\hat{A}} \) is isomorphic to \( K \), viewed as constant gauge transformations in \( \hat{P} \).
Then, similarly to the above reasoning, using~\eqref{eq:yangMillsHiggs:GWS:adjointActionOfK} and the fact that \( K \) is abelian, we find
\begin{equation}
	 \hat{\GauGroup}_D \intersect \hat{\GauGroup}_{\hat{A}}
		=  \hat{\GauGroup}_{D_{\LieA{p}}} \intersect \hat{\GauGroup}_{\hat{A}}
		=
		\begin{cases}
			K		& \text{if } D_{\LieA{p}} \text{ is proportional to } t_- ,
			\\
			\Z_2	& \text{otherwise}.
		\end{cases}
\end{equation}
We now turn to the stabilizer of \( \Pi \in \SectionSpaceAbb{F}^* \), which we view as \( \Pi \in \sSectionSpace(P \times_G \C^2) \) using the volume form.
As we have seen above, a point \( f \in \C^2 \) has stabilizer \( K \) if and only if it is of the form
\begin{equation}
	f = \frac{r}{\sqrt{2}} \Matrix{\E^{\I \alpha} & 0 \\ 0 & \E^{- \I \alpha}} \cdot \Vector{0 \\ \nu}
		= \frac{r}{\sqrt{2}} \Vector{0 \\ \E^{- \I \alpha} \nu}
\end{equation}
for some \( r \in \R_{>0} \) and \( \E^{\I \alpha} \in \UGroup(1) \), that is, if its first component \( f_1 \) vanishes.
Thus, we have
\begin{equation}
	\hat{\GauGroup}_\Pi \intersect \hat{\GauGroup}_{\hat{A}}
		=
		\begin{cases}
			K		& \text{if } \Pi_1 = 0 \text{ on } \hat{P},
			\\
			\Z_2	& \text{otherwise}.
		\end{cases}
\end{equation}
Hence, in summary, we find
\begin{prop}
	\label{prop:yangMillsHiggs:GWS:orbitTypesCotBundle}
	The stabilizer of \( (A, \varphi, D, \Pi) \in \CotBundle \SectionSpaceAbb{Q} \) is conjugate to \( K \) if the following conditions are met on \( \hat{P} \) (otherwise it is conjugate to \( \Z_2 \)):
	\begin{thmenumerate}
		\item
			\( \tau \) is proportional to \( t_- \), \ie \( W_\pm = 0 \),
		\item
			\( D_{\LieA{p}} \) is proportional to \( t_- \) and
		\item
			the first component of \( \Pi \) vanishes. 
			\qedhere
	\end{thmenumerate}
\end{prop}
We note that \( \CotBundle \SectionSpaceAbb{Q} \) thus has the same orbit types as \( \SectionSpaceAbb{Q} \), \cf \cref{prop:yangMillsHiggs:GWS:stabilizerConnHiggs}.
\begin{remark}
	More generally, instead of~\eqref{eq:yangMillsHiggs:GWS:representation} we could consider the representation
	\begin{equation}
		\rho^Y_{a, \vartheta} (z_1, z_2) = a \cdot \E^{\I Y \vartheta} \cdot (z_1, z_2), \qquad a \in \SUGroup(2), \vartheta \in [0, 4\pi),
	\end{equation}
	which changes the weak hypercharge of the Higgs field from \( \frac{1}{2} \) to \( Y \in \Q \).
	In this case, the stabilizer group \( K \) is replaced by
	\begin{equation}
		K^Y \defeq \set*{\left( \Matrix{\E^{\I Y \vartheta} & 0 \\ 0 & \E^{-\I Y \vartheta}}, \E^{\I \vartheta} \right) \given \vartheta \in [0, 4\pi)}.
	\end{equation}
	Moreover, the generic orbit type is no longer \( \Z_2 \) but the subgroup of \( \Z_2 \times \UGroup(1) \) generated by the elements
	\begin{equation}
		\left( \Matrix{\E^{\I \pi n} & 0 \\ 0 & \E^{-\I \pi n}}, \E^{\I \frac{\pi n}{Y}} \right), \qquad n \in \Z.
	\end{equation}
	We see, in particular, that the orbit type stratification of the configuration space \( \SectionSpaceAbb{Q} \) depends on the weak hypercharge of the Higgs field.
\end{remark}
Note that there might be topological obstructions related to the conditions in \cref{prop:yangMillsHiggs:GWS:orbitTypesCotBundle}.
To accomplish the full classification of gauge orbit types we should derive a theorem analogous to \cref{prop:yangMillsHiggs:classificationOrbitTypesSU} for the case of \( G = \SUGroup(2) \times \UGroup(1) \).
This classification clearly depends on the choice of \( M \).
\begin{remark}
	\label{rem:yangMillsHiggs:GWS:SThree}
	Let us consider the special case \( M = S^3 \).
	By the standard principal fiber bundle classification theorem, all \( G \)-bundles over \( S^3 \) are trivial, because \( \pi_2(G) = 0 \).
	That is, \( P \) is trivial in that case.
	The same applies to any subbundle of \( P \) and hence to any holonomy-induced Howe subbundle.
	As a consequence, the classification problem for that base manifold boils down to the algebraic problem we just solved.
	More complicated cases will be studied elsewhere.
\end{remark}

\subsection{Momentum map and reduced phase space}
Let us determine the expression of the momentum map given by~\eqref{eq:yangMillsHiggs:momentumMap},
\begin{equation}
	\SectionSpaceAbb{J}: \CotBundle \SectionSpaceAbb{Q} \to \GauAlgebra^*,
	\qquad
	(A, D, \varphi, \Pi) \mapsto \dif_A D + \varphi \diamond \Pi \vol_g.
\end{equation}
First, we calculate the second summand.
For this purpose, we consider the momentum map \( \FibreBundleModel{J}: \C^2 \times \C^2 \to \LieA{g}^* \) for the lifted \( G \)-action on \( \CotBundle \C^2 \).
The latter is determined by the equations
\begin{align}
	\label{eq:yangMillsHiggs:GWS:momentumMap:t}
	\dualPair{\FibreBundleModel{J}(z, v)}{t_a} &= \dualPair{v}{\rho_{t_a} z} 
	= \Re(v^* t_a z) \, ,
	\\
	\label{eq:yangMillsHiggs:GWS:momentumMap:i}
	\dualPair{\FibreBundleModel{J}(z, v)}{\I} &= \dualPair{v}{\rho_{\I} z} = - \frac{1}{2} \Im(v^* z) \, ,
\end{align}
and hence, in vector form, it is given by
\begin{equation}
	\FibreBundleModel{J}(z,v) = -\frac{1}{2} \Im \Vector{v_1^* z_2 + v_2^* z_1 \\ 
	\I v_2^* z_1 - \I v_1^*z_2 \\ v_1^*z_1 - v_2^*z_2 \\ v_1^* z_1 + v_2^*z_2}.
\end{equation}
Thus, we have
\begin{equation}\begin{split}
	\restr{(\varphi \diamond \Pi)}{\hat{P}} 
		&= \FibreBundleModel{J}\left(\frac{\eta}{\sqrt{2}} (0, \nu), \Pi\right)
		= - \frac{\eta \nu}{2 \sqrt{2}} \Im \Vector{\Pi_1^* \\ - \I \Pi_1^* \\ - \Pi_2^* \\ \Pi_2^*}
		\\
		&= - \frac{\eta \nu}{2} \left(\frac{\I \Pi_1}{2}  t - \frac{\I \Pi_1^*}{2} \bar{t} + \frac{\Im \Pi_2}{\sqrt{2}}  t_- \right).
\end{split}\end{equation}
If \( \Pi \) lies in the singular orbit type \( (K) \), then its first component \( \Pi_1 \) vanishes and thus the current \( \varphi \diamond \Pi \) is proportional to \( t_- \) in this case.

Next, let us expand \( \restr{D}{\hat{P}} \in \DiffFormSpace^2(M, \hat{P} \times_K \LieA{g}^*) \) in the basis \( \set{t, \bar{t}, t_-, t_+} \):
\begin{equation}
	\label{eq:yangMillsHiggs:GWS:Ddecomposition}
	\restr{D}{\hat{P}} = \underbrace{\frac{D_-}{g} t + \frac{D_+}{g} \bar{t} + \left( \frac{e}{gg'} D_Z - \frac{g \cos \weinbergAngle - g' \sin \weinbergAngle}{2 g g'} D_\gamma \right) t_-}_{D_{\LieA{p}}} + \underbrace{\frac{D_\gamma}{2 e}  t_+}_{D_{\LieA{k}}},
\end{equation}
where the normalization was chosen in such a way that the symplectic form stays in Darboux form in the new coordinates \( (W_\pm, D_\pm, Z, D_Z, A_\gamma, D_\gamma) \) (with respect to the scalar product in which \( \set{t_a, i} \) is an orthonormal basis).
In this notation, using~\eqref{eq:yangMillsHiggs:GWS:hatAandTauInComponents} and~\eqref{eq:yangMillsHiggs:GWS:commutationTBarTTPM}, we find
\begin{equation}\begin{split}
	\restr{(\dif_A D)}{\hat{P}}
		&= \dif \restr{D}{\hat{P}} + \commutator{\hat{A}}{\restr{D}{\hat{P}}} + \commutator{\tau}{\restr{D}{\hat{P}}}
		\\
		&= 
			\frac{1}{g} \dif D_- t + \frac{1}{g} \dif D_+ \bar{t} 
			+ \frac{e}{g g'} \dif D_Z t_- - \frac{g \cos \weinbergAngle - g' \sin \weinbergAngle}{2 g g'} \dif D_\gamma t_- 
			\\ 
			&\quad
			+ \frac{1}{2 e} \dif D_\gamma t_+
			+ \I \left(\sin \weinbergAngle A_\gamma + \cos \weinbergAngle Z\right) \wedge ( D_- t - D_+ \bar{t})
			\\
			&\quad
			+ \I W_+ \wedge \left(D_+ t_3 - (\sin \weinbergAngle D_\gamma + \cos \weinbergAngle D_Z) t \right)
			\\
			&\quad
			- \I W_- \wedge \left(D_- t_3 - (\sin \weinbergAngle D_\gamma + \cos \weinbergAngle D_Z) \bar{t}\right)	.
\end{split}\end{equation}
Hence, the Gauß constraint~\eqref{eq:yangMillsHiggs:GaussConstraint} takes the following form:
\begin{subequations}
	\label{eq:yangMillsHiggs:GWS:GaussConstraint}
	\begin{gather+}
		\begin{aligned}
		\dif D_+ - \I g & (\sin \weinbergAngle A_\gamma + \cos \weinbergAngle Z) \wedge D_+ 
			\\
			&= - \I g W_- \wedge (\sin \weinbergAngle D_\gamma + \cos \weinbergAngle D_Z) 
			- \I \frac{\eta \nu g}{4} \Pi_1^* \vol_g,
		\end{aligned}
		\\
		\begin{aligned}
		\dif D_- + \I g & (\sin \weinbergAngle A_\gamma + \cos \weinbergAngle Z) \wedge D_- 
			\\
			&= \I g W_+ \wedge (\sin \weinbergAngle D_\gamma + \cos \weinbergAngle D_Z) 
			+ \I \frac{\eta \nu g}{4} \Pi_1 \vol_g,
		\end{aligned}
		\\
		\dif D_Z 
			= \I g \cos \weinbergAngle (W_- \wedge D_- - W_+ \wedge D_+)
			+ \frac{\eta \nu g g'}{2 \sqrt{2} e} \Im \Pi_2 \vol_g,
		\\
		\dif D_\gamma
			= \I e (W_- \wedge D_- - W_+ \wedge D_+).
	\end{gather+}
\end{subequations}
Thus, considered on the singular stratum where we have \( W_\pm = 0 = D_\pm \) according to \cref{prop:yangMillsHiggs:GWS:orbitTypesCotBundle}, the Gauß constraint is equivalent to the two equations
\begin{equation}
	\dif D_Z = \frac{\eta \nu g g'}{2 \sqrt{2} e} \Im \Pi_2 \vol_g,
	\qquad
	\dif D_\gamma = 0.
\end{equation}
Since these equations are decoupled, the Gauß constraint cuts out a smooth subbundle of \( \seam{(\CotBundle \SectionSpaceAbb{Q})}^{(K)}_{(K)} \), whose fiber is parametrized by the fields \( D_\gamma \in \clDiffFormSpace^2(M) \), \( D_Z \in \DiffFormSpace^2(M) \) and \( \frac{\nu}{\sqrt{2}} \Re \Pi_2 \in \sFunctionSpace(M) \).

\begin{figure}
	\centering
	\begin{tikzcd}
		\seam{\check{\SectionSpaceAbb{P}}}^{(\Z_2)}_{(\Z_2)}
			\arrow[rr, dashed]
			\arrow[rrd, dashed]
			&
			& \seam{\check{\SectionSpaceAbb{P}}}^{(\Z_2)}_{(K)}
				\arrow[d, dashed]
				\arrow[r, phantom, "\Bigg\rbrace", very near start]
			&[-1.9em] \check{\SectionSpaceAbb{P}}_{(\Z_2)} 
			&[-2.3em] \scriptstyle D_\pm \, \neq \, 0
		\\
			&
			& \seam{\check{\SectionSpaceAbb{P}}}^{(K)}_{(K)}
				\arrow[r, phantom, "\Bigg\rbrace", very near start]
			& \check{\SectionSpaceAbb{P}}_{(K)} 
			& \scriptstyle  D_\pm \, = \, 0
		\\[-3.5ex]
			\arrow[d, two heads]
			&
			& \arrow[d, two heads]
			&
			&
		\\[-1.5ex]
		\check{\SectionSpaceAbb{Q}}_{(\Z_2)}
			\arrow[rr, dashed]
			&
			& \check{\SectionSpaceAbb{Q}}_{(K)}
			&
			&
		\\[-3.5ex]
		\scriptstyle W_\pm \, \neq \, 0
			&
			& \scriptstyle W_\pm \, = \, 0
			&
			&
	\end{tikzcd}
	\caption{Schematic illustration of the secondary stratification of the reduced phase \( \check{\SectionSpaceAbb{P}} \) and its relation to the orbit type stratification of the reduced configuration space \( \check{\SectionSpaceAbb{Q}} \). Dotted arrows mean that the target lies in the closure of the source.}
	\label{fig:yangMillsHiggs:GWS:secondaryStratification}
\end{figure}
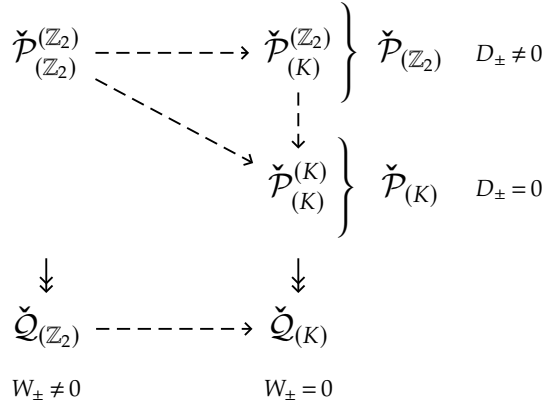

According to \cref{prop:cotangentBundle:singularCotangentBundleRed}, the reduced phase space decomposes into
\begin{equation}
	\check{\SectionSpaceAbb{P}} 
		= \underbrace{\seam{\check{\SectionSpaceAbb{P}}}^{(K)}_{(K)}}_{\check{\SectionSpaceAbb{P}}_{(K)}}
		\union 
		\underbrace{
		\seam{\check{\SectionSpaceAbb{P}}}^{(\Z_2)}_{(K)}
		\union
		\seam{\check{\SectionSpaceAbb{P}}}^{(\Z_2)}_{(\Z_2)}
		}_{\check{\SectionSpaceAbb{P}}_{(\Z_2)}}
\end{equation}
and the strata \( \check{\SectionSpaceAbb{P}}_{(K)} \) and \( \check{\SectionSpaceAbb{P}}_{(\Z_2)} \) are symplectic.
Moreover, \( \check{\SectionSpaceAbb{P}}^{(K)}_{(K)} \) is symplectomorphic to the cotangent bundle of \( \check{\SectionSpaceAbb{Q}}_{(K)} \).
As we have seen, this singular stratum is the (reduced) phase space of the theory of a photon and the \( Z \)-boson without any other intermediate bosons.
In contrast, on the generic stratum \( \seam{\check{\SectionSpaceAbb{P}}}^{(\Z_2)}_{(\Z_2)} \isomorph \CotBundle \check{\SectionSpaceAbb{Q}}_{(\Z_2)} \) all intermediate vector bosons are present.
This cotangent bundle is stitched together by the seam \( \seam{\check{\SectionSpaceAbb{P}}}^{(\Z_2)}_{(K)} \), where no \( W \)-bosons are present but their conjugate momenta are non-zero.
The secondary stratification is schematically illustrated in \cref{fig:yangMillsHiggs:GWS:secondaryStratification}.
In the next section, we study the structure of the reduced phase space in detail and show that it is similar to that of the harmonic oscillator discussed in \cref{ex:cotangentBundle:harmonicOss}.

\subsection{Description of the reduced phase space}

In this section, we give an explicit parameterization of the reduced phase space.
The main idea is to identify a part of the configuration space on which the group of gauge transformations acts transitively and thereby to reduce the symmetry to a subgroup.
Next, we repeat this process until only a finite-dimensional symmetry remains.
To simplify the discussion, we limit our attention to the case \( M = S^3 \).
As noted in \cref{rem:yangMillsHiggs:GWS:SThree}, for \( M = S^3 \) all bundles occurring in the construction are trivial and we can hence represent all geometric objects on the bundle as objects living on \( S^3 \).
The phenomena resulting from a non-trivial topology of these bundles for a general base manifold will be discussed in a separate paper.

\paragraph*{1. Reduction of the gauge group from \( \sFunctionSpace(M, G) \) to \( \sFunctionSpace(M, K) \):}
Recall from \cref{prop:yangMillsHiggs:configSpaceSymmBreaking} that the Higgs mechanism yields a parametrization of \( (A, \varphi) \in \SectionSpaceAbb{Q} \) in terms of the variables \( (\phi, \hat{A}, \tau, \eta) \), which are viewed as elements of a bundle over \( \sSectionSpace(P \times_G G \slash K) \).
In the present case of a trivial bundle \( P \), we can strengthen this result.
For this purpose, recall that \( G \slash K \) is diffeomorphic to \( S^3 \) via the \( G \)-action on \( \C^2 \) defined in~\eqref{eq:yangMillsHiggs:GWS:representation}.
Hence, in particular, the \( K \)-bundle \( G \to G \slash K \) is trivial.
Therefore, every smooth map \( M \to G \slash K \) lifts to a smooth map \( M \to G \) and so the action of \( \GauGroup = \sFunctionSpace(M, G) \) on \( \sSectionSpace(P \times_G G \slash K) \isomorph \sFunctionSpace(M, G \slash K) \) is transitive.
Moreover, the stabilizer of the constant function taking values in the identity coset is the subgroup \( \hat{\GauGroup} = \sFunctionSpace(M, K) \) (this is in accordance with \cref{prop:yangMillsHiggs:stabPhi}).
Therefore, \( \sFunctionSpace(M, G \slash K) \) is diffeomorphic to the homogeneous space \( \GauGroup \slash \hat{\GauGroup} \).
In other words, the decomposition~\eqref{eq:yangMillsHiggs:GWS:decompositionHiggsField} takes the form
\begin{equation}
	\label{eq:yangMillsHiggs:GWS:decompositionHiggsField:trivial}
	\varphi = \frac{\eta}{\sqrt{2}} \, \lambda \cdot \Vector{0 \\ \nu},
\end{equation}
where \( \eta \in \sFunctionSpace(M, \R_{> 0}) \) and \( \lambda \in \GauGroup \) is unique up to an element of \( \hat{\GauGroup} \).
By implementing the unitary gauge, \ie, by gauging away \( \lambda \), we obtain the following refinement of \cref{prop:yangMillsHiggs:configSpaceSymmBreaking}.
Recall the decomposition \( \LieA{g} = \LieA{p} \oplus \LieA{k} \) with \( \LieA{p} \) spanned by \( \set{t, \bar{t}, t_-} \) and \( \LieA{k} \) spanned by \( t_+ \).
\begin{prop}
	\label{prop:yangMillsHiggs:GWS:unitaryGauge}
	The assignment\footnote{With a slight abuse of notation, we continue to use the notation \( \hat{A} \) and \( \tau \). Note, however, that these fields differ from the ones introduced in~\eqref{eq:yangMillsHiggs:GWS:connectionInComponents} by a gauge transformation.}
	\begin{equation}
		(A, \varphi)
		\mapsto
		(\equivClass{\lambda, \hat{A}, \tau}, \eta),
	\end{equation}
	where \( \lambda \) and \( \eta \) are determined by~\eqref{eq:yangMillsHiggs:GWS:decompositionHiggsField:trivial} and \( (\lambda^{-1} \cdot A) = \hat{A} + \tau \), defines a \( \GauGroup \)-equivariant diffeomorphism
	\begin{equation}
		\label{eq:yangMillsHiggs:GWS:unitaryGauge:configSpace}
		\SectionSpaceAbb{Q} \to \GauGroup \times_{\hat{\GauGroup}} \bigl(\DiffFormSpace^1(M, \LieA{k}) \times \DiffFormSpace^1(M, \LieA{p})\bigr) \times \sFunctionSpace(M, \R_{> 0}).
	\end{equation}
	Here, on the right hand side, \( \GauGroup \) acts by left translation on the first factor and \( \hat{\GauGroup} \) acts by gauge transformations on the space of connections \( \DiffFormSpace^1(M, \LieA{k}) \) and via the adjoint action on \( \DiffFormSpace^1(M, \LieA{p}) \).
	In particular, the gauge orbit space \( \check{\SectionSpaceAbb{Q}} \) is isomorphic to \( \bigl(\DiffFormSpace^1(M, \LieA{k}) \times \DiffFormSpace^1(M, \LieA{p})\bigr) \slash \hat{\GauGroup} \times \sFunctionSpace(M, \R_{> 0}) \) in the sense of stratified spaces.
\end{prop}
Recall from the discussion in \cref{sec:cotangentBundle:normalForm}, that there is a natural description of the cotangent bundle of an associated bundle such that the momentum map is brought into a convenient normal form, \cf \cref{prop:cotangentBundle:simpleNormalForm} (in \cref{sec:cotangentBundle:normalForm} the focus lies on certain associated bundles where the slice is the typical fiber --- however, the discussion there does not really use the properties of a slice).
Using the same strategy, we identify
\begin{equation}
	\GauGroup \times_{\hat{\GauGroup}} \Bigl(\sFunctionSpace(M, \LieA{p}) \times \DiffFormSpace^1(M, \LieA{k})^2 \times \DiffFormSpace^1(M, \LieA{p})^2\Bigr) \times \sFunctionSpace(M, \R_{> 0}) \times \sFunctionSpace(M, \R)
\end{equation}
with \( \TBundle \SectionSpaceAbb{Q} \) via the map
\begin{equation}
	\bigl(\equivClass{\lambda, (\varsigma, \hat{A}, \diF \hat{A}, \tau, \diF \tau)}, \eta, \diF \eta\bigr)
	\mapsto
	\Vector{
		A = \lambda \cdot (\hat{A} + \tau)
		\\
		\diF A = \AdAction_\lambda (\diF \hat{A} + \diF \tau) - \dif_A \varsigma
		\\
		\varphi = \frac{\eta}{\sqrt{2}} \, \lambda \cdot \Vector{0 \\ \nu}
		\\
		\diF \varphi =  \frac{\diF \eta}{\eta} \varphi + \varsigma \ldot \varphi 
	}.
\end{equation}
Here, we have denoted \( \DiffFormSpace^1(M, \cdot) \times \DiffFormSpace^1(M, \cdot)\equiv \DiffFormSpace^1(M, \cdot)^2 \).
A straightforward calculation shows that the dual map
\begin{equation}\begin{split}
	\CotBundle \SectionSpaceAbb{Q} 
	&\to
	\GauGroup \times_{\hat{\GauGroup}} \Bigl(\DiffFormSpace^3(M, \LieA{p}^*) \times \DiffFormSpace^1(M, \LieA{k}) \times \DiffFormSpace^2(M, \LieA{k}^*) \times \DiffFormSpace^1(M, \LieA{p}) \times \DiffFormSpace^2(M, \LieA{p}^*)\Bigr)
	\\
	&\qquad\qquad \times \sFunctionSpace(M, \R_{> 0}) \times \DiffFormSpace^3(M, \R),
	\\
	(A, D, &\varphi, \Pi) 
	\mapsto
	(\equivClass{\lambda, \nu, \hat{A}, \hat{D}, \tau, D_\tau}, \eta, \Pi_\eta)
\end{split}\end{equation}
is given by
\begin{align}
	\hat{D} + D_\tau &= \CoAdAction_{\lambda^{-1}} D,
	\\
	\Pi_\eta &= \frac{1}{\eta} \Re (\Pi^* \varphi),
	\\
	\nu &= \pr_{\LieA{p}^*} (\dif_A D + \varphi \diamond \Pi).
\end{align}
Moreover, let us parameterize \( (\hat{A}, \tau) \) in terms of the fields \( (W_\pm, Z, A_\gamma) \) as in~\eqref{eq:yangMillsHiggs:GWS:hatAandTauInComponents} and \( (\hat{D}, D_\tau) \) in terms of the fields \( (D_\pm, D_Z, D_\gamma) \) as in~\eqref{eq:yangMillsHiggs:GWS:Ddecomposition}.
Then,~\eqref{eq:yangMillsHiggs:GWS:GaussConstraint} shows that the \( \LieA{k}^* \)-projection of the Gauß constraint is given by
\begin{equation}
	\label{eq:yangMillsHiggs:GWS:GaussRed}
	0 = \pr_{\LieA{k}^*} (\dif_A D + \varphi \diamond \Pi) = \frac{1}{2e} \dif D_\gamma +  \Im \bigl(D_- \wedge W_-\bigr).
\end{equation}
Denote \( \SectionSpaceAbb{Q}_{\textrm{red}} = \DiffFormSpace^1(M, \LieA{k}) \times \DiffFormSpace^1(M, \C) \).
Its elements are \( (A_\gamma, W_-) \).
Moreover, elements of \( \CotBundle_{(A_\gamma, W_-)} \SectionSpaceAbb{Q}_{\textrm{red}} = \DiffFormSpace^2(M, \LieA{k}^*) \times \DiffFormSpace^2(M, \C) \) are denoted by \( (D_\gamma, D_-) \). 
The right hand side of~\eqref{eq:yangMillsHiggs:GWS:GaussRed} is the momentum map for the induced \( \hat{\GauGroup} \)-action\footnote{Note that \( \lambda \in \hat{\GauGroup} \) acts on \( A_\gamma \in \DiffFormSpace^1(M, \LieA{k}) \) by \( \lambda \cdot A_\gamma = A_\gamma - \frac{1}{e} \lambda^{-1} \dif \lambda \).} on \( \CotBundle \SectionSpaceAbb{Q}_{\textrm{red}} \), because the momentum map for the lift of the \( K \)-action~\eqref{eq:yangMillsHiggs:GWS:adjointActionOfK} to \( \CotBundle \C \) is given by \( \FibreBundleModel{J}_K(z, w) = - \Im(w^* z) \).
In summary, the upshot of this first symmetry reduction is the following description of the reduced phase space.
\begin{prop}
	\label{prop:yangMillsHiggs:GWS:QPhaseSpace}
	The diffeomorphism~\eqref{eq:yangMillsHiggs:GWS:unitaryGauge:configSpace} induces an isomorphism 
	\begin{equation}
		\check{\SectionSpaceAbb{P}}
		=
		\CotBundle \SectionSpaceAbb{Q} \sslash \GauGroup 
		=
		\bigl(\CotBundle \SectionSpaceAbb{Q}_{\textrm{red}} \sslash \hat{\GauGroup} \bigr) \times \CotBundle \bigl(\DiffFormSpace^1(M, \R \, t_-)  \times \sFunctionSpace(M, \R_{> 0}) \bigr)
	\end{equation}
	of symplectic stratified spaces.
\end{prop}
\begin{proof}
	First note that the Gauß constraint~\eqref{eq:yangMillsHiggs:GaussConstraint} is equivalent to \( \nu = 0 \) and the condition~\eqref{eq:yangMillsHiggs:GWS:GaussRed}, which is the momentum map constraint for the \( \hat{\GauGroup} \)-action as we have discussed above.
	Now the assertion follows from the decomposition \( \LieA{p} = \C \oplus \R \) with \( \C \) spanned by \( \set{t, \bar{t}} \) and \( \R \) spanned by \( t_- \) and from the fact that \( K \) acts trivially on \( t_- \), \cf~\eqref{eq:yangMillsHiggs:GWS:adjointActionOfK}.
\end{proof}
Hence, the singular structure of the phase space is completely encoded in the symplectic reduction of \( \CotBundle \SectionSpaceAbb{Q}_{\textrm{red}} \) by the \( \hat{\GauGroup} \)-action.

\paragraph*{2. Reduction from \( \hat{\GauGroup} \) to \( K \):}
Since \( S^3 \) is simply connected, the Hodge decomposition theorem yields
\begin{equation}
	\DiffFormSpace^1(M, \LieA{k}) = \dif \sFunctionSpace(M, \LieA{k}) \oplus \dif^* \DiffFormSpace^2(M, \LieA{k}). 
\end{equation}
Note that every map \( f: M \to \LieA{k} \) induces a map \( \hat{f} = \exp_K \circ f: M \to K \) such that \( \hat{f}^{-1} \dif \hat{f} = \dif f \).
Accordingly, every \( K \)-connection \( A_\gamma \) can be written as
\begin{equation}
	A_\gamma = \psi^{-1} \dif \psi + \beta,
\end{equation}
where \( \psi \in \sFunctionSpace(M, K) \) and \( \beta \in \dif^* \DiffFormSpace^2(M, \LieA{k}) \) is uniquely determined by the curvature \( F_{A_\gamma} \) of \( A_\gamma \). 
The action of \( \hat{\GauGroup} \) on \( \dif \DiffFormSpace^0(M, \LieA{k}) \), viewed as a subspace of the space of \( K \)-connections, is transitive with kernel consisting of the constant functions.
We identify this kernel with \( K \).
Note that \( K \) acts trivially on \( \dif^* \DiffFormSpace^2(M, \LieA{k}) \) and by rotation on \( \DiffFormSpace^1(M, \C) \), \cf~\eqref{eq:yangMillsHiggs:GWS:adjointActionOfK}.

\begin{lemma}
	The map
	\begin{equation}\label{eq:yangMillsHiggs:GWS:diffeoQRed}\begin{split}
		\SectionSpaceAbb{Q}_{\textrm{red}}
		\to
		\bigl( \hat{\GauGroup} \times_K \DiffFormSpace^1(M, \C) \bigr) \times \dif \DiffFormSpace^1 (M, \LieA{k})
		\\
		(A_\gamma, W_-)
		\mapsto
		(\equivClass{\psi, v}, F_{A_\gamma}),
	\end{split}\end{equation}
	with \( v = \psi^{-1} W_- \) is a diffeomorphism.
\end{lemma}
\begin{proof}
	The inverse map \( (\equivClass{\psi, v}, \beta) \mapsto (A_\gamma, W_-) \) is given by
	\begin{equation}
		\label{eq:yangMillsHiggs:GWS:diffeoQRedInverse}
		A_\gamma = \psi^{-1} \dif \psi + \dif^* \laplace^{-1} \beta,
		\qquad
		W_- = \psi \ldot v,
	\end{equation}
	with \( \psi \in \hat{\GauGroup} \), \( \beta \in \dif \DiffFormSpace^1(M, \LieA{k}) \) and \( v \in \DiffFormSpace^1(M, \C) \), because \( \dif \dif^* \laplace^{-1} \) is the identity operator on \( \dif \DiffFormSpace^1(M, \LieA{k}) \). 
\end{proof}
To get a convenient description of the cotangent bundle \( \CotBundle \SectionSpaceAbb{Q}_{\textrm{red}} \), we follow the same strategy as above.
Let \( \sFunctionSpace(M, \LieA{k})_0 \) denote the functions whose average over \( M \) vanishes.
The surjective linear map
\begin{equation}
	\sFunctionSpace(M, \LieA{k})
	\to
	\sFunctionSpace(M, \LieA{k})_0,
	\qquad
	\psi \mapsto \psi - \frac{1}{\vol_M} \int_M \psi \vol_g
\end{equation}
has kernel \( \LieA{k} \) and thus yields the identification \( \sFunctionSpace(M, \LieA{k}) \slash \LieA{k} \isomorph \sFunctionSpace(M, \LieA{k})_0 \).
Dually, the integral map \( \int_M: \DiffFormSpace^3(M, \LieA{k}^*) \to \LieA{k}^* \) has as kernel the space \( \DiffFormSpace^3(M, \LieA{k}^*)_0 \), which is the dual space to \( \sFunctionSpace(M, \LieA{k})_0 \).
Linearizing the reconstruction equations~\eqref{eq:yangMillsHiggs:GWS:diffeoQRedInverse} yields
\begin{equation}\begin{split}
	\diF A_\gamma &= \dif \diF \psi + \dif^* \laplace^{-1} \diF \beta,
	\\
	\diF W_- &= \diF \psi \ldot (\psi \cdot v) + \psi \ldot \diF v,
\end{split}\end{equation}
where \( \diF \psi \in \sFunctionSpace(M, \LieA{k})_0 \), \( \diF \beta \in \dif \DiffFormSpace^1 (M, \LieA{k}) \) and \( \diF v \in \DiffFormSpace^1(M, \C) \).
By dualizing these equations, we get a diffeomorphism
\begin{equation}\begin{split}
	\CotBundle \SectionSpaceAbb{Q}_{\textrm{red}}
	\to
	\hat{\GauGroup} \times_K \bigl(\DiffFormSpace^3(M, \LieA{k}^*)_0 \times \DiffFormSpace^1(M, \C)^2\bigr) \times \dif \DiffFormSpace^1 (M, \LieA{k}) \times \dif^* \DiffFormSpace^2 (M, \LieA{k}^*)
	\\
	(A_\gamma, D_\gamma, W_-, D_-)
	\mapsto
	\bigl(\equivClass{\psi, \Pi_0, v, D_v}, F_{A_\gamma}, D_{F_{A_\gamma}} \bigr),
\end{split}\end{equation}
with 
\begin{equation}\begin{split}
	\Pi_0 &= \pr_{\DiffFormSpace^3(M, \LieA{k}^*)_0} \left(\frac{1}{2 e} \dif D_\gamma + \Im(D_- \wedge W_-) \right),
	\\
	D_v &= \psi^{-1} \ldot D_-,
	\\
	D_{F_{A_\gamma}} &= \dif^* \laplace^{-1} D_\gamma.
\end{split}\end{equation}
Moreover, the reduced Gauß constraint~\eqref{eq:yangMillsHiggs:GWS:GaussRed} is equivalent to \( \Pi_0 = 0 \) and
\begin{equation}
	0 = \int_M \left(\frac{1}{2e} \dif D_\gamma + \Im \bigl(D_- \wedge W_-\bigr)\right)
		= \int_M \Im \bigl(D_- \wedge W_-\bigr).
\end{equation}
The right hand side is the momentum map for the lifted \( K \)-action on \( \CotBundle \bigl(\DiffFormSpace^1(M, \C)\bigr) \).
Thus, the second symmetry reduction yields the following.
\begin{prop}
	\label{prop:yangMillsHiggs:GWS:QredPhaseSpace}
	The diffeomorphism~\eqref{eq:yangMillsHiggs:GWS:diffeoQRed} induces an isomorphism 
	\begin{equation}
		\CotBundle \SectionSpaceAbb{Q}_{\textrm{red}} \sslash \hat{\GauGroup} 
		=
		\Bigl(\CotBundle \bigl(\DiffFormSpace^1(M, \C)\bigr) \sslash K \Bigr) \times \dif \DiffFormSpace^1 (M, \LieA{k}) \times \dif^* \DiffFormSpace^2 (M, \LieA{k}^*)
	\end{equation}
	of symplectic stratified spaces.
\end{prop}
Hence, in combination with the first reduction, see \cref{prop:yangMillsHiggs:GWS:QPhaseSpace}, we find that the singular structure of the reduced phase space \( \check{\SectionSpaceAbb{P}} \) is completely determined by the singular cotangent bundle reduction of \( \CotBundle \bigl(\DiffFormSpace^1(M, \C)\bigr) \) with respect to the action of the \emph{finite-dimensional} (compact) Lie group \( K \).
Note that \( \CotBundle \bigl(\DiffFormSpace^1(M, \C)\bigr) \) is pointwise the phase space of three (coupled) harmonic oscillators and the \( K \) action corresponds to the diagonal \( \UGroup(1) \)-symmetry.
This shows that the singularity structure of the reduced phase space is essentially determined by a finite-dimensional Lie group action.
The reduced phase space \( \CotBundle \bigl(\DiffFormSpace^1(M, \C)\bigr) \sslash K \) may thus be studied using classical geometric invariant theory for the finite-dimensional reference system.
This will be done elsewhere.

\subsection{Dynamics on the singular stratum}

It is a challenge for further research to study the dynamics of this model on its full stratified phase space.
As a first step, we analyze the dynamics on the singular stratum.
First, a word of warning concerning conventions is in order.
Since we followed the usual physics convention and introduced the coupling constants when writing the connection \( A \) in terms of \( W \) and \( B \) fields, see~\eqref{eq:yangMillsHiggs:GWS:connectionInComponents}, we need to use a different scalar product on \( \LieA{g} \) in the Hamiltonian to counterbalance the coupling constants\footnote{The reader might find it instructive to compare the situation at hand to that of classical mechanics. There, the kinetic part of the Lagrangian is usually written in the form \( L = \frac{m}{2} v^2 \). But, the mass can be absorbed in the metric \( g \) on the configuration space by setting \( \norm{v}^2_g \defeq m v^2 \). Then, the Hamiltonian is given by \( H = \frac{1}{2} \norm{p}^2_{g^{-1}} \), where the norm is taken with respect to the inverse (or dual) metric \( g^{-1} \). In our field theoretic setting, the coupling constants play the role of inverse masses.}.
Let \( \kappa \) be the scalar product on \( \LieA{g} \) in which \( \set{t_a, i} \) is orthogonal with norm \( \kappa(t_a, t_a) = \frac{1}{g^2} \) and \( \kappa(\I, \I) = \frac{1}{g'^2} \) (and hence \( \kappa(t_+, t_\pm) = \frac{g'^2 \pm g^2}{g^2 g'^2} = \kappa(t_-, t_\mp) \)).
Moreover, let \( \kappa^{-1} \) be the inverse of \( \kappa \), \ie, we have \( \kappa^{-1}(t_a, t_a) = g^2 \) and \( \kappa^{-1}(\I, \I) = g'^2 \).
In terms of these scalar products, the Hamiltonian \( \SectionMapAbb{H} \) defined by~\eqref{eq:yangMillsHiggs:hamiltonian} reads as follows:
\begin{equation}
	\label{eq:yangMillsHiggs:GWS:hamiltonian}
	\SectionSpaceAbb{H}(A, D, \varphi, \Pi)
		= \int_M \frac{\ell}{2} \left(\norm{D}^2_{\kappa^{-1}} + \norm{F_A}^2_{\kappa}
		+ \norm{\Pi}^2_{\C} + \norm{\dif_A \varphi}^2_{\C} + 2 \, V(\varphi) \right) \vol_g.
\end{equation}

\begin{prop}
	On the singular stratum \( \seam{(\CotBundle \SectionSpaceAbb{Q})}^{(K)}_{(K)} \), the Hamiltonian~\eqref{eq:yangMillsHiggs:GWS:hamiltonian} has the form
	\begin{equation}\label{eq:yangMillsHiggs:GWS:hamiltonianSingular}\begin{split}
		\SectionMapAbb{H}(A_\gamma, Z, &\eta, D_\gamma, D_Z, \Pi_2)
			\\
			&= \int_M \frac{\ell}{2} \Big( 
				\norm{D_\gamma}^2 + \norm{D_Z}^2
				+ \norm{\dif A_\gamma}^2 + \norm{\dif Z}^2
				+ \norm{\Pi_2}_{\C}^2
				\\
				&\qquad\quad
				+ \frac{\nu^2}{2} \norm{\dif \eta}^2 + \frac{\eta^2 \nu^2 (g^2 + g'^2)}{8} \norm{Z}^2
				+ \lambda \nu^2 (\eta^2 - 1)^2 \Big) \vol_g.
	\end{split}\end{equation} 
\end{prop}
By examining the Hamiltonian~\eqref{eq:yangMillsHiggs:GWS:hamiltonianSingular}, we can read off the particle content over the singular stratum.
We obtain a non-interacting system consisting of electrodynamics described by the photon \( A_\gamma \), the theory of a massive vector boson described by the \( Z \)-boson with mass \( m_Z^2 = \frac{\eta^2 \nu^2 (g^2 + g'^2)}{4} \) and the theory of a self-interacting real scalar field described by the Higgs boson \( \eta \) with mass \( m_\eta^2 = - 4 \lambda \nu^2 \).
\begin{proof}
	Over the singular stratum, we have \( D_\pm = 0 \) and thus
	\begin{equation}\begin{split}
		\norm{D}^2_{\kappa^{-1}} 
			&= \norm*{D_Z - \frac{g^2-g'^2}{2gg'} D_\gamma}^2 + \frac{(g^2 + g'^2)^2}{4 g^2 g'^2} \norm{D_\gamma}^2
			\\
			&\qquad+ \frac{g^2 - g'^2}{g g'} \scalarProd*{D_Z - \frac{g^2-g'^2}{2gg'} D_\gamma}{D_\gamma}
			\\
			&= \norm{D_Z}^2 + \norm{D_\gamma}^2.
	\end{split}\end{equation}
	The curvature of \( A \) reads in terms of the fields after symmetry breaking as follows:
	\begin{equation}\begin{split}
		\restr{(F_A)}{\hat{P}}
			&= F_{\hat{A}} + \dif_{\hat{A}} \tau + \frac{1}{2} \wedgeLie{\tau}{\tau}
			\\
			&= F_{\hat{A}} 
				+ \dif \tau 
				+ \I g^2 \left( \sin \weinbergAngle  A_\gamma + \cos \weinbergAngle Z \right) \wedge (W_+ t - W_- \bar{t})
				+ \I g^2 \, W_+ \wedge W_- \, t_3.
	\end{split}\end{equation}
	Hence, on the singular stratum we simply have
	\begin{equation}
		\restr{(F_A)}{\hat{P}} = \left(e \dif A_\gamma + \frac{g \cos \weinbergAngle - g' \sin \weinbergAngle}{2} \dif Z \right) t_+ + \frac{g g'}{2 e} \dif Z \, t_-.
	\end{equation}
	The norm of \( F_A \) with respect to \( \kappa \), is thus, on the singular stratum, given by
	\begin{equation}
		\norm{F_A}^2_\kappa = \norm{\dif A_\gamma}^2 + \norm{\dif Z}^2.
	\end{equation}
	According to~\eqref{eq:yangMillsHiggs:GWS:representationAlgebra}, we get
	\begin{equation}
		\dif_A \varphi = \left( \dif + g W^a t_a + \frac{\I g'}{2} B \one \right) \varphi.
	\end{equation}
	Using the representation~\eqref{eq:yangMillsHiggs:GWS:decompositionHiggsField} and the definition~\eqref{eq:yangMillsHiggs:GWS:WpmZAgamma} of \( W_\pm \) and \( Z \), we find in terms of the fields after symmetry breaking
	\begin{equation}\begin{split}
		\restr{(\dif_A \varphi)}{\hat{P}} 
			&= \frac{\dif \eta}{\sqrt{2}} \Vector{0 \\ \nu} + \frac{g \eta}{\sqrt{2}} W^a t_a \Vector{0 \\ \nu} + \frac{\I g' \eta}{2 \sqrt{2}} B \Vector{0 \\ \nu}
			\\
			&= \frac{\I g \eta \nu}{2} W_+ \Vector{1 \\ 0} + \left(\frac{\nu}{\sqrt{2}} \dif \eta - \frac{\I \eta \nu \sqrt{g^2 + g'^2}}{2 \sqrt{2}} Z\right)  \Vector{0 \\ 1}.
	\end{split}\end{equation}
	Thus, on the singular stratum,
	\begin{equation}
		\norm{\dif_A \varphi}^2_{\C} = \frac{\nu^2}{2} \norm{\dif \eta}^2 + \frac{\eta^2 \nu^2 (g^2 + g'^2)}{8} \norm{Z}^2.
	\end{equation}
	Plugging these identities into~\eqref{eq:yangMillsHiggs:GWS:hamiltonian} yields~\eqref{eq:yangMillsHiggs:GWS:hamiltonianSingular}.
\end{proof}

\section{Outlook}
There is a number of fundamental issues, which may be studied in the future:
\begin{enumerate}
	\item
		Extend the discussion of the singular reduction for the Glashow--Weinberg--Salam model to Cauchy surfaces \( M \) of arbitrary topological type. 
	\item
		Analyze the structure of the singular reduced phase space for the Glashow--Weinberg--Salam model in terms of geometric invariant theory.
	\item
		Study the dynamics of the full Hamiltonian system on the reduced phase space and clarify, in particular, the meaning of the seams.	
	\item
		Pass to the quantum theory by developing a theory of geometric quantization for infinite-dimensional stratified symplectic spaces. 
\end{enumerate}
As mentioned in the beginning, there is a general symplectic reduction theory in infinite dimensions going beyond the cotangent bundle case.
This will be presented elsewhere \parencite{DiezThesis}.
 
%%%%%%%%%%%%%%%%%%%%%%%%%%%%%%%%%%%%%%%%%%%%%%%%%%%%%%%%%%%%%%%%%%%%%%%%%%%%%%%%%%%%%%%%

\appendix
\section{Appendix: Calculus on infinite-dimensional manifolds}
\label{Frechet}

%%%%%%%%%%%%%%%%%%%%%%%%%%%%%%%%%%%%%%%%%%%%%%%%%%%%%%%%%%%%%%%%%%%%%%%%%%%%%%%%%%%%%%%%

Our references for terminology and notation in the framework of infinite-dimensional differential geometry are~\parencite{Hamilton1982} for the Fréchet case and~\parencite{Neeb2006} for the locally convex setting.

By a manifold we understand a possibly infinite-dimensional smooth manifold \( M \) without boundary modeled on locally convex vector spaces.
More precisely, different connected components of \( M \) are allowed to be modeled on non-isomorphic vector spaces, see \eg \parencites[Section~II.1]{Lang1999}[Definition~3.1.1]{MarsdenRatiuEtAl2002}.
A subset \( S \) of \( M \) is called a submanifold if at each point \( s \in S \) there exists a chart \( \kappa: M \supseteq U \to X \) and a closed subspace \( Y \subseteq X \) such that \( \kappa(U \intersect S) = \kappa(U) \intersect Y \).
The submanifold \( S \) is said to be split if, in addition, each subspace \( Y \) is topologically complemented in \( X \).
A Lie group is a group endowed with a smooth manifold structure such that multiplication and inversion are smooth maps.
As a vector space, the Lie algebra \( \LieA{g} \) of a Lie group \( G \) coincides with the tangent space \( \TBundle_e G \) at the identity element, and the Lie bracket is obtained by identifying \( \LieA{g} \) with the space of left invariant vector fields.
This turns \( \LieA{g} \) into a locally convex Lie algebra.

%%%%%%%%%%%%%%%%%%%%%%%%%%%%%%%%%%%%%%%%%%%%%%%%%%%%%%%%%%%%%%%%%%%%%%%%%%%%%%%%%%%%%%%%%%%%%%%%%%%%

\subsection{Cotangent bundles in infinite dimensions}
\label{sec:cotangentBundle:definition}

%%%%%%%%%%%%%%%%%%%%%%%%%%%%%%%%%%%%%%%%%%%%%%%%%%%%%%%%%%%%%%%%%%%%%%%%%%%%%%%%%%%%%%%%%%%%%%%%%%

The tangent bundle \( \TBundle Q \) of a manifold \( Q \) is itself a smooth manifold again in such a way that the the projection \( \TBundle Q \to Q \) is a smooth locally trivial bundle.
The dual bundle \( \TBundle' Q \defeq \bigDisjUnion_{q \in Q} (\TBundle_q Q)' \) is more problematic, \cf \parencites[Remark~I.3.9]{Neeb2006}{Urbanski1974}.
In order to endow \( \TBundle' Q \) with a smooth fiber bundle structure we need a locally convex topology on the dual \( X' \) of the model space \( X \) of \( Q \) such that for every local diffeomorphism \( \phi: X \to X \) the map
\begin{equation}
	\tau_\phi: X \times X' \to X', \qquad (x, \alpha) \mapsto \alpha \circ \tangent_x \phi
\end{equation}
is smooth since otherwise the notion of smoothness of \( \TBundle' Q \) is chart-dependent.
It is straightforward to construct a map \( \phi \) such that \( \tau_\phi \) involves the natural pairing \( X \times X' \to \R \).
However, this pairing is discontinuous in \( 0 \) for any vector topology on \( X' \) except when \( X \) is a Banach space, see \parencite[Satz~1]{Maissen1963}.
Thus, in summary, we cannot endow \( \TBundle' Q \) with a natural smooth bundle structure for non-Banach manifolds \( Q \).

Hence, the most important class of examples of symplectic manifolds is a priori not available in infinite dimensions.
We now present a substitute which will play the role of the cotangent bundle.
\begin{defn}
	A \emphDef{dual pairing between two vector bundles} \( E \to Q \) and \( F \to Q \) is a smooth map \( h: E \times_Q F \to \R \) which is fiberwise non-degenerate, \ie, \( h_q: E_q \times F_q \to \R \) is a non-degenerate bilinear form for every \( q \in Q \).
	If \( E \) is a vector bundle in duality to the tangent bundle \( F = \TBundle Q \), then we will write \( \CotBundle Q \equiv E \) and call it a \emphDef{cotangent bundle} of $Q$.
	Correspondingly, we denote the bundle projection \( \CotBundle Q \to Q \) by \( \CotBundleProj \).
\end{defn}	
We will often denote the duality by brackets and write the dual pair as \( \dualPair{E}{F} \).
If \( Q \) carries a Riemannian metric \( g \), then the pairing \( g: \TBundle Q \times_Q \TBundle Q \to \R \) identifies the cotangent bundle of \( Q \) with \( \TBundle Q \).
In the following, we assume that a dual pair \( \dualPair{\CotBundle Q}{\TBundle Q} \) is fixed and we will simply refer to \( \CotBundle Q \) as \emph{the} cotangent bundle of $Q$.
The reader should however keep in mind that there is no smooth \emph{canonical} cotangent bundle and always a choice of a dual pair is required.

Recall that, for a dual pair \( \dualPair{X_2}{X_1} \) of vector spaces \parencite{Koethe1983}, one has a natural inclusion of \( X_2 \) into the topological dual \( X'_1 \) of \( X_1 \).
Similarly, for a dual pair of vector bundles, we obtain a natural vector bundle injection of \( E \) into the dual bundle \( F' \), whose fiber over \( m \) is the topological dual \( F'_m \) of \( F_m \).
In particular, every cotangent bundle \( \CotBundle Q \) comes with a natural injection into the topological cotangent bundle \( \TBundle' Q \).

\begin{prop}
\label{Cot-Bundle}
	For a cotangent bundle \( \CotBundle Q \), the formula
	\begin{equation}
		\theta_p (X) = \dualPair{p}{\tangent_p \CotBundleProj (X)}, \qquad X \in \TBundle_p (\CotBundle Q)
	\end{equation}
	defines a smooth \( 1 \)-form \( \theta \) on \( \CotBundle Q \).
	Furthermore, \( \omega = \dif \theta \) is a symplectic form.
	We say that \( \omega \) is the \emphDef{canonical symplectic form} on \( \CotBundle Q \).
\end{prop}
\begin{proof}
	Let \( U \subseteq Q \) be an open subset over which \( \CotBundleProj: \CotBundle Q \to Q \) as well as \( \tau: \TBundle Q \to Q \) trivialize.
	Denote the fiber model space of \( \TBundle Q \) and of \( \CotBundle Q \) by \( X \) and \( X^* \), respectively.
	Using a chart on \( Q \) we identify \( U \) as a subspace of \( X \).
	Then, the local chart representation of the pairing is a smooth map \( U \times X^* \times X \to \R \).
	In this chart, the canonical form \( \theta: (U \times X^*) \times (X \times X^*) \to \R \) is given by
	\begin{equation}
		\theta_{q, \alpha}(u, \beta) = \dualPair{\alpha}{u}_q
	\end{equation}
	and, hence, \( \theta \) is a smooth \( 1 \)-form on \( \CotBundle Q \).
	For the symplectic structure 
	\begin{equation}
		\omega:  (U \times X^*) \times {(X \times X^*)}^2 \to \R 
	\end{equation}
	we find\footnote{This local expression for \( \omega \) is well-known for the case that \( \TBundle Q \) is put in duality with itself using a Riemannian metric on \( Q \), see for example \parencite[p.~591]{Marsden1972}.} 
	\begin{equation}
		\omega_{q,p}(u_1, \beta_1, u_2, \beta_2)
			= \difp_q \left(\dualPair{p}{u_2}_q\right) (u_1) - 
			\difp_q \left(\dualPair{p}{u_1}_q\right) (u_2) + \dualPair{\beta_1}{u_2}_q - 
			\dualPair{\beta_2}{u_1}_q.
	\end{equation}
	In finite dimensions, this corresponds to the classical formula \( \omega = g\indices{^i_j} \dif p_i \wedge \dif q^j + \difp_k \, (g\indices{^i_j} \, p_i) \dif q^k \wedge \dif q^j \), where \( g\indices{^i_j}(q) \) denote the components of the dual pair \( \dualPairDot_q \).
	In either case, we conclude that \( \omega \) is non-degenerate, because the dual pair \( \dualPairDot_q \) possesses this property for every \( q \in Q \).
\end{proof}

\subsection{Momentum maps in infinite dimensions}
Let \( G \) be a Lie group acting smoothly on \( Q \).
By linearization, we get a smooth action of \( G \) on the tangent bundle, which we write using the lower dot notation as \( g \ldot Y \in \TBundle_{g \cdot q} Q \) for \( g \in G \) and \( Y \in \TBundle_q Q \).
The action on \( \TBundle Q \) induces a \( G \)-action on \( \CotBundle Q \) by requiring that the pairing be left invariant, that is,
\begin{equation}
 	\dualPair{g \cdot p}{Y}_q 
 		= \dualPair{p}{g^{-1} \ldot Y}_{g^{-1} \cdot q}, 
 		\qquad p \in \CotBundle_{g^{-1} \cdot q} Q, Y \in \TBundle_q Q. 
\end{equation} 
With respect to this action, the cotangent bundle projection \( \CotBundleProj \) is \( G \)-equivariant and so the action on \( \CotBundle Q \) is a lift of the action on \( Q \). 

Let $\LieA{g}$ be the Lie algebra of $G$.
Similarly to the above strategy for the cotangent bundle, the choice of a Fréchet space \( \LieA{g}^* \) and of a separately continuous non-degenerate bilinear form \( \kappa: \LieA{g}^* \times \LieA{g} \to \R \) yields a dual pair, with $\LieA{g}^*$ serving as the dual space of $\LieA{g}$.
The following is the infinite-dimensional version of the well-known result in finite dimensions that every lifted action on the cotangent bundle has a momentum map, see \eg \parencite[Example~4.5.4]{OrtegaRatiu2003}.
\begin{prop}
	\label{prop:contangentBundle:existenceMomentumMap}
	Let \( G \) be a Lie group that acts smoothly on \( Q \).
	Then, the lifted action of \( G \) on the cotangent bundle \( \CotBundle Q \) preserves the canonical symplectic form \( \omega \).
	Let, moreover, \( \kappa(\LieA{g}^*, \LieA{g}) \) be a dual pair.
	If, for every \( p \in \CotBundle Q \), the functional \( \xi \mapsto \theta_p (\xi \ldot p) \) on \( \LieA{g} \) is represented by an element \( J(p) \in \LieA{g}^* \) with respect to \( \kappa \), then the resulting map \( J: \CotBundle Q \to \LieA{g}^* \) is a smooth \( G \)-equivariant \( \LieA{g}^* \)-valued momentum map for the lifted \( G \)-action on \( \CotBundle Q \).
\end{prop}
\begin{proof}
	The canonical form \( \theta \) is \( G \)-invariant, because
	\begin{equation}\begin{split}
		\theta_{g \cdot p}(g \ldot X)
			= \dualPair{g \cdot p}{\tangent_{g \cdot p} \CotBundleProj (g \ldot X)}
			= \dualPair{g \cdot p}{g \ldot \tangent_{p} \CotBundleProj (X)}
			= \dualPair{p}{\tangent_{p} \CotBundleProj (X)}.
	\end{split}\end{equation}
	In particular, the action also leaves the symplectic form \( \omega = \dif \theta \) invariant.
	By assumption, the functional \( \xi \mapsto \theta_p (\xi \ldot p) \) on \( \LieA{g} \) is represented by an element \( J(p) \in \LieA{g}^* \) with respect to the given dual pair \( \kappa:\LieA{g}^* \times \LieA{g} \to \R \).
	This is to say,
	\begin{equation} \label{eq:cotbundle:momentumMapDef}
		\kappa(J(p), \xi) 
			= \theta_p (\xi \ldot p)
			= \dualPair{p}{\xi \ldot \CotBundleProj(p)}.
	\end{equation}
	On the other hand, we have \( \difLie_{\xi^*} \theta = 0 \), because \( \theta \) is \( G \)-invariant and so
	\begin{equation}
		0 
			= \difLie_{\xi^*} \theta 
			= \xi^* \contr \dif \theta + \dif (\xi^* \contr \theta)
			= \xi^* \contr \omega + \kappa(\dif J, \xi),
	\end{equation}
	which shows that \( J \) is a \( \LieA{g}^* \)-valued momentum map.
	The \( G \)-equivariance property of \( J \) is a direct consequence of the \( G \)-invariance of the pairing \( \dualPair{\CotBundle Q}{\TBundle Q} \) and~\eqref{eq:cotbundle:momentumMapDef}.
\end{proof}
In contrast to the finite-dimensional case, lifted actions of infinite-dimensional Lie groups do not need to possess a momentum map, see \cref{ex:contangentBundle:actionWithoutMomentumMap}.

We next state two results concerning the infinitesimal momentum map geometry.
These are well known for finite-dimensional symplectic manifolds, see \eg \parencite[Proposition~4.5.12]{OrtegaRatiu2003}.
A proof in the infinite-dimensional setting is given in \parencite{DiezThesis}.
\begin{lemma}[Bifurcation lemma]
	\label{prop:bifurcationLemma}
	Let \( (M, \omega) \) be a symplectic \( G \)-manifold with equivariant momentum map \( J: M \to \LieA{g}^* \).
	Then, the following holds:
	\begin{thmenumerate}
		\item (weak version)
			\label{prop::bifurcationLemma:weak}
			\begin{align}
				\ker \tangent_m J &= (\LieA{g} \cdot m)^\omega, \\
				\LieA{g}_m &= (\img \tangent_m J)^\perp.
			\end{align}
		\item (strong version)
			\label{prop::bifurcationLemma:strong}
			If, moreover, \( \LieA{g} \cdot m = (\LieA{g} \cdot m)^{\omega\omega} \), then in addition
			\begin{equation} \label{eq:bifurcationLemma:strong:ker}
				(\ker \tangent_m J)^\omega = \LieA{g} \cdot m
			\end{equation}
			holds and we have
			\begin{equation}
				\ker \tangent_m J \intersect (\ker \tangent_m J)^\omega = \LieA{g}_\mu \cdot m,
			\end{equation}
			where \( \mu = J(m) \in \LieA{g}^* \).
			\qedhere
	\end{thmenumerate}
\end{lemma}
We say that the orbit is \emphDef{symplectically closed} if the condition \( \LieA{g} \cdot m = (\LieA{g} \cdot m)^{\omega\omega} \) is satisfied.
\begin{prop}
	\label{prop:bifuractionLemma:symplecticSubspace}
	Let \( (M, \omega) \) be a symplectic \( G \)-manifold with equivariant momentum map \( J: M \to \LieA{g}* \).
	Assume that the orbit \( \LieA{g} \cdot m \) through \( m \in M \) is symplectically closed.
	Then, for every complement \( X \) of \( \LieA{g} \ldot m \) in \( \TBundle_m M \),
	\begin{equation}
		V = X \intersect \ker \tangent_m J
	\end{equation}
	is a symplectic subspace of \( (\TBundle_m M, \omega_m) \).
\end{prop}

\subsection{Slices and stratifications}
\label{sec:calculus:groupActionsSlices}

Since slices play a fundamental role in the construction of the normal form of the lifted action on the cotangent bundle, for the convenience of the reader, we now recall the definition of a slice in infinite dimensions, \cf \parencite{DiezSlice}.   
\begin{defn}
	\label{defn:slice:slice}
	Let \( M \) be a \( G \)-manifold.
	A \emphDef{slice} at \( m \in M \) is a submanifold \( S \subseteq M \) containing \( m \) with the following properties:
	\begin{thmenumerate}[label=(SL\arabic*), ref=(SL\arabic*), leftmargin=*] 
		\item \label{i::slice:SliceDefSliceInvariantUnderStab}
			The submanifold \( S \) is invariant under the induced action of the stabilizer subgroup \( G_m \), that is \( G_m \cdot S \subseteq S \).

		\item \label{i::slice:SliceDefOnlyStabNotMoveSlice}
			Any \( g \in G \) with \( (g \cdot S) \cap S \neq \emptyset \) is necessarily an element of \( G_m \). 

		\item \label{i::slice:SliceDefLocallyProduct}
			The stabilizer \( G_m \) is a principal Lie subgroup\footnote{A Lie subgroup \( H \subseteq G \) is called \emphDef{principal} if the natural fibration \( G \to G \slash H \) is a principal bundle, see \parencite[Section~7.1.4]{GloecknerNeeb2013}.} of \( G \) and the principal bundle \( G \to G \slash G_m \) admits a local section \( \chi: G \slash G_m \supseteq U \to G \) defined on an open neighborhood \( U \) of the identity coset \( \equivClass{e} \) in such a way that the map
			\begin{equation}
				\chi^S: U \times S \to M, \qquad (\equivClass{g}, s) \mapsto \chi(\equivClass{g}) \cdot s
			\end{equation}
			is a diffeomorphism onto an open neighborhood \( V \subseteq M \) of \( m \). We call \( V \) a \emphDef{slice neighborhood} of \( m \).

		\item \label{i::slice:SliceDefPartialSliceSubmanifold}
			The partial slice \( S_{(G_m)} = \set{s \in S \given G_s \text{ is conjugate to } G_m} \) is a closed submanifold of \( S \).

		\item
			\label{i:slice:linearSlice}
			There exist a continuous representation of \( G_m \) on a locally convex vector space \( X \) and a \( G_m \)-equivariant diffeomorphism \( \iota_S \) from a \( G_m \)-invariant open neighborhood of \( 0 \) in \( X \) onto \( S \) such that \( \iota_S(0) = m \).
			\qedhere
	\end{thmenumerate}
\end{defn}

In the finite-dimensional context, the existence of slices for proper actions is ensured by Palais' slice theorem \parencite{Palais1961}.
Passing to the infinite-dimensional case, this may no longer be true and additional hypotheses have to be made.
We refer the reader to \parencite{DiezSlice,Subramaniam1986} for general slice theorems in infinite dimensions and \parencite{AbbatiCirelliEtAl1989,Ebin1970,CerveraMascaroEtAl1991} for constructions of slices in concrete examples.
One of the problems one faces in the infinite-dimensional setting is the failure of the inverse function theorem and one usually is forced to use hard alternatives such as the Nash--Moser theorem.
However, for linear actions of compact groups the situation is better and we have the following existence result.
\begin{thm}[{\parencite[Theorem~3.15]{DiezSlice}}]
	\label{prop:slice:sliceTheoremLinearAction}
	Let \( X \) be a Fréchet space and let \( G \) be a compact Lie group that acts linearly and continuously on \( X \).
	Then, there exists a slice at every point of \( X \).
\end{thm}

In \cref{sec:yangMillsHiggs}, we need to construct a slice for a group action on a product manifold.
This situation is covered by the next result.
\begin{prop}[{\parencite[Proposition~3.29]{DiezSlice}}]
	\label{prop:slice:sliceForProduct}
	Let \( G \) be a Lie group that acts smoothly on the manifolds \( M \) and \( N \).
	Assume that the \( G \)-action admits a slice \( S_m \subseteq M \) at a given point \( m \in M \) and that the \( G_m \)-action on \( N \) admits a slice \( S_n \) at the point \( n \in N \).
	Then, \( S_m \times S_n \) is a slice at \( (m, n) \) for the diagonal action of \( G \) on the product \( M \times N \).
\end{prop}
As in the finite-dimensional case, the existence of slices implies many nice properties of the orbit space.
For example, we have the following.
\begin{prop}
	\label{prop:slice:orbitTypeSubsetIsSubmanifold}
	If the \( G \)-action admits a slice at every point of \( M \), then \( M_{(H)} \) is a submanifold of \( M \).
	Moreover, \( \check{M}_{(H)} = M_{(H)} \slash G \) carries a smooth manifold structure such that the natural projection \( \pi_{(H)}: M_{(H)} \to \check{M}_{(H)} \) is a smooth submersion.  
\end{prop}
The proof of this proposition is based on the following lemma whose proof carries over word by word from the finite-dimensional setting, compare \parencite[Lemma~4.2.9]{Pflaum2000} or \parencite[Lemma~2.1.14]{OrtegaRatiu2003}.
\begin{lemma} \label{prop::compactLieSubgroup:conjugatedSubgroupEqual}
	Let \( G \) be a Lie group. 
	Let \( H \) and \( K \) be two compact Lie subgroups of \( G \). 
	If \( K \) is conjugate to \( H \) and \( K \subseteq H \), then \( K = H \).
\end{lemma}
If, in addition, a certain approximation property is satisfied, then the orbit type manifolds fit together nicely and so the orbit space is a stratified space, see \parencite[Theorem~4.2]{DiezSlice}.
For completeness, we include here our definition of stratification and refer the reader to \parencite{DiezSlice} for further details and comparison with other notions of stratified spaces in the literature.
\begin{defn}
	\label{def:stratification:stratification}
	Let \( X \) be Hausdorff topological space. 
	A partition \( \stratification{Z} \) of \( X \) into subsets \( X_\sigma \) indexed by \( \sigma \in \Sigma \) is called a \emphDef{stratification} of \( X \) if the following conditions are satisfied:
	\begin{thmenumerate}[label=(DS\arabic*), ref=(DS\arabic*), leftmargin=*]
		\item \label{i::stratification:stratumIsManifold} 
			Every piece \( X_\sigma \) is a locally closed, smooth manifold (whose manifold topology coincides with the relative topology).
			We will call \( X_\sigma \) a \emphDef{stratum} of \( X \).
 
		\item \label{i::stratification:frontierCondition} (frontier condition)
			Every pair of disjoint strata \( X_\varsigma \) and \( X_\sigma \) with \( X_\varsigma \cap \closureSet{X_\sigma} \neq \emptyset \) satisfies:
			\begin{thmenumerate}[label=\alph*), ref=(DS2\alph*)]
				\item \label{i:stratification:frontierConditionBoundary}
					\( X_\varsigma \) is contained in the frontier \( \closureSet{X_\sigma} \setminus X_\sigma\) of \( X_\sigma \),
				\item \label{i:stratification:frontierConditionIntersection}
					\( X_\sigma \) does not intersect \( \closureSet{X_\varsigma} \).
			\end{thmenumerate}
			In this case, we write \( X_\varsigma < X_\sigma \) or \( \varsigma < \sigma \). 
			\qedhere
	\end{thmenumerate}
\end{defn}

\begin{refcontext}[sorting=nyt]{}
	\printbibliography
\end{refcontext}

\end{document}